\documentclass[11pt,english]{article}
\usepackage[T1]{fontenc}
\usepackage[latin9]{inputenc}
\usepackage[active]{srcltx}
\usepackage{color}
\usepackage{amsmath}
\usepackage{amsthm}
\usepackage{amssymb}

\makeatletter

\providecommand{\tabularnewline}{\\}

\numberwithin{equation}{section}
\numberwithin{figure}{section}
\theoremstyle{plain}
\newtheorem{thm}{\protect\theoremname}
  \theoremstyle{definition}
  \newtheorem{defn}[thm]{\protect\definitionname}
  \theoremstyle{plain}
  \newtheorem{lem}[thm]{\protect\lemmaname}
  \theoremstyle{definition}
  \newtheorem{example}[thm]{\protect\examplename}
  \theoremstyle{plain}
  \newtheorem{cor}[thm]{\protect\corollaryname}
  \theoremstyle{remark}
  \newtheorem*{rem*}{\protect\remarkname}
  \theoremstyle{plain}
  \newtheorem{fact}[thm]{\protect\factname}
  \theoremstyle{remark}
  \newtheorem*{acknowledgement*}{\protect\acknowledgementname}

\usepackage{hyperref}
\hypersetup{
    colorlinks,
    allcolors=red
}
\usepackage[lined,boxed,ruled,norelsize,algo2e]{algorithm2e}

\usepackage[margin=1in]
           {geometry}

\usepackage{thmtools}
\usepackage{thm-restate}
\usepackage{amsfonts}
\usepackage{tikz}

\makeatother

\usepackage{babel}
  \providecommand{\acknowledgementname}{Acknowledgement}
  \providecommand{\corollaryname}{Corollary}
  \providecommand{\definitionname}{Definition}
  \providecommand{\examplename}{Example}
  \providecommand{\factname}{Fact}
  \providecommand{\lemmaname}{Lemma}
  \providecommand{\remarkname}{Remark}
\providecommand{\theoremname}{Theorem}

\begin{document}
\global\long\def\defeq{\stackrel{\mathrm{{\scriptscriptstyle def}}}{=}}
\global\long\def\norm#1{\left\Vert #1\right\Vert }
\global\long\def\R{\mathbb{R}}
 \global\long\def\Rn{\mathbb{R}^{n}}
\global\long\def\tr{\mathrm{Tr}}
\global\long\def\diag{\mathrm{diag}}
\global\long\def\Diag{\mathrm{Diag}}
\global\long\def\C{\mathbb{C}}
 \global\long\def\E{\mathbb{E}}
\global\long\def\vol{\mathrm{vol}}
\global\long\def\argmax{\mathrm{argmax}}

\global\long\def\ham{\mathrm{Ham}}
\global\long\def\e#1{ \exp\left(#1\right)}
\global\long\def\Var{\mathrm{Var}}
\global\long\def\dint{{\displaystyle \int}}
\global\long\def\step{\delta}
\global\long\def\Ric{\mathrm{Ric}}
\global\long\def\P{\mathbb{P}}

\title{Convergence Rate of Riemannian Hamiltonian Monte Carlo and Faster
Polytope Volume Computation}

\author{Yin Tat Lee\thanks{University of Washington and Microsoft Research, yintat@uw.edu},
Santosh S. Vempala\thanks{Georgia Tech, vempala@gatech.edu}}

\maketitle
\begin{abstract}
We give the first rigorous proof of the convergence of Riemannian
Hamiltonian Monte Carlo, a general (and practical) method for sampling
Gibbs distributions. Our analysis shows that the rate of convergence
is bounded in terms of natural smoothness parameters of an associated
Riemannian manifold. We then apply the method with the manifold defined
by the log barrier function to the problems of (1) uniformly sampling
a polytope and (2) computing its volume, the latter by extending Gaussian
cooling to the manifold setting. In both cases, the total number of
steps needed is $O^{*}(mn^{\frac{2}{3}})$, improving
the state of the art. A key ingredient of our analysis is a proof
of an analog of the KLS conjecture for Gibbs distributions over manifolds. 
\end{abstract}
\tableofcontents{}\pagebreak{}
\section{Introduction}

Hamiltonian dynamics provide an elegant alternative to Newtonian mechanics.
The Hamiltonian $H$, which captures jointly the potential and kinetic
energy of a particle, is a function of its position and velocity.
First-order differential equations describe the change in both. 
\begin{align*}
\frac{dx}{dt} & =\frac{\partial H(x,v)}{\partial v},\\
\frac{dv}{dt} & =-\frac{\partial H(x,v)}{\partial x}.
\end{align*}
As we review in Section \ref{sec:Basics-of-Hamiltonian}, these equations
preserve the Hamiltonian $H$.

Riemannian Hamiltonian Monte Carlo (or RHMC) \cite{Neal1996,Neal2012}\cite{Girolami2009,Girolami2011}
is a Markov Chain Monte Carlo method for sampling from a desired distribution.
The target distribution is encoded in the definition of the Hamiltonian.
Each step of the method consists of the following: At a current point
$x$,
\begin{enumerate}
\item Pick a random velocity $y$ according to a local distribution defined
by $x$ (in the simplest setting, this is the standard Gaussian distribution
for every $x$).
\item Move along the Hamiltonian curve defined by Hamiltonian dynamics at
$(x,y)$ for time (distance) $\step$.
\end{enumerate}
For a suitable choice of $H$, the marginal distribution of the current
point $x$ approaches the desired target distribution. Conceptually,
the main advantage of RHMC is that it does not require a Metropolis
filter (as in the Metropolis-Hastings method) and its step sizes are
therefore not severely limited even in high dimension.

Over the past two decades, RHMC has become very popular in Statistics
and Machine Learning, being applied to Bayesian learning, to evaluate
expectations and likelihood of large models by sampling from the appropriate
Gibbs distribution, etc. It has been reported to significantly outperform
other known methods \cite{Betancourt2017,Neal2012} and much effort
has been made to make each step efficient by the use of numerical
methods for solving ODEs.

In spite of all these developments and the remarkable empirical popularity
of RHMC, analyzing its rate of convergence and thus rigorously explaining
its success has remained an open question. 

\subsection{Results}

In this paper, we analyze the mixing rate of Hamiltonian Monte Carlo
for a general function $f$ as a Gibbs sampler, i.e., to generate
samples from the density proportional to $e^{-f(x)}$. The corresponding
Hamiltonian is $H(x,v)=f(x)+\frac{1}{2}\log((2\pi)^{n}\det g(x))+\frac{1}{2}v^{T}g(x)^{-1}v$
for some metric $g$. We show that for $x$ in a compact manifold,
the conductance of the Markov chain is bounded in terms of a few parameters
of the metric $g$ and the function $f$. The parameters and resulting
bounds are given in Corollary \ref{cor:gen-convergence_simplified}
and Theorem \ref{thm:gen-convergence}. Roughly speaking, the guarantee
says that Hamiltonian Monte Carlo mixes in polynomial time for smooth
Hamiltonians. We note that these guarantees use only the smoothness
and Cheeger constant (expansion) of the function, without any convexity
type assumptions. Thus, they might provide insight in nonconvex settings
where (R)HMC is often applied.

We then focus on logconcave densities in $\R^{n}$, i.e., $f(x)$
is a convex function. This class of functions appears naturally in
many contexts and is known to be sampleable in polynomial-time given
access to a function value oracle. For logconcave densities, the current
fastest sampling algorithms use $n^{4}$ function calls, even for
uniform sampling \cite{LV06,LV07}, and $n^{2.5}$ oracle calls given
a warm start after appropriate rounding (linear transformation) \cite{LeeV16KLS}.
In the prototypical setting of uniform sampling from a polytope $Ax\ge b$,
with $m$ inequalities, the general complexity is no better, with
each function evaluation taking $O(mn)$ arithmetic operations, for
an overall complexity of $n^{4}\cdot mn=mn^{5}$ in the worst case
and $n^{2.5}\cdot mn$ after rounding from a warm start. The work
of Kannan and Narayanan \cite{KanNar2009} gives an algorithm of complexity
$mn^{2}\cdot mn^{\omega-1}$ from an arbitrary start and $mn\cdot mn^{\omega-1}$
from a warm start (here $\omega$ is the matrix multiplication exponent),
which is better than the general case when the number of facets $m$
is not too large. This was recently improved to $mn^{0.75}\cdot mn^{\omega-1}$
from a warm start \cite{LeeV16}; the subquadratic complexity for
the number of steps is significant since all known general oracle
methods cannot below a quadratic number of steps. The leading algorithms
and their guarantees are summarized in Table \ref{tab:sampling}.
\begin{table}[h]
\centering{}%
\begin{tabular}{|c|l|c|c|}
\hline 
Year  & Algorithm & Steps & Cost per step\tabularnewline
\hline 
\hline 
1997 \cite{KLS97} & Ball walk\textsuperscript{\#} & $n^{3}$ & $mn$\tabularnewline
\hline 
2003 \cite{LV3} & Hit-and-run\textsuperscript{\#} & $n^{3}$ & $mn$\tabularnewline
\hline 
2009 \cite{KanNar2009} & Dikin walk & $mn$ & $mn^{\omega-1}$\tabularnewline
\hline 
2016 \cite{LeeV16} & Geodesic walk & $mn^{\frac{3}{4}}$ & $mn^{\omega-1}$\tabularnewline
\hline 
2016 \cite{LeeV17KLS} & Ball walk\textsuperscript{\#} & $n^{2.5}$ & $mn$\tabularnewline
\hline 
\textcolor{red}{This paper} & \textcolor{red}{RHMC} & \textcolor{red}{$mn^{\frac{2}{3}}$} & \textcolor{red}{$mn^{\omega-1}$}\tabularnewline
\hline 
\end{tabular}\caption{\label{tab:sampling}The complexity of uniform polytope sampling from
a warm start, where each step of every algorithm uses $\widetilde{O}(n)$
bit of randomness. The entries marked \protect\textsuperscript{\#}
are for general convex bodies presented by oracles, while the rest
are for polytopes.}
\end{table}

In this paper, using RHMC, we improve the complexity of sampling polytopes.
In fact we do this for a general family of Gibbs distributions, of
the form $e^{-\alpha\phi(x)}$ where $\phi(x)$ is a convex function
over a polytope. When $\phi(x)$ is the standard logarithmic barrier
function and $g(x)$ is its Hessian, we get a sampling method that
mixes in only $n^{\frac{1}{6}}m^{\frac{1}{2}}+\frac{n^{\frac{1}{5}}m^{\frac{2}{5}}}{\alpha^{\frac{1}{5}}+m^{-\frac{1}{5}}}+\frac{n^{\frac{2}{3}}}{\alpha+m^{-1}}$
steps from a warm start! When $\alpha=1/m$, the resulting distribution
is very close to uniform over the polytope. 

\begin{restatable}{thm}{logbarrier}

\label{thm:logbarrier}Let $\phi$ be the logarithmic barrier for
a polytope $M$ with $m$ constraints and $n$ variables. Hamiltonian
Monte Carlo applied to the function $f=\exp(-\alpha\phi(x))$ and
the metric given by $\nabla^{2}\phi$ with appropriate step size mixes
in 
\[
\widetilde{O}\left(\frac{n^{\frac{2}{3}}}{\alpha+m^{-1}}+\frac{m^{\frac{1}{3}}n^{\frac{1}{3}}}{\alpha^{\frac{1}{3}}+m^{-\frac{1}{3}}}+m^{\frac{1}{2}}n^{\frac{1}{6}}\right)
\]
steps where each step is the solution of a Hamiltonian ODE.

\end{restatable}

In recent independent work, Mangoubi and Smith \cite{mangoubi2017rapid}
analyze Euclidean HMC in the oracle setting, i.e., assuming an oracle
for evaluating $\phi$. Their analysis formally gives a dimension-independent
convergence rate based on certain regularity assumptions such as strong convexity and smoothness of the Hamiltonian $H$.
Unfortunately, these assumptions do not hold for the polytope sampling problem.

An important application of sampling is integration. The complexity
of integration for general logconcave functions is also $n^{4}$ oracle
calls. For polytopes, the most natural questions is computing its
volume. For this problem, the current best complexity is $n^{4}\cdot mn$,
where the factor of $O(mn)$ is the complexity of checking membership
in a polytope. Thus, even for explicitly specified polytopes, the
complexity of estimating the volume from previous work is asymptotically
the same as that for a general convex body given by a membership oracle.
Here we obtain a volume algorithm with complexity $mn^{\frac{2}{3}}\cdot mn^{\omega-1}$,
improving substantially on previous algorithms. The volume algorithm
is based using Hamiltonian Monte Carlo for sampling from a sequence
of Gibbs distributions over polytopes. We remark that in the case
when $m=O(n)$\footnote{We suspect that the LS barrier \cite{lee2014path} might be used to
get a faster algorithm even in the regime even if $m$ is sub-exponential.
However, our proof requires a delicate estimate of the fourth derivative
of the barrier functions. Therefore, such a result either requires
a new proof or a unpleasantly long version of the current proof.}, the final complexity is $o(n^{4})$ arithmetic operations, improving
by more than a \emph{quadratic factor} in the dimension over the previous
best complexity of $\tilde{O}(n^{6})$ operations for arbitrary polytopes.
These results and prior developments are given in Table \ref{tab:volume}. 

\begin{table}[h]
\centering{}%
\begin{tabular}{|c|l|c|c|}
\hline 
Year  & Algorithm & Steps & Cost per step\tabularnewline
\hline 
\hline 
1989 \cite{DyerFK89} & DFK & $n^{23}$ & $mn$\tabularnewline
\hline 
1989-93 \cite{LS90,DyerF90,ApplegateK91,LS92,LS93} & many improvements & $n^{7}$ & $mn$\tabularnewline
\hline 
1997 \cite{KLS97} & DFK, Speedy walk, isotropy & $n^{5}$ & $mn$\tabularnewline
\hline 
2003 \cite{LV2} & Annealing, hit-and-run & $n^{4}$ & $mn$\tabularnewline
\hline 
2015 \cite{CV2015} & Gaussian Cooling\textsuperscript{{*}} & $n^{3}$ & $mn$\tabularnewline
\hline 
\textcolor{red}{This paper} & \textcolor{red}{RHMC + Gaussian Cooling} & \textcolor{red}{$mn^{\frac{2}{3}}$} & \textcolor{red}{$mn^{\omega-1}$}\tabularnewline
\hline 
\end{tabular}\caption{\label{tab:volume} The complexity of volume estimation, each step
uses $\widetilde{O}(n)$ bit of randomness, all except the last for
general convex bodies (the result marked \protect\textsuperscript{{*}}
is for well-rounded convex bodies). The current paper applies to general
polytopes, and is the first improvement utilizing their structure. }
\end{table}

\begin{restatable}{thm}{volume}
\label{thm:volume}For any polytope $P=\left\{ x:\,Ax\ge b\right\} $
with $m$ constraints and $n$ variables, and any $\varepsilon>0$,
the Hamiltonian volume algorithm estimates the volume of $P$ to within
$1\pm\varepsilon$ multiplicative factor using $\widetilde{O}\left(mn^{\frac{2}{3}}\varepsilon^{-2}\right)$
steps where each step consists of solving a first-order ODE and takes
time $\widetilde{O}\left(mn^{\omega-1}L^{O(1)}\log^{O(1)}\frac{1}{\varepsilon}\right)$
and $L$ is the bit complexity\footnote{$L=\log(m+d_{\max}+\norm b_{\infty})$ where $d_{\max}$ is the largest
absolute value of the determinant of a square sub-matrix of $A$.} of the polytope.
\end{restatable}

A key ingredient in the analysis of RHMC is a new isoperimetric inequality
for Gibbs distributions over manifolds. This inequality can be seen
as an evidence of a manifold version of the KLS hyperplane conjecture.
For the family of Gibbs distributions induced by convex functions
with convex Hessians, the expansion is within a constant factor of
that of a hyperplane cut. This result might be of independent interest.

\subsection{Approach and contributions}

Traditional methods to sample from distributions in $\R^{n}$ are
based on random walks that take straight line steps (grid walk, ball
walk, hit-and-run). While this leads to polynomial-time convergence
for logconcave distributions, the length of each step has to be small
due to boundary effects, and a Metropolis filter (rejection sampling)
has to be applied to ensure the limiting distribution is the desired
one. These walks cannot afford a step of length greater than $\delta=O\left(\frac{1}{\sqrt{n}}\right)$
for a distribution in isotropic position, and take a quadratic number
of steps even for the hypercube. The Dikin walk for polytopes \cite{KanNar2009},
which explicitly takes into account the boundary of polytope at each
step, has a varying step size, but still runs into similar issues
and the bound on its convergence rate is $O(mn)$ for a polytope with
$m$ facets. 

In a recent paper \cite{LeeV16}, we introduced the geodesic walk.
Rather than using straight lines in Euclidean space, each step of
the walk is along a geodesic (locally shortest path) of a Riemannian
metric. More precisely, each step first makes a deterministic move
depending on the current point (drift), then moves along a geodesic
in a random initial direction and finally uses a Metropolis filter.
Each step can be computed by solving a first-order ODE. Due to the
combination of drift and geodesic, the local $1$-step distributions
are smoother than that of the Dikin walk and larger steps can be taken
while keeping a bounded rejection probability for the filter. For
sampling polytopes, the manifold/metric defined by the standard log
barrier gives a convergence rate of $mn^{\frac{3}{4}}$, going below
the quadratic (or higher) bound of all previous sampling methods. 

One major difficulty with geodesic walk is ensuring the stationary
distribution is uniform. For high dimensional problems, this necessitates
taking a sufficiently small step size and then rejecting some samples
according to the desired transition probabilities according to Metropolis
filter. Unfortunately, computing these transition probabilities can
be very expensive. For the geodesic walk, it entails solving an $n\times n$
size matrix ODE.

Hamiltonian Monte Carlo bears some similarity to the geodesic walk
\textemdash{} each step is a random (non-linear) curve. But the Hamiltonian-preserving
nature of the process obviates the most expensive ingredient, Metropolis
filter. Due to this, the step size can be made longer, and as a result
we obtain a faster sampling algorithm for polytopes that mixes in
$mn^{\frac{2}{3}}$ steps (the per-step complexity remains essentially
the same, needing the solution of an ODE). 

To get a faster algorithm for volume computation, we extends the analysis
to a general family of Gibbs distributions, including $f(x)=e^{-\alpha\phi(x)}$
where $\phi(x)$ is the standard log-barrier and $\alpha>0$. We show
that the smoothness we need for the sampling corresponding to a variant
of self-concordance defined in Definition \ref{def:self_concordance}.
Furthermore, we establish an isoperimetric inequality for this class
of functions. This can be viewed as an extension of the KLS hyperplane
conjecture from Euclidean to Riemannian metrics (the analogous case
in Euclidean space to what we prove here is the isoperimetry of the
Gaussian density function multiplied by any logconcave function, a
case for which the KLS conjecture holds). The mixing rate for this
family of functions is $sublinear$ for $\alpha=\Omega(1)$. 

Finally, we study the Gaussian Cooling schedule of \cite{CV2015}.
We show that in the manifold setting, the Gaussian distribution $e^{-\norm x^{2}/2}$
can be replaced by $e^{-\alpha\phi(x)}$. Moreover, the speed of Gaussian
Cooling depends on the ``thin-shell'' constant of the manifold and
classical self-concordance of $\phi$. 

Combining all of these ideas, we obtain a faster algorithm for polytope
volume computation. The resulting complexity of polytope volume computation
is the same as that of sampling uniformly from a warm start: $mn^{\frac{2}{3}}$
steps. To illustrate the improvement, for polytopes with $m=O(n)$
facets, the new bound is $n^{\frac{5}{3}}$ while the previous best
bound was $n^{4}$.

\subsection{Practicality }

From the experiments, the ball walk/hit-and-run seem to mix in $n^{2}$
steps, the geodesic walk seems to mix in sublinear number of steps
(due to the Metropolis filter bottleneck) and RHMC seems to mix in
only polylogarithmic number of steps. One advantage of RHMC compared
to the geodesic walk is that it does not require the expensive Metropolis
filter that involves solving $n\times n$ matrix ODEs. In the future,
we plan to do an empirical comparison study of different sampling
algorithms. We are hopeful that using RHMC we might finally be able
to sample from polytopes in millions of dimensions after more than
three decades of research on this topic!

\subsection{Notation }

Throughout the paper, we use lowercase letter for vectors and vector
fields and uppercase letter for matrices and tensors. We use $e_{k}$
to denote coordinate vectors. We use $\frac{d}{dt}$ for the usual
derivative, e.g. $\frac{df(c(t))}{dt}$ is the derivative of some
function $f$ along a curve $c$ parametrized by $t$, we use $\frac{\partial}{\partial v}$
for the usual partial derivative. We use $D^{k}f(x)[v_{1},v_{2},\cdots,v_{k}]$
for the $k^{th}$ directional derivative of $f$ at $x$ along $v_{1},v_{2},\cdots,v_{k}$.
We use $\nabla$ for the usual gradient and the connection (manifold
derivative, defined in Section \ref{sec:RG} which takes into account
the local metric), $D_{v}$ for the directional derivative of a vector
with respect to the vector (or vector field) $v$ (again, defined
in Section \ref{sec:RG}), and $D_{t}$ if the curve $v(t)$ is clear
from the context. We use $g$ for the local metric. Given a point
$x\in M$, $g$ is a matrix with entries $g_{ij}.$ Its inverse has
entries $g^{ij}.$ Also, $n$ is the dimension, $m$ the number of
inequalities. We use $d_{TV}$ for the total variation (or $L_{1}$)
distance between two distributions.

\subsection{Organization}

In Section \ref{sec:Basics-of-Hamiltonian}, we define the Riemannian
Hamiltonian Monte Carlo and study its basic properties such as time-reversibility.
In Section \ref{sec:Convergence}, we give the the first convergence
rate analysis of RHMC. However, the convergence rate is weak for the
sampling applications (it is polynomial, but not better than previous
methods). In Section \ref{sec:Convergence_imporved}, we introduce
more parameters and use them to get a tighter analysis of RHMC. In
Section \ref{sec:KLS}, we study the isoperimetric constant of $f(x)=e^{-\alpha\phi(x)}$
under the metric $\nabla^{2}\phi(x)$. In Section \ref{sec:cooling},
we study the generalized Gaussian Cooling schedule and its relation
to the thin-shell constant. Finally, in Section \ref{sec:Logarithmic},
we compute the parameters we need for the log barrier function.
\pagebreak{}
\section{\label{sec:Basics-of-Hamiltonian}Basics of Hamiltonian Monte Carlo}

In this section, we define the Hamiltonian Monte Carlo method for
sampling from a general distribution $e^{-H(x,y)}$. Hamiltonian Monte
Carlo uses curves instead of straight lines and this makes the walk
time-reversible even if the target distribution is not uniform, with
no need for a rejection sampling step. In contrast, classical approaches
such as the ball walk require an explicit rejection step to converge
to a desired stationary distribution.
\begin{defn}
Given a continuous, twice-differentiable function $H:\mathcal{M}\times\Rn\subset\R^{n}\times\R^{n}\rightarrow\R$
(called the \emph{Hamiltonian}, which often corresponds to the total
energy of a system) where $\mathcal{M}$ is the $x$ domain of $H$,
we say $(x(t),y(t))$ follows a \emph{Hamiltonian curve} if it satisfies
the \emph{Hamiltonian equations}

\begin{align}
\frac{dx}{dt} & =\frac{\partial H(x,y)}{\partial y},\nonumber \\
\frac{dy}{dt} & =-\frac{\partial H(x,y)}{\partial x}.\label{eq:ham}
\end{align}
We define the map $T_{\step}(x,y)\defeq(x(\step),y(\step))$ where
the $(x(t),y(t))$ follows the Hamiltonian curve with the initial
condition $(x(0),y(0))=(x,y).$
\end{defn}

Hamiltonian Monte Carlo is the result of a sequence of randomly generated
Hamiltonian curves.

\begin{algorithm2e}

\caption{Hamiltonian Monte Carlo}

\SetAlgoLined

\textbf{Input:} some initial point $x^{(1)}\in\mathcal{M}$.

\For{$i=1,2,\cdots,T$}{

Sample $y^{(k+\frac{1}{2})}$ according to $e^{-H(x^{(k)},y)}/\pi(x^{(k)})$
where $\pi(x)=\int_{\Rn}e^{-H(x,y)}dy$.

With probability $\frac{1}{2}$, set $(x^{(k+1)},y^{(k+1)})=T_{\step}(x^{(k)},y^{(k+\frac{1}{2})})$. 

Otherwise, $(x^{(k+1)},y^{(k+1)})=T_{-\step}(x^{(k)},y^{(k+\frac{1}{2})})$.

}

\textbf{Output:} $(x^{(T+1)},y^{(T+1)})$.

\end{algorithm2e}
\begin{lem}[Energy Conservation]
\label{lem:ham_eng_pre}For any Hamiltonian curve $(x(t),y(t))$,
we have that
\[
\frac{d}{dt}H(x(t),y(t))=0.
\]
\end{lem}

\begin{proof}
Note that
\[
\frac{d}{dt}H(x(t),y(t))=\frac{\partial H}{\partial x}\frac{dx}{dt}+\frac{\partial H}{\partial y}\frac{dy}{dt}=\frac{\partial H}{\partial x}\frac{\partial H}{\partial y}-\frac{\partial H}{\partial y}\frac{\partial H}{\partial x}=0.
\]
\end{proof}
\begin{lem}[Measure Preservation]
\label{lem:ham_measure_pre}For any $t\geq0$, we have that
\[
\det\left(DT_{t}(x,y)\right)=1
\]
where $DT_{t}(x,y)$ is the Jacobian of the map $T_{t}$ at the point
$(x,y)$.
\end{lem}

\begin{proof}
Let $(x(t,s),y(t,s))$ be a family of Hamiltonian curves given by
$T_{t}(x+sd_{x},y+sd_{y})$. We write 
\[
u(t)=\frac{\partial}{\partial s}x(t,s)|_{s=0}\,,\,v(t)=\frac{\partial}{\partial s}y(t,s)|_{s=0}.
\]
By differentiating the Hamiltonian equations (\ref{eq:ham}) w.r.t.
$s$, we have that
\begin{align*}
\frac{du}{dt} & =\frac{\partial^{2}H(x,y)}{\partial y\partial x}u+\frac{\partial^{2}H(x,y)}{\partial y\partial y}v,\\
\frac{dv}{dt} & =-\frac{\partial^{2}H(x,y)}{\partial x\partial x}u-\frac{\partial^{2}H(x,y)}{\partial x\partial y}v,\\
(u(0),v(0)) & =(d_{x},d_{y}).
\end{align*}
This can be captured by the following matrix ODE
\begin{align*}
\frac{d\Phi}{dt} & =\left(\begin{array}{cc}
\frac{\partial^{2}H(x(t),y(t))}{\partial y\partial x} & \frac{\partial^{2}H(x(t),y(t))}{\partial y\partial y}\\
-\frac{\partial^{2}H(x(t),y(t))}{\partial x\partial x} & -\frac{\partial^{2}H(x(t),y(t))}{\partial x\partial y}
\end{array}\right)\Phi(t)\\
\Phi(0) & =I
\end{align*}
using the equation
\[
DT_{t}(x,y)\left(\begin{array}{c}
d_{x}\\
d_{y}
\end{array}\right)=\left(\begin{array}{c}
u(t)\\
v(t)
\end{array}\right)=\Phi(t)\left(\begin{array}{c}
d_{x}\\
d_{y}
\end{array}\right).
\]
Therefore, $DT_{t}(x,y)=\Phi(t)$. Next, we observe that
\[
\frac{d}{dt}\log\det\Phi(t)=\tr\left(\Phi(t)^{-1}\frac{d}{dt}\Phi(t)\right)=\tr\left(\begin{array}{cc}
\frac{\partial^{2}H(x(t),y(t))}{\partial y\partial x} & \frac{\partial^{2}H(x(t),y(t))}{\partial y\partial y}\\
-\frac{\partial^{2}H(x(t),y(t))}{\partial x\partial x} & -\frac{\partial^{2}H(x(t),y(t))}{\partial x\partial y}
\end{array}\right)=0.
\]
Hence,
\[
\det\Phi(t)=\det\Phi(0)=1.
\]
\end{proof}
Using the previous two lemmas, we now show that Hamiltonian Monte
Carlo indeed converges to the desired distribution. 
\begin{lem}[Time reversibility]
\label{lem:time_reversibility}Let $p_{x}(x')$ denote the probability
density of one step of the Hamiltonian Monte Carlo starting at $x$.
We have that\textup{
\[
\pi(x)p_{x}(x')=\pi(x')p_{x'}(x)
\]
for almost everywhere in $x$ and $x'$ where $\pi(x)=\int_{\Rn}e^{-H(x,y)}dy$.}
\end{lem}

\begin{proof}
Fix $x$ and $x'$. Let $F_{\step}^{x}(y)$ be the $x$ component
of $T_{\step}(x,y)$. Let $V_{+}=\{y:\text{ }F_{\step}^{x}(y)=x'\}$
and $V_{-}=\{y:\text{ }F_{-\step}^{x}(x)=x')\}$. Then, 
\[
\pi(x)p_{x}(x')=\frac{1}{2}\int_{y\in V_{+}}\frac{e^{-H(x,y)}}{\left|\det\left(DF_{\step}^{x}(y)\right)\right|}+\frac{1}{2}\int_{y\in V_{-}}\frac{e^{-H(x,y)}}{\left|\det\left(DF_{-\step}^{x}(y)\right)\right|}.
\]
We note that this formula assumed that $DF_{\step}^{x}$ is invertible.
Sard's theorem showed that $F_{\step}^{x}(N)$ is measure $0$ where
$N\defeq\{y:\,DF_{s}^{x}(y)\text{ is not invertible}\}$. Therefore,
the formula is correct except for a measure zero subset.

By reversing time for the Hamiltonian curve, we have that for the
same $V_{\pm}$, 
\begin{equation}
\pi(x')p_{x'}(x)=\frac{1}{2}\int_{y\in V_{+}}\frac{e^{-H(x',y')}}{\left|\det\left(DF_{-\step}^{x'}(y')\right)\right|}+\frac{1}{2}\int_{y\in V_{-}}\frac{e^{-H(x',y')}}{\left|\det\left(DF_{\step}^{x'}(y')\right)\right|}\label{eq:ham_transit}
\end{equation}
where $y'$ denotes the $y$ component of $T_{\step}(x,y)$ and $T_{-\step}(x,y)$
in the first and second sum respectively.

We compare the first terms in both equations. Let $DT_{\step}(x,y)=\left(\begin{array}{cc}
A & B\\
C & D
\end{array}\right)$. Since $T_{\step}\circ T_{-\step}=I$ and $T_{\step}(x,y)=(x',y')$,
the inverse function theorem shows that $DT_{-\step}(x',y')$ is the
inverse map of $DT_{\step}(x,y)$. Hence, we have that
\[
DT_{-\step}(x',y')=\left(\begin{array}{cc}
A & B\\
C & D
\end{array}\right)^{-1}=\left(\begin{array}{cc}
\cdots & -A^{-1}B(D-CA^{-1}B)^{-1}\\
\cdots & \cdots
\end{array}\right).
\]
Therefore, we have that $F_{\step}^{x}(y)=B$ and $F_{-\step}^{x'}(y')=-A^{-1}B(D-CA^{-1}B)^{-1}$.
Hence, we have that
\begin{align*}
\left|\det\left(DF_{-\step}^{x'}(y')\right)\right| & =\left|\det A^{-1}\det B\det\left(D-CA^{-1}B\right)^{-1}\right|=\frac{\left|\det B\right|}{\left|\det\left(\begin{array}{cc}
A & B\\
C & D
\end{array}\right)\right|}.
\end{align*}
Using that $\det\left(DT_{t}(x,y)\right)=\det\left(\begin{array}{cc}
A & B\\
C & D
\end{array}\right)=1$ (Lemma \ref{lem:ham_measure_pre}), we have that
\[
\left|\det\left(DF_{-\step}^{x'}(y')\right)\right|=\left|\det\left(DF_{\step}^{x}(y)\right)\right|.
\]
Hence, we have that
\begin{align*}
\frac{1}{2}\int_{y\in V_{+}}\frac{e^{-H(x,y)}}{\left|\det\left(DF_{\step}^{x}(y)\right)\right|} & =\frac{1}{2}\int_{y\in V_{+}}\frac{e^{-H(x,y)}}{\left|\det\left(DF_{-\step}^{x'}(y')\right)\right|}\\
 & =\frac{1}{2}\int_{y\in V_{+}}\frac{e^{-H(x',y')}}{\left|\det\left(DF_{-\step}^{x'}(y')\right)\right|}
\end{align*}
where we used that $e^{-H(x,y)}=e^{-H(x',y')}$ (Lemma \ref{lem:ham_eng_pre})
at the end. 

For the second term in (\ref{eq:ham_transit}), by the same calculation,
we have that
\[
\frac{1}{2}\int_{y\in V_{-}}\frac{e^{-H(x,y)}}{\left|\det\left(DF_{-\step}^{x}(y)\right)\right|}=\frac{1}{2}\int_{y\in V_{+}}\frac{e^{-H(x',y')}}{\left|\det\left(DF_{\step}^{x'}(y')\right)\right|}
\]

Combining both terms we have the result.
\end{proof}
The main challenge in analyzing Hamiltonian Monte Carlo is to bound
its mixing time. 

\subsection{Hamiltonian Monte Carlo on Riemannian manifolds\label{sec:RMCMC}}

Suppose we want to sample from the distribution $e^{-f(x)}$. We define
the following energy function $H$:
\begin{equation}
H(x,v)\defeq f(x)+\frac{1}{2}\log((2\pi)^{n}\det g(x))+\frac{1}{2}v^{T}g(x)^{-1}v.\label{eq:H_manifold}
\end{equation}
One can view $x$ as the location and $v$ as the velocity. The following
lemma shows that the first variable $x(t)$ in the Hamiltonian curve
satisfies a second-order differential equation. When we view the domain
$\mathcal{M}$ as a manifold, this equation is simply $D_{t}\frac{dx}{dt}=\mu(x)$,
namely, $x$ acts like a particle under the force field $\mu$. (For
relevant background on manifolds, we refer the reader to Appendix
\ref{sec:RG}).
\begin{lem}
\label{lem:RMCMC}In Euclidean coordinates, The Hamiltonian equation
for (\ref{eq:H_manifold}) can be rewritten as 
\begin{align*}
D_{t}\frac{dx}{dt}= & \mu(x),\\
\frac{dx}{dt}(0)\sim & N(0,g(x)^{-1})
\end{align*}
where $\mu(x)=-g(x)^{-1}\nabla f(x)-\frac{1}{2}g(x)^{-1}\tr\left[g(x)^{-1}Dg(x)\right]$
and $D_{t}$ is the Levi-Civita connection on the manifold $\mathcal{M}$
with metric $g$.
\end{lem}

\begin{proof}
From the definition of the Hamiltonian curve, we have that
\begin{align*}
\frac{dx}{dt} & =g(x)^{-1}v\\
\frac{dv}{dt} & =-\nabla f(x)-\frac{1}{2}\tr\left[g(x)^{-1}Dg(x)\right]+\frac{1}{2}\frac{dx}{dt}^{T}Dg(x)\frac{dx}{dt}.
\end{align*}
Putting the two equations together, we have that
\begin{align*}
\frac{d^{2}x}{dt^{2}}= & -g(x)^{-1}Dg(x)[\frac{dx}{dt}]g(x)^{-1}v+g(x)^{-1}\frac{dv}{dt}\\
= & -g(x)^{-1}Dg(x)[\frac{dx}{dt}]\frac{dx}{dt}-g(x)^{-1}\nabla f(x)-\frac{1}{2}g(x)^{-1}\tr\left[g(x)^{-1}Dg(x)\right]+\frac{1}{2}g(x)^{-1}\frac{dx}{dt}^{T}Dg(x)\frac{dx}{dt}.
\end{align*}
Hence, 
\begin{align}
\frac{d^{2}x}{dt^{2}}+g(x)^{-1}Dg(x)[\frac{dx}{dt}]\frac{dx}{dt}-\frac{1}{2}g(x)^{-1}\frac{dx}{dt}^{T}Dg(x)\frac{dx}{dt}= & -g(x)^{-1}\nabla f(x)-\frac{1}{2}g(x)^{-1}\tr\left[g(x)^{-1}Dg(x)\right].\label{eq:ham_dx2}
\end{align}

Using the formula of Christoffel symbols
\[
D_{t}\frac{dx}{dt}=\frac{d^{2}x}{dt^{2}}+\sum_{ijk}\frac{dx_{i}}{dt}\frac{dx_{j}}{dt}\Gamma_{ij}^{k}e_{k}\quad\text{where}\quad\Gamma_{ij}^{k}=\frac{1}{2}\sum_{l}g^{kl}(\partial_{j}g_{li}+\partial_{i}g_{lj}-\partial_{l}g_{ij}),
\]
we have that
\begin{align*}
D_{t}\frac{dx}{dt} & =\frac{d^{2}x}{dt^{2}}+\frac{1}{2}g(x)^{-1}\sum_{ijl}\frac{dx_{i}}{dt}\frac{dx_{j}}{dt}(\partial_{j}g_{li}+\partial_{i}g_{lj}-\partial_{l}g_{ij})e_{l}\\
 & =\frac{d^{2}x}{dt^{2}}+g(x)^{-1}Dg(x)[\frac{dx}{dt}]\frac{dx}{dt}-\frac{1}{2}g(x)^{-1}\frac{dx}{dt}^{T}Dg(x)\frac{dx}{dt}.
\end{align*}
Putting this into (\ref{eq:ham_dx2}) gives
\begin{align*}
D_{t}\frac{dx}{dt}= & -g(x)^{-1}\nabla f-\frac{1}{2}g(x)^{-1}\tr\left[g(x)^{-1}Dg(x)\right].
\end{align*}
\end{proof}
Motivated by this, we define the Hamiltonian map as the first component
of the Hamiltonian dynamics operator $T$ defined earlier. For the
reader familiar with Riemannian geometry, this is similar to the exponential
map (for background, see Appendix \ref{sec:RG}).
\begin{defn}
Let $\ham_{x,\step}(v_{x})=\gamma(\step)$ where $\gamma(t)$ be the
solution of the Hamiltonian equation $D_{t}\frac{d\gamma}{dt}=\mu$
with initial conditions $\gamma(0)=x$ and $\gamma'(0)=v_{x}$. We
also denote $\ham_{x,1}(v_{x})$ by $\ham_{x}(v_{x})$.
\end{defn}

We now give two examples of Hamiltonian Monte Carlo.
\begin{example}
When $g(x)=I$, the Hamiltonian curve acts like stochastic gradient
descent for the function $f$ with each random perturbation drawn
from a standard Gaussian.
\[
D_{t}\frac{dx}{dt}=-\nabla f(x).
\]
When $g(x)=\nabla^{2}f(x)$, the Hamiltonian curve acts like a stochastic
Newton curve for the function $f+\psi$:
\[
D_{t}\frac{dx}{dt}=-\left(\nabla^{2}f(x)\right)^{-1}\nabla(f(x)+\psi(x))
\]
where the volumetric function $\psi(x)=\log\det\nabla^{2}f(x)$.
\end{example}

Next we derive a formula for the transition probability in Euclidean
coordinates.
\begin{lem}
\label{lem:prob_formula} For any $x\in\mathcal{M}\subset\Rn$ and
$s>0$, the probability density of the 1-step distribution from $x$
is given by 
\begin{equation}
p_{x}(y)=\sum_{v_{x}:\ham_{x,\step}(v_{x})=y}\left|\det(D\ham_{x,\step}(v_{x}))\right|^{-1}\sqrt{\frac{\det\left(g(y)\right)}{\left(2\pi\right)^{n}}}\exp\left(-\frac{1}{2}\norm{v_{x}}_{x}^{2}\right)\label{eq:1_step_prof}
\end{equation}
where $D\ham_{x,\step}(v_{x})$ is the Jacobian of the Hamiltonian
map $\ham_{x,\step}$\textup{.}
\end{lem}

\begin{proof}
We prove the formula by separately considering each $v_{x}\in T_{x}\mathcal{M}$
s.t. $\ham_{x,\step}(v_{x})=y$, then summing up. In the tangent space
$T_{x}\mathcal{M}$, the point $v_{x}$ follows a Gaussian step. Therefore,
the probability density of $v_{x}$ in $T_{x}\mathcal{M}$ is as follows:
\[
p_{x}^{T_{x}\mathcal{M}}(v_{x})=\frac{1}{\left(2\pi\right)^{n/2}}\exp\left(-\frac{1}{2}\norm{v_{x}}_{x}^{2}\right).
\]
Let $y=\ham_{x,\step}(v_{x})$ and $F:T_{x}\mathcal{M}\rightarrow\Rn$
be defined by $F(v)=\text{id}_{\mathcal{M}\rightarrow\Rn}\circ\ham_{x,\step}(v)$.
Here $\Rn$ is the same set as $\mathcal{M}$ but endowed with the
Euclidean metric. Hence, we have
\[
DF(v_{x})=D\text{id}_{\mathcal{M}\rightarrow\Rn}(y)D\ham_{x,\step}(v_{x}).
\]
The result follows from $p_{x}(y)=\left|\det(DF(v_{x}))\right|^{-1}p_{x}^{T_{x}M}(v_{x})$
and
\begin{eqnarray*}
\det DF(v_{x}) & = & \det\left(D\text{id}_{\mathcal{M}\rightarrow\Rn}(y)\right)\det\left(D\ham_{x,\step}(v_{x})\right)\\
 & = & \det(g(y))^{-1/2}\det\left(D\ham_{x,\step}(v_{x})\right).
\end{eqnarray*}
\end{proof}

\pagebreak{}
\section{Convergence of Riemannian Hamiltonian Monte Carlo\label{sec:Convergence}}

Hamiltonian Monte Carlo is a Markov chain on a manifold whose stationary
stationary distribution has density $q(x)$ proportional to $\exp(-f(x))$.
We will bound the conductance of this Markov chain and thereby its
mixing time to converge to the stationary distribution. Bounding conductance
involves showing (a) the induced metric on the state space satisfies
a strong isoperimetric inequality and (b) two points that are close
in metric distance are also close in probabilistic distance, i.e.,
the one-step distributions from two nearby points have large overlap.
In this section and the next, we present general conductance bounds
using parameters determined by the associated manifold. In Section
\ref{sec:Logarithmic-Barrier}, we bound these parameters for the
manifold corresponding to the logarithmic barrier in a polytope.

\subsection{Basics of geometric Markov chains}

\label{sec:general_geometric_random_walk}For completeness, we will
discuss some standard techniques in geometric random walks in this
subsection. For a Markov chain with state space $\mathcal{M}$, stationary
distribution $q$ and next step distribution $p_{u}(\cdot)$ for any
$u\in\mathcal{M}$, the conductance of the Markov chain is 

\[
\phi\defeq\inf_{S\subset\mathcal{M}}\frac{\int_{S}p_{u}(\mathcal{M}\setminus S)dq(u)}{\min\left\{ q(S),q(\mathcal{M}\setminus S)\right\} }.
\]

The conductance of an ergodic Markov chain allows us to bound its
mixing time, i.e., the rate of convergence to its stationary distribution,
e.g., via the following theorem of Lovász and Simonovits.
\begin{thm}[\cite{LS93}]
\label{thm:cheeger}Let $q_{t}$ be the distribution of the current
point after $t$ steps of a Markov chain with stationary distribution
$q$ and conductance at least $\phi,$ starting from initial distribution
$q_{0}.$ For any $\varepsilon>0$, 
\[
d_{TV}(q_{t},q)\le\varepsilon+\sqrt{\frac{1}{\varepsilon}\E_{x\sim q_{0}}\frac{dq_{0}(x)}{dq(x)}}\left(1-\frac{\phi^{2}}{2}\right)^{t}.
\]
\end{thm}

\begin{defn}
\label{def:manifold_conductance}The isoperimetry of a metric space
$\mathcal{M}$ with target distribution $q$ is 
\[
\psi\defeq\inf_{\delta>0}\min_{S\subset\mathcal{M}}\frac{\int_{d(S,x)\le\delta}q(x)dx-q(S)}{\delta\min\left\{ q(S),q(\mathcal{M}\setminus S)\right\} }
\]
where $d$ is the shortest path distance in $\mathcal{M}$. 
\end{defn}

The proof of the following theorem follows the standard outline for
geometric random walks (see e.g., \cite{VemSurvey}).
\begin{lem}
\label{lem:gen-convergence_metric}Given a metric space $\mathcal{M}$
and a time-reversible Markov chain $p$ on $\mathcal{M}$ with stationary
distribution $q$. Fix any $r>0$. Suppose that for any $x,y\in\mathcal{M}$
with $d(x,z)<r$, we have that $d_{TV}(p_{x},p_{y})\leq0.9$. Then,
the conductance of the Markov chain is $\Omega(r\psi)$. 
\end{lem}

\begin{proof}
Let $S$ be any measurable subset of $\mathcal{M}$. Then our goal
is to bound the conductance of the Markov chain
\[
\frac{\int_{S}p_{x}(\mathcal{M}\setminus S)\,dq(x)}{\min\left\{ q(S),q(\mathcal{M}\setminus S)\right\} }=\Omega\left(r\psi\right).
\]
Since the Markov chain is time-reversible (For any two subsets $A,B$,
$\int_{A}p_{x}(B)\,dq(x)=\int_{B}p_{x}(A)\,dq(x)$), we can write
the numerator of the left hand side above as 
\[
\frac{1}{2}\left(\int_{S}p_{x}(\mathcal{M}\setminus S)\,dq(x)+\int_{\mathcal{M}\setminus S}p_{x}(S)\,dq(x)\right).
\]
Define
\begin{align*}
S_{1} & =\{x\in S\,:\,p_{x}(\mathcal{M}\setminus S)<0.05\}\\
S_{2} & =\{x\in\mathcal{M}\setminus S\,:\,p_{x}(S)<0.05\}\\
S_{3} & =\mathcal{M}\setminus S_{1}\setminus S_{2}.
\end{align*}
Without loss of generality, we can assume that $q(S_{1})\ge(1/2)q(S)$
and $q(S_{2})\ge(1/2)q(\mathcal{M}\setminus S)$ (if not, $\int_{S}p_{x}(\mathcal{M}\setminus S)\,dq(x)=\Omega(1)$
and hence the conductance is $\Omega(1)$.)

Next, we note that for any two points $x\in S_{1}$ and $y\in S_{2}$,
$d_{TV}(p_{x},p_{y})>0.9$. Therefore, by the assumption, we have
that $d(x,y)\geq r$. Therefore, by the definition of $\psi_{r}$,
we have that
\begin{align*}
q(S_{3}) & \geq\int_{d(S_{1},x)\le r}q(x)dx-q(S_{1})\\
 & \geq r\psi\min\left\{ q(S_{1}),q(\mathcal{M}\setminus S_{1})\right\} \\
 & \geq r\psi\min\left\{ q(S_{1}),q(S_{2})\right\} .
\end{align*}
Going back to the conductance, 

\begin{align*}
\frac{1}{2}\left(\int_{S}p_{x}(\mathcal{M}\setminus S)\,dq(x)+\int_{\mathcal{M}\setminus S}p_{x}(S)\,dq(x)\right) & \ge\frac{1}{2}\int_{S_{3}}(0.05)dq(x)\\
 & =\Omega\left(r\psi\right)\min\{q(S_{1}),q(S_{2})\}\\
 & =\Omega\left(r\psi\right)\min\{q(S),q(M\setminus S)\}.
\end{align*}
Therefore, the conductance of the Markov chain is $\Omega(r\psi)$. 
\end{proof}
Combining Theorem \ref{lem:gen-convergence_metric} and Lemma \ref{lem:gen-convergence_metric}
gives the following result for bounding mixing time of general geometric
random walk.
\begin{lem}
\label{lem:geometric_walk}Given a metric space $\mathcal{M}$ and
a time-reversible Markov chain $p$ on $\mathcal{M}$ with stationary
distribution $q$. Suppose that there exist $r>0$ and $\psi>0$ such
that
\begin{enumerate}
\item For any $x,y\in\mathcal{M}$ with $d(x,z)<r$, we have that $d_{TV}(p_{x},p_{y})\leq0.9$.
\item For any $S\subset\mathcal{M}$, we have that
\[
\int_{0<d(S,x)\le r}q(x)dx\geq r\psi\min\left\{ q(S),q(\mathcal{M}\setminus S)\right\} .
\]
\end{enumerate}
Let $q_{t}$ be the distribution of the current point after $t$ steps
of a Markov chain with stationary distribution $q$ starting from
initial distribution $q_{0}.$ For any $\varepsilon>0$, 
\[
d_{TV}(q_{t},q)\le\varepsilon+\sqrt{\frac{1}{\varepsilon}\E_{x\sim q_{0}}\frac{dq_{0}(x)}{dq(x)}}\left(1-\Omega(r^{2}\psi^{2})\right)^{t}.
\]
\end{lem}

\subsection{Overlap of one-step distributions}

The mixing of the walk depends on smoothness parameters of the manifold
and the functions $f,g$ used to define the Hamiltonian. Since each
step of our walk involves a Gaussian vector, many smoothness parameters
depend on choices of the random vector. Formally, let $\gamma$ be
the Hamiltonian curve used in a step of Hamiltonian Monte Carlo. In
the analysis, we need a large fraction of Hamiltonian curves from
any point on the manifold to be well-behaved. A Hamiltonian curve
can be problematic when its velocity or length is too large and this
happens with non-zero probability. Rather than using supremum bounds
for our smoothness parameters, it suffices to use large probability
bounds, where the probability is over the random choice of Hamiltonian
curve at any point $x\in\Omega$. To capture the notion that ``most
Hamiltonian curves are well-behaved'', we use an auxiliary function
$\ell(\gamma)\geq0$ which assigns a real number to each Hamiltonian
curve $\gamma$ and measures how ``good'' the curve is. The smoothness
parameters assume that this function $\ell$ is bounded and Lipshitz.
One possible choice of such $\ell$ is $\ell(\gamma)=\norm{\gamma'(0)}_{\gamma(0)}$
which measures the initial velocity, but this will give us a weaker
bound. Instead, we use the following which jointly bounds the change
in position (first term) and change in velocity (second term).
\begin{defn}
\label{def:ell_func}An \emph{auxiliary} function $\ell$ is a non-negative
real-valued function on the set of Hamiltonian curves, i.e., maps
$\gamma:[0,\step]\rightarrow M$, with bounded parameters $\ell_{0}$,
$\ell_{1}$ such that 
\begin{enumerate}
\item For any variation $\gamma_{s}$ of a Hamiltonian curve (see Definition
\ref{def:variation}) with $\ell(\gamma_{s})\leq\ell_{0}$, we have
\[
\left|\frac{d}{ds}\ell(\gamma_{s})\right|\leq\ell_{1}\left(\norm{\frac{d}{ds}\gamma_{s}(0)}_{\gamma_{s}(0)}+\step\norm{D_{s}\gamma_{s}'(0)}_{\gamma_{s}(0)}\right).
\]
 
\item For any $x\in M$, $\P_{\gamma\sim x}(\ell(\gamma)\leq\frac{1}{2}\ell_{0})\geq1-\frac{1}{100}\min\left(1,\frac{\ell_{0}}{\ell_{1}\delta}\right)$
where $\gamma\sim x$ indicates a random Hamiltonian curve starting
at $x$, chosen by picking a random Gaussian initial velocity according
to the local metric at $x$. 
\end{enumerate}
\end{defn}

\subsubsection{Proof Outline}

To bound the conductance of HMC, we need to show that one-step distributions
from nearby points have large overlap for reasonably large step size
$\delta$. To this end, recall that the probability density of going
from $x$ to $y$ is given by the following formula
\[
p_{x}(y)=\sum_{v_{x}:\ham_{x,\step}(v_{x})=y}\left|\det\left(D\ham_{x,\step}(v_{x})\right)\right|^{-1}\sqrt{\frac{\det\left(g(y)\right)}{\left(2\pi\right)^{n}}}\exp\left(-\frac{1}{2}\norm{v_{x}}_{x}^{2}\right).
\]
In Section \ref{subsec:Variation-of-Hamiltonian}, we introduce the
concept of variations of Hamiltonian curves and use it to bound $\left|\det\left(D\ham_{x,\step}(v_{x})\right)\right|^{-1}$.
We can show that $p_{x}(y)$ is in fact close to 
\begin{equation}
\widetilde{p}_{x}(y)=\sum_{v_{x}:\ham_{x,\step}(v_{x})=y}\frac{1}{\delta^{n}}\cdot\sqrt{\frac{\det\left(g(y)\right)}{\left(2\pi\right)^{n}}}\exp\left(-\frac{1}{2}\norm{v_{x}}_{x}^{2}\right).\label{eq:p_tilde_x_simplified}
\end{equation}
To compare $p_{x}(y)$ with $p_{z}(y)$, we need to relate $v_{x}$
and $v_{z}$ that map $x$ and $z$ to $y$ respectively. In Section
\ref{subsec:unique_curves}, we shows that if $x$ and $z$ are close
enough, for every $v_{x}$, there is a unique $v_{z}$ such that $v_{x}$
is close to $v_{z}$ and that $\ham_{z,\step}(v_{z})=\ham_{x,\step}(v_{x})$.
Combining these facts, we obtain our main theorem for this section,
stated in Subsection \ref{sec:Smoothness-of-P}.

In the analysis, we use three important operators from the tangent
space to itself. The motivation for defining these operators comes
directly from Lemma \ref{lem:Jacobi_field}, which studies the variation
in Hamiltonian curves as the solution of a Jacobi equation. In words,
the operator $R(.)$ below allows us to write the change in the Hamiltonian
curve as an ODE. 
\begin{defn}
\label{def:R_M_Phi}Given a Hamiltonian curve $\gamma$, let $R(\gamma,t)$,
$M(\gamma,t)$ and $\Phi(\gamma,t)$ be the operators from $TM$ to
$TM$ defined by
\begin{align*}
R(t)u & =R(u,\gamma'(t))\gamma'(t),\\
M(t)u & =D_{u}\mu(\gamma(t)),\\
\Phi(t)u & =M(t)u-R(t)u.
\end{align*}
When $\gamma$ is explicit from the context, we simply write them
as $R(t)$, $M(t)$ and $\Phi(t)$.
\end{defn}

The key parameter $R_{1}$ we use in this section is a bound on the
Frobenius norm of $\Phi$ formally defined as follows.
\begin{defn}
\label{def:R_1}Given a manifold $\mathcal{M}$ with metric $g$ and
an auxiliary function $\ell$ with parameters $\ell_{0},\ell_{1},$
we define the smoothness parameter $R_{1}$ depending only on $\mathcal{M}$
and the step size $\step$ such that 
\[
\norm{\Phi(\gamma,t)}_{F,\gamma(t)}\leq R_{1}
\]
for any $\gamma$ such that $\ell(\gamma)\leq\ell_{0}$ and any $0\leq t\leq\delta$
where the Frobenius norm $\norm A_{F,\gamma(t)}$ is defined by $\norm A_{F,\gamma(t)}^{2}=\E_{\alpha,\beta\sim N(0,g(x)^{-1})}(\alpha^{T}A\beta)^{2}$.
\end{defn}

The above definitions are related to but different from our previous
paper analyzing the geodesic walk \cite{LeeV16}. 

\subsubsection{Variation of Hamiltonian curve}

\label{subsec:Variation-of-Hamiltonian}To bound the determinant of
the Jacobian of $\ham_{x}$, we study variations of Hamiltonian curves.
\begin{defn}
\label{def:variation}We call $\gamma_{s}(t)$ a Hamiltonian variation
if $\gamma_{s}(\cdot)$ satisfies the Hamiltonian equation for every
$s$. We call $\frac{\partial\gamma_{s}}{\partial s}$ a Jacobi field.
\end{defn}

The following lemma shows that a Jacobi field satisfies the following
Jacobi equation. 
\begin{lem}
\label{lem:Jacobi_field}Given a path $c(s)$, let $\gamma_{s}(t)=\ham_{c(s)}(t(v+sw))$
be a Hamiltonian variation. The Jacobi field $\psi(t)\defeq\frac{\partial}{\partial s}\gamma_{s}(t)|_{s=0}$
satisfies the following Jacobi equation
\begin{align}
D_{t}^{2}\psi(t) & =\Phi(t)\psi(t)\label{eq:Jacobi_equ}
\end{align}
Let $\Gamma_{t}$ parallel transport from $T_{\gamma(t,0)}\mathcal{M}$
to $T_{\gamma(0,0)}\mathcal{M}$ and $\overline{\psi}(t)=\Gamma_{t}\psi(t)$.
Then, $\overline{\psi}(t)$ satisfies the following ODE on the tangent
space $T_{\gamma(0,0)}\mathcal{M}$:
\begin{align}
\overline{\psi}''(t) & =\Gamma_{t}\Phi(t)\Gamma_{t}^{-1}\overline{\psi}(t)\quad\forall t\geq0,\label{eq:Jacobi_equ_2}\\
\overline{\psi}'(0) & =w,\nonumber \\
\overline{\psi}(0) & =D_{s}c(0).\nonumber 
\end{align}
\end{lem}

\begin{proof}
Taking derivative $D_{s}$ on both sides of $D_{t}\frac{\partial\gamma}{\partial t}=\mu(\gamma)$,
and using Fact \ref{fact:formula_R}, we get 
\begin{align*}
D_{s}\mu(\gamma) & =D_{s}D_{t}\frac{\partial\gamma}{\partial t}\\
 & =D_{t}D_{s}\frac{\partial\gamma}{\partial t}+R(\frac{\partial\gamma}{\partial s},\frac{\partial\gamma}{\partial t})\frac{\partial\gamma}{\partial t}\\
 & =D_{t}^{2}\frac{\partial\gamma}{\partial s}+R(\frac{\partial\gamma}{\partial s},\frac{\partial\gamma}{\partial t})\frac{\partial\gamma}{\partial t}.
\end{align*}
In short, we have $D_{t}^{2}\psi(t)=\Phi(t)\psi(t)$. This shows (\ref{eq:Jacobi_equ}). 

Equation (\ref{eq:Jacobi_equ_2}) follows from the fact that 
\[
D_{t}v(t)=\Gamma_{t}\frac{d}{dt}\left(\Gamma_{t}^{-1}v(t)\right)
\]
for any vector field on $\gamma_{0}(t)$ (see Definition \ref{def:Directional-derivatives}
in the appendix) applied to $v(t)=\overline{\psi}'(t)$.
\end{proof}
We now proceed to estimate the determinant of the Jacobian of $\ham_{x}$.
For this we will use the following elementary lemmas describing the
solution of the following second-order matrix ODE:
\begin{align}
\frac{d^{2}}{dt^{2}}\Psi(t) & =\Phi(t)\Psi(t),\label{eq:matrix_ODE}\\
\frac{d}{dt}\Psi(0) & =B,\nonumber \\
\Psi(0) & =A.\nonumber 
\end{align}
\begin{lem}
\label{lem:ODE_upper}Consider the matrix ODE (\ref{eq:matrix_ODE}).
Let $\lambda=\max_{0\leq t\leq\ell}\norm{\Phi(t)}_{2}$ . For any
$t\geq0$, we have that
\[
\norm{\Psi(t)}_{2}\leq\norm A_{2}\cosh(\sqrt{\lambda}t)+\frac{\norm B_{2}}{\sqrt{\lambda}}\sinh(\sqrt{\lambda}t).
\]
\end{lem}

\begin{lem}
\label{lem:matrix_ODE_est_1}Consider the matrix ODE (\ref{eq:matrix_ODE}).
Let $\lambda=\max_{0\leq t\leq\ell}\norm{\Phi(t)}_{F}$. For any $0\leq t\leq\frac{1}{\sqrt{\lambda}}$,
we have that
\[
\norm{\Psi(t)-A-Bt}_{F}\leq\lambda\left(t^{2}\norm A_{2}+\frac{t^{3}}{5}\norm B_{2}\right).
\]
In particular, this shows that
\[
\Psi(t)=A+Bt+\int_{0}^{t}(t-s)\Phi(s)(A+Bs+E(s))ds
\]
with $\norm{E(s)}_{F}\leq\lambda\left(s^{2}\norm A_{2}+\frac{s^{3}}{5}\norm B_{2}\right)$. 
\end{lem}

The proofs of these lemmas are in Appendix \ref{sec:Matrix-ODE}.
We continue with the main proof here.
\begin{lem}
\label{lem:Jac-approx}Let $\gamma(t)=\ham_{x}(tv_{x})$ be a Hamiltonian
curve and step size $\step$ satisfy $0<\step^{2}\leq\frac{1}{R_{1}}$
where $R_{1}=\max_{0\leq t\leq h}\norm{\Phi(t)}_{F,\gamma(t)}$. Then
$D\ham_{x,\step}$ is invertible with $\norm{D\ham_{x,\step}-\step I}_{F,\gamma(\delta)}\leq\frac{\step}{5}$.
Also, we have,
\begin{equation}
\left|\log\det\left(\frac{1}{\step}D\ham_{x,\step}(v_{x})\right)-\int_{0}^{\step}\frac{t(\step-t)}{\step}\tr\Phi(t)dt\right|\leq\frac{\left(\step^{2}R_{1}\right)^{2}}{10}.\label{eq:Jac_approx}
\end{equation}
\end{lem}

\begin{proof}
We want to compute $D\ham_{x,\step}(v_{x})[w]$ for some $w\in T_{x}M$.
By definition, we have that 
\begin{equation}
D\ham_{x,\step}(v_{x})[w]=\frac{\partial}{\partial s}\gamma(t,s)|_{t=\step,s=0}\label{eq:Jham_def}
\end{equation}
where $\gamma(t,s)=\ham_{x}(t(v_{x}+sw))$. Define $\overline{\psi}(t)$
as in Lemma \ref{lem:Jacobi_field} with $c(s)=x$. So, $D_{s}c(0)=0$,
i.e., $\overline{\psi}(0)=0$. Then, by the lemma,
\[
\overline{\psi}''(t)=\Gamma_{t}\Phi(t)\Gamma_{t}^{-1}\overline{\psi}(t),\quad\overline{\psi}'(0)=w,\quad\overline{\psi}(0)=0.
\]
Now, we define $\Psi$ be the solution of the matrix ODE
\begin{align*}
\Psi''(t) & =\Gamma_{t}\Phi(t)\Gamma_{t}^{-1}\Psi(t)\quad\forall t\geq0,\\
\Psi'(0) & =I,\\
\Psi(0) & =0.
\end{align*}

By the definition of $\Psi$, we see that $\overline{\psi}(t)=\Psi(t)w$.
Therefore, we have that
\[
\frac{\partial}{\partial s}\gamma(t,s)|_{s=0}=\Gamma_{t}^{-1}\overline{\psi}(t)=\Gamma_{t}^{-1}\Psi(t)w.
\]
Combining it with (\ref{eq:Jham_def}), we have that $D\ham_{x,\step}(v_{x})=\Gamma_{t}^{-1}\Psi(\step)$.
Since $\Gamma_{t}$ is an orthonormal matrix, we have that 
\begin{equation}
\log\det\left(D\ham_{x,\step}(v_{x})\right)=\log\det\Psi(\step).\label{eq:ham_phi_logdet}
\end{equation}

Note that $\norm{\Gamma_{t}\Phi(t)\Gamma_{t}^{-1}}_{F,\gamma(t)}=\norm{\Phi(t)}_{F,\gamma(t)}\leq R_{1}$
for all $0\leq t\leq\step$. Using this, Lemma \ref{lem:matrix_ODE_est_1}
shows that 
\begin{equation}
\norm{\frac{1}{\step}\Psi(\step)-I}_{F,x}\leq R_{1}\left(\frac{\step^{2}}{5}\norm I_{2}\right)\leq\frac{1}{5}\label{eq:estimate_psi_error}
\end{equation}
Hence, $\Psi(\step)$ is invertible, and so is $D\ham_{x}$.

By Lemma \ref{lem:logdet_est}, we have that
\begin{equation}
\left|\log\det(\frac{1}{\step}\Psi(\step))-\tr\left(\frac{1}{\step}\Psi(\step)-I\right)\right|\leq\left(\frac{1}{5}\step^{2}R_{1}\right)^{2}.\label{eq:log_psi_est_1}
\end{equation}

Now we need to estimate $\tr(\Psi(\step)-\step I)$. Lemma \ref{lem:matrix_ODE_est_1}
shows that 
\[
\Psi(\step)=\step I+\int_{0}^{\step}t(\step-t)\Phi(t)dt+\int_{0}^{\step}(\step-t)\Phi(t)E(t)dt
\]
with $\norm{E(t)}_{F,\gamma(t)}\leq\frac{t^{3}R_{1}}{5}$ for all
$0\leq t\leq\step$. Since 
\[
\left|\tr\left(\int_{0}^{\step}(\step-t)\Phi(t)E(t)dt\right)\right|\leq\int_{0}^{\step}(\step-t)\norm{\Phi(t)}_{F,x}\norm{E(t)}_{F,x}dt\leq\frac{\step^{5}}{20}R_{1}^{2},
\]
we have that
\begin{equation}
\left|\tr\left(\frac{1}{\step}\Psi(\step)-I-\int_{0}^{\step}\frac{t(\step-t)}{\step}\Phi(t)dt\right)\right|\leq\frac{\step^{4}}{20}R_{1}^{2}.\label{eq:log_psi_est_2}
\end{equation}
Combining (\ref{eq:log_psi_est_1}) and (\ref{eq:log_psi_est_2}),
we have
\[
\left|\log\det(\frac{1}{\step}\Psi(\step))-\int_{0}^{\step}\frac{t(\step-t)}{\step}\tr\Phi(t)dt\right|\leq\left(\frac{1}{5}\step^{2}R_{1}\right)^{2}+\frac{\step^{4}}{20}R_{1}^{2}\leq\frac{\left(\step^{2}R_{1}\right)^{2}}{10}.
\]
Applying (\ref{eq:ham_phi_logdet}), we have the result.
\end{proof}

\subsubsection{Local Uniqueness of Hamiltonian Curves}

\label{subsec:unique_curves}Next, we study the local uniqueness of
Hamiltonian curves. We know that for every pair $x,y$, there can
be multiple Hamiltonian curves connecting $x$ and $y$. Due to this,
the probability density $p_{x}$ at $y$ in $\mathcal{M}$ is the
sum over all possible Hamiltonian curves connecting $x$ and $y$.
The next lemma establishes a 1-1 map between Hamiltonian curves connecting
$x$ to $y$ as we vary $x.$
\begin{lem}
\label{lem:smoothness_dx}Let $\gamma(t)=\ham_{x}(tv_{x})$ be a Hamiltonian
curve and let the step size $\step$ satisfy $0<\step^{2}\leq\frac{1}{R_{1}}$,
where $R_{1}=\max_{0\leq t\leq\step}\norm{\Phi(t)}_{F,\gamma(t)}$.
Let the end points be $x=\gamma(0)$ and $y=\gamma(\step)$. Then
there is an unique smooth invertible function $v:U\subset\mathcal{M}\rightarrow V\subset T\mathcal{M}$
such that
\[
y=\ham_{z,\step}(v(z))
\]
for any $z\in U$ where $U$ is a neighborhood of $x$ and $V$ is
a neighborhood of $v_{x}=v(x)$. Furthermore, we have that $\norm{\nabla_{\eta}v(x)}_{x}\leq\frac{5}{2\step}\norm{\eta}_{x}$
and
\[
\norm{\frac{1}{\delta}\eta+\nabla_{\eta}v(x)}_{x}\leq\frac{3}{2}R_{1}\delta\norm{\eta}_{x}.
\]

Let $\gamma_{s}(t)=\ham_{c(s)}(t\cdot v(c(s)))$ where $c(s)$ is
any path with $c(0)=x$ and $c'(0)=\eta$. Then, for all $0\leq t\leq\step$,
we have that
\[
\norm{\left.\frac{\partial}{\partial s}\right|_{s=0}\gamma_{s}(t)}_{\gamma(t)}\leq5\norm{\eta}_{x}\quad\text{and}\quad\norm{\left.D_{s}\gamma'_{s}(t)\right|_{s=0}}_{\gamma(t)}\leq\frac{10}{\step}\norm{\eta}_{x}.
\]
\end{lem}

\begin{proof}
Consider the smooth function $f(z,w)=\ham_{z,\step}(w)$. From Lemma
\ref{lem:Jac-approx}, the Jacobian of $w$ at $(x,v_{x})$ in the
$w$ variables, i.e., $D\ham_{x,\step}(v_{x})$, is invertible. Hence,
the implicit function theorem shows that there is a open neighborhood
$U$ of $x$ and a unique function $v$ on $U$ such that $f(z,v(z))=f(x,v_{x})$,
i.e. $\ham_{z,\step}(v(z))=\ham_{x,\step}(v_{x})=y$.

To bound $\nabla_{\eta}v(x)$, we let $\gamma_{s}(t)=\ham_{c(s)}(t\cdot v(c(s)))$
and $c(s)$ be any path with $c(0)=x$ and $c'(0)=\eta$. Let $\Gamma_{t}$
be the parallel transport from $T_{\gamma(t)}\mathcal{M}$ to $T_{\gamma(0)}\mathcal{M}$.
Define 
\[
\overline{\psi}(t)=\Gamma_{t}\left.\frac{\partial}{\partial s}\right|_{s=0}\gamma_{s}(t).
\]
Lemma \ref{lem:Jacobi_field} shows that $\overline{\psi}(t)$ satisfies
the following ODE
\begin{align*}
\overline{\psi}''(t) & =\Gamma_{t}\Phi(t)\Gamma_{t}^{-1}\overline{\psi}(t)\quad\forall t\geq0,\\
\overline{\psi}'(0) & =\nabla_{\eta}v(x),\\
\overline{\psi}(0) & =\eta.
\end{align*}
Moreover, we know that $\overline{\psi}(\step)=0$ because $\gamma_{s}(\step)=\ham_{c(s)}(\step\cdot v(c(s)))=\ham_{c(s),\step}(v(c(s)))=y$
for small enough $s$. 

To bound $\norm{\nabla_{\eta}v(x)}_{x}$, we note that Lemma \ref{lem:matrix_ODE_est_1}
shows that
\begin{equation}
\norm{\overline{\psi}(t)-\eta-t\cdot\nabla_{\eta}v(x)}_{x}\leq R_{1}t^{2}\left(\norm{\eta}_{x}+\frac{t}{5}\norm{\nabla_{\eta}v(x)}_{x}\right).\label{eq:smoothness_dx_ode}
\end{equation}
Since $\overline{\psi}(\step)=0$ and $\step^{2}\leq\frac{1}{R_{1}}$,
we have that 
\[
\norm{\eta+\step\cdot\nabla_{\eta}v(x)}_{x}\leq\norm{\eta}_{x}+\frac{\step}{5}\norm{\nabla_{\eta}v(x)}_{x}
\]
which implies that 
\[
\step\norm{\nabla_{\eta}v(x)}_{x}-\norm{\eta}_{x}\leq\norm{\eta}_{x}+\frac{\step}{5}\norm{\nabla_{\eta}v(x)}_{x}.
\]
Therefore, $\norm{\nabla_{\eta}v(x)}_{x}\leq\frac{5}{2\step}\norm{\eta}_{x}$.
More precisely, from (\ref{eq:smoothness_dx_ode}), we have that
\[
\norm{\frac{1}{\delta}\eta+\nabla_{\eta}v(x)}_{x}\leq\frac{3}{2}R_{1}\delta\norm{\eta}_{x}.
\]

Putting this into (\ref{eq:smoothness_dx_ode}), for $t\le\step$
we get
\[
\norm{\overline{\psi}(t)}_{x}\leq\norm{\eta}_{x}+\frac{5}{2}\norm{\eta}_{x}+R_{1}t^{2}\left(\norm{\eta}_{x}+\frac{1}{2}\norm{\eta}_{x}\right)\leq5\norm{\eta}_{x}.
\]
Now, apply the conclusion of Lemma \ref{lem:matrix_ODE_est_1} after
taking a derivative, we have that 
\[
\overline{\psi}'(t)=\nabla_{\eta}v(x)+\int_{0}^{t}\Phi(s)(\eta+s\cdot\nabla_{\eta}v(x)+E(s))ds
\]
where $\norm{E(s)}_{x}\leq R_{1}\left(s^{2}\norm{\eta}_{x}+\frac{s^{3}}{5}\norm{\nabla_{\eta}v(x)}_{x}\right)\leq\frac{3}{2}\norm{\eta}_{x}$.
Hence, bounding each term and noting that $t\le\delta$, we have
\[
\norm{\overline{\psi}'(t)}_{x}\leq\frac{5}{2\step}\norm{\eta}_{x}+\step R_{1}\left(1+\frac{5}{2}+\frac{3}{2}\right)\norm{\eta}_{x}\leq\frac{10}{\step}\norm{\eta}_{x}.
\]
\end{proof}
When we vary $x$, the Hamiltonian curve $\gamma$ from $x$ to $y$
varies and we need to bound $\ell(\gamma)$ over the variation. 
\begin{lem}
\label{lem:one_one_cor}Given a Hamiltonian curve $\gamma(t)=\ham_{x}(t\cdot v_{x})$
with step size $\step$ satisfying $\step^{2}\leq\frac{1}{R_{1}}$,
let $c(s)$ be any geodesic starting at $\gamma(0)$. Let $x=c(0)=\gamma(0)$
and $y=\gamma(\step)$. Suppose that the auxiliary function $\ell$
satisfies $\norm{\frac{dc}{ds}}_{c(0)}\leq\frac{\ell_{0}}{7\ell_{1}}$
and $\ell(\gamma)\leq\frac{1}{2}\ell_{0}$. Then, there is a unique
vector field $v$ on $c$ such that 
\[
y=\ham_{c(s),\step}(v(s)).
\]
Moreover, this vector field is uniquely determined by the geodesic
$c(s)$ and any $v(s)$ on this vector field. Also, we have that $\ell(\ham_{c(s),\step}(v(s)))\leq\ell_{0}$
for all $s\le1$.
\end{lem}

\begin{proof}
Let $s_{\max}$ be the supremum of $s$ such that $v(s)$ can be defined
continuously such that $y=\ham_{c(s),\step}(v(s))$ and $\ell(\gamma_{s})\leq\ell_{0}$
where $\gamma_{s}(t)=\ham_{c(s)}(t\cdot v(s)).$ Lemma \ref{lem:smoothness_dx}
shows that there is a neighborhood $N$ at $x$ and a vector field
$u$ on $N$ such that for any $z\in N$, we have that 
\[
y=\ham_{z,\step}(u(z)).
\]
Also, this lemma shows that $u(s)$ is smooth and hence the parameter
$\ell_{1}$ shows that $\ell(\gamma_{s})$ is Lipschitz in $s$. Therefore,
$\ell(\gamma_{s})\leq\ell_{0}$ in a small neighborhood of $0$. Hence
$s_{\max}>0$.

Now, we show $s_{\max}>1$ by contradiction. By the definition of
$s_{\max}$, we have that $\ell(\gamma_{s})\leq\ell_{0}$ for any
$0\leq s<s_{\max}$. Hence, we can apply Lemma \ref{lem:smoothness_dx}
to show that $\norm{D_{s}v(s)}_{\gamma(s)}\leq\frac{5}{2\delta}\norm{\frac{dc}{ds}}_{\gamma(s)}=\frac{5}{2\delta}L$
where $L$ is the length of $c$ up to $s=1$ (since the speed is
constant on any geodesic and the curve is defined over $[0,1]$).
Therefore, the function $v$ is Lipschitz and hence $v(s_{\max})$
is well-defined and $\ell(\gamma_{s_{\max}})\leq\ell_{0}$ by continuity.
Hence, we can apply Lemma \ref{lem:smoothness_dx} at $\ell(s_{\max})$
and extend the domain of $\ell(s)$ beyond $s_{\max}$.

To bound $\ell(\gamma_{s})$ beyond $s_{\max}$, we note that $\norm{D_{s}\gamma_{s}'}_{\gamma(s)}=\norm{D_{s}v(s)}_{\gamma(s)}\leq\frac{5}{2\step}L$
and $\norm{\frac{d}{ds}c}_{\gamma(s)}=\norm{\frac{d}{ds}c}_{\gamma(0)}=L$.
Hence, $\left|\frac{d}{ds}\ell(\gamma_{s})\right|\leq(L+\frac{5}{2}L)\ell_{1}$
by the definition of $\ell_{1}$. Therefore, if $L\leq\frac{\ell_{0}}{7\ell_{1}}$,
we have that $\ell(\gamma_{s})\leq\ell(\gamma)+\frac{1}{2}\ell_{0}\leq\ell_{0}$
for all $s\leq1.01$ wherever $v(s)$ is defined. Therefore, this
contradicts the assumption that $s_{\max}$ is the supremum. Hence,
$s_{\max}>1$. 

The uniqueness follows from Lemma \ref{lem:smoothness_dx}.
\end{proof}

\subsubsection{Smoothness of one-step distributions}

\label{sec:Smoothness-of-P}
\begin{lem}
\label{lem:TV_diff_simplified}For $\step^{2}\leq\frac{1}{100\sqrt{n}R_{1}}$
and $\delta^{3}\leq\frac{\ell_{0}}{100\sqrt{n}R_{1}\ell_{1}}$, the
one-step Hamiltonian walk distributions $p_{x},p_{z}$ from $x,z$
satisfy
\[
d_{\mathrm{TV}}(p_{x},p_{z})=O\left(\frac{1}{\step}\right)d(x,z)+\frac{1}{25}.
\]
\end{lem}

\begin{proof}
We first consider the case $d(x,y)<\frac{\ell_{0}}{7\ell_{1}}$. Let
$c(s)$ be a unit speed geodesic connecting $x$ and $z$ of length
$L<\frac{\ell_{0}}{7\ell_{1}}$.

Let $\widetilde{\ell}\defeq\min\left(1,\frac{\ell_{0}}{\ell_{1}\delta}\right)$.
By the definition of $\ell_{0}$, with probability at least $1-\widetilde{\ell}$
in paths $\gamma$ start at $x$, we have that $\ell(\gamma)\leq\frac{1}{2}\ell_{0}$
. Let $V_{x}$ be the set of $v_{x}$ such that $\ell(\ham_{x}(t\cdot v_{x}))\leq\frac{1}{2}\ell_{0}$.
Since the distance from $x$ to $z$ is less than $\frac{\ell_{0}}{7\ell_{1}}$
and $\ell(\gamma)\leq\frac{1}{2}\ell_{0}$, for those $v_{x}$, Lemma
\ref{lem:one_one_cor} shows there is a family of Hamiltonian curves
$\gamma_{s}$ that connect $c(s)$ to $y$, and $\ell(\gamma_{s})\leq\ell_{0}$
for each of them. 

For any $v_{x}\in V$, we have that $\ell(\gamma_{s})\leq\ell_{0}.$
When $\ell(\gamma_{s})\leq\ell_{0}$, by the definition of $R_{1}$,
we indeed have that $\norm{\Phi(t)}\leq R_{1}$ and hence Lemma \ref{lem:Jac-approx}
shows that
\begin{align*}
\left|\log\det\left(\frac{1}{\step}D\ham_{x,\step}(v_{x})\right)\right| & \leq\left|\int_{0}^{\step}\frac{t(\step-t)}{\step}\tr\Phi(t)dt\right|+\frac{\left(\step^{2}R_{1}\right)^{2}}{10}\\
 & \leq\frac{\delta^{2}}{6}\sqrt{n}R_{1}+\frac{\left(\step^{2}R_{1}\right)^{2}}{10}\leq\frac{1}{600}\widetilde{\ell}
\end{align*}
where we used our assumption on $\delta$. We use $p(v_{x})$ to denote
the probability density of choosing $v_{x}$ and
\[
\widetilde{p}(v_{x})\defeq\sqrt{\frac{\det\left(g(\ham_{x}(\delta\cdot v_{x}))\right)}{\left(2\pi\step^{2}\right)^{n}}}\exp\left(-\frac{1}{2}\norm{v_{x}}_{x}^{2}\right).
\]
Hence, we have that
\begin{equation}
C^{-1}\cdot p(v_{x})\leq\widetilde{p}(v_{x})\leq C\cdot p(v_{x})\label{eq:jac_approx_2}
\end{equation}
where $C\defeq1+\frac{1}{600}\widetilde{\ell}$. As we noted, for
every $v_{x}$, there is a corresponding $v_{z}$ such that $\ham_{z}(\delta\cdot v_{z})=y$
and $\ell(\ham_{z}(t\cdot v_{z}))\leq\ell_{0}$. Therefore, we have
that
\begin{align*}
(1-C^{2})p(v_{x})+C(\widetilde{p}(v_{x})-\widetilde{p}(v_{z}))\leq p(v_{x})-p(v_{z})\leq & (1-C^{-2})p(v_{x})+C^{-1}(\widetilde{p}(v_{x})-\widetilde{p}(v_{z})).
\end{align*}
Since $\int_{V_{x}}p(v_{x})dv_{x}\geq1-\frac{1}{100}\min\left(1,\frac{\ell_{0}}{\ell_{1}\delta}\right)$,
we have that
\begin{align}
d_{TV}(p_{x},p_{z}) & \leq\frac{\widetilde{\ell}}{100}+\int_{V_{x}}\left|p(v_{x})-p(v_{z})\right|dv_{x}\nonumber \\
 & \leq\frac{\widetilde{\ell}}{100}(1+\int_{V_{x}}p(v_{x})dv_{x})+2\int_{V_{x}}\left|\widetilde{p}(v_{x})-\widetilde{p}(v_{z})\right|dv_{x}\nonumber \\
 & \leq\frac{\widetilde{\ell}}{50}+2\int_{V_{x}}\int_{s}\left|\frac{d}{ds}\widetilde{p}(v_{c(s)})\right|dsdv_{x}\label{eq:d_TV_integrate}
\end{align}

Note that
\[
\frac{d}{ds}\widetilde{p}(v_{c(s)})=\left(-\frac{1}{2}\frac{d}{ds}\norm{v(s)}_{c(s)}^{2}\right)\widetilde{p}(v_{c(s)}).
\]
Using (\ref{eq:jac_approx_2}), we have that $\widetilde{p}(v_{c(s)})\leq2\cdot p(v_{c(s)})$
and hence
\begin{align}
\int_{V_{x}}\left|\frac{d}{ds}\widetilde{p}(v_{c(s)})\right|dy\leq & \int_{V_{x}}\left|\frac{d}{ds}\norm{v(s)}_{c(s)}^{2}\right|p(v_{c(s)})dy\nonumber \\
\leq & \E_{\ell(\gamma_{s})\leq\ell_{0}}\left|\frac{d}{ds}\norm{v(s)}_{c(s)}^{2}\right|\label{eq:second_term_TV}
\end{align}
Using that $\norm{\partial_{s}c(s)}=1$, Lemma \ref{lem:smoothness_dx}
shows that
\[
\norm{\frac{1}{\delta}\partial_{s}c(s)+D_{s}v(s)}_{x}\leq\frac{3}{2}R_{1}\delta.
\]
Therefore, we have that
\begin{align*}
\left|\frac{d}{ds}\norm{v(s)}_{c(s)}^{2}\right| & =2\left|\left\langle v(s),D_{s}v(s)\right\rangle _{c(s)}\right|\\
 & \leq\frac{2}{\delta}\left|\left\langle v(s),\partial_{s}c(s)\right\rangle _{c(s)}\right|+3R_{1}\delta\norm{v(s)}_{c(s)}.
\end{align*}
Since $v(s)$ is a random Gaussian vector from the local metric, we
have that $\left|\left\langle v(s),\partial_{s}c(s)\right\rangle _{c(s)}\right|=O(1)$
and $\norm{v(s)}_{c(s)}=O(\sqrt{n})$ with high probability. Putting
it into (\ref{eq:second_term_TV}), we have that
\[
\E_{\ell(\gamma_{s})\leq\ell_{0}}\left|\frac{d}{ds}\norm{v(s)}_{c(s)}^{2}\right|=O\left(\frac{1}{\step}+\step R_{1}\sqrt{n}\right).
\]
Putting this into \ref{eq:second_term_TV}, we have that
\[
\int_{V_{x}}\left|\frac{d}{ds}\widetilde{p}(v_{c(s)})\right|dy=O\left(\frac{1}{\step}+\step R_{1}\sqrt{n}\right).
\]

Putting this into (\ref{eq:d_TV_integrate}), we get
\[
d_{\mathrm{TV}}(p_{x},p_{z})=O\left(\frac{1}{\step}+\step R_{1}\sqrt{n}\right)L+\frac{\widetilde{\ell}}{50}
\]
for any $L<\frac{\ell_{0}}{7\ell_{1}}$. By taking a minimal length
geodesic, and summing over segment of length $\frac{\ell_{0}}{8\ell_{1}}$,
for any $x$ and $z$, we have 
\begin{align*}
d_{\mathrm{TV}}(p_{x},p_{z}) & =O\left(\frac{1}{\step}+\step R_{1}\sqrt{n}\right)d(x,z)+\frac{1}{25}\\
 & =O\left(\frac{1}{\step}\right)d(x,z)+\frac{1}{25}.
\end{align*}
\end{proof}

\subsection{Convergence bound}

Combining Lemma \ref{lem:TV_diff_simplified} and Lemma \ref{lem:geometric_walk},
we have the following result.
\begin{thm}
\label{thm:gen-convergence_simplified}Given a manifold $\mathcal{M}$.
Let $\ell_{0},\ell_{1},R_{1}$ be the parameters of the Hamiltonian
Monte Carlo defined in Definition \ref{def:ell_func} and \ref{def:R_1}.
Let $q_{t}$ be the distribution of the current point after $t$ steps
Hamiltonian Monte Carlo with step size $\step$ satisfying 
\[
\step^{2}\leq\frac{1}{100\sqrt{n}R_{1}}\quad\text{and}\quad\delta^{3}\leq\frac{\ell_{0}}{100\sqrt{n}R_{1}\ell_{1}},
\]
starting from initial distribution $q_{0}.$ Let $q$ be the distribution
proportional to $e^{-f}$. For any $\varepsilon>0$, we have that
\[
d_{TV}(q_{t},q)\le\varepsilon+\sqrt{\frac{1}{\varepsilon}\E_{x\sim q_{0}}\frac{dq_{0}(x)}{dq(x)}}\left(1-\frac{(\delta\psi)^{2}}{2}\right)^{t}
\]
where $\psi$ is the conductance of the manifold defined in Definition
\ref{def:manifold_conductance}.
\end{thm}

For $\ell(\gamma)=\norm{\gamma'(0)}_{\gamma(0)}$, we can bound $\ell_{0}$
and $\ell_{1}$ as follows:
\begin{lem}
\label{lem:simple_ell}For the auxiliary function $\ell(\gamma)=\norm{\gamma'(0)}_{\gamma(0)}$,
we have that $\ell_{0}=10\sqrt{n}$ and $\ell_{1}=O(\frac{1}{\delta})$.
Furthermore, we have that 
\[
\norm{\gamma'(t)}_{\gamma(t)}\leq\norm{\gamma'(0)}_{\gamma(0)}+R_{0}t
\]
where $R_{0}=\sup_{x\in\mathcal{M}}\norm{\mu(x)}_{x}$.
\end{lem}

\begin{proof}
For $\ell_{0}$, we note that $\gamma'(0)\sim N(0,g(x)^{-1})$ and
hence $\norm{\gamma'(0)}_{\gamma(0)}\leq5\sqrt{n}$ with probability
$e^{-O(\sqrt{n})}$.

For $\ell_{1}$, we note that
\begin{align*}
\left|\frac{d}{ds}\ell(\gamma_{s})\right| & =\left|\frac{d}{ds}\norm{\gamma_{s}'(0)}_{\gamma_{s}(0)}\right|\leq\frac{1}{2}\left|\frac{\frac{d}{ds}\norm{\gamma_{s}'(0)}_{\gamma_{s}(0)}^{2}}{\norm{\gamma_{s}'(0)}_{\gamma_{s}(0)}}\right|\\
 & =\left|\frac{\left\langle D_{s}\gamma_{s}'(0),\gamma_{s}'(0)\right\rangle _{\gamma_{s}(0)}}{\norm{\gamma_{s}'(0)}_{\gamma_{s}(0)}}\right|\\
 & \leq\norm{D_{s}\gamma_{s}'(0)}_{\gamma_{s}(0)}.
\end{align*}
Hence, we have that $\ell_{1}=\frac{1}{\delta}$.

Next, we note that
\begin{align*}
\left|\frac{d}{dt}\norm{\gamma'(t)}_{\gamma(t)}^{2}\right| & =2\left|\left\langle D_{t}\gamma'(t),\gamma'(t)\right\rangle _{\gamma(t)}\right|\\
 & \leq2\norm{\mu(\gamma(t))}_{\gamma(t)}\norm{\gamma'(t)}_{\gamma(t)}.
\end{align*}
Therefore, we get the last result.
\end{proof}
Combining Lemma \ref{lem:simple_ell} and Theorem \ref{thm:gen-convergence_simplified},
we have the following result. This result might be more convenient
to establish an upper bound on the rate of convergence, as it depends
on only two worst-case smoothness parameters. In the next section,
we will see a more refined bound that uses the randomness of Hamiltonian
curves via additional parameters. 
\begin{cor}
\label{cor:gen-convergence_simplified}Given a manifold $\mathcal{M}$.
Let $q_{t}$ be the distribution of the current point after $t$ steps
Hamiltonian Monte Carlo with step size $\step$ and $q$ be the distribution
proportional to $e^{-f}$. Let $R_{0}$ and $R_{1}$ be parameters
such that
\begin{enumerate}
\item $\norm{\mu(x)}_{x}\leq R_{0}$ for any $x\in\mathcal{M}$ where $\mu$
is defined in Lemma \ref{lem:RMCMC}.
\item $\E_{\alpha,\beta\sim N(0,g(x)^{-1})}\left\langle D_{\alpha}\mu(x),\beta\right\rangle _{x}^{2}\leq R_{1}^{2}$
for any $x\in\mathcal{M}$.
\item $\E_{\alpha,\beta\sim N(0,g(x)^{-1})}\left\langle R(\alpha,v)v,\beta\right\rangle _{x}^{2}\leq R_{1}^{2}$
for any $x\in\mathcal{M}$ and any $\norm v_{x}\leq\sqrt{n}$ where
$R$ is the Riemann curvature tensor of $\mathcal{M}$.
\end{enumerate}
Suppose that $\delta\leq\frac{\sqrt{n}}{R_{0}}$ and $\step^{2}\leq\frac{1}{100\sqrt{n}R_{1}}$,
then for any $\varepsilon>0$, we have that
\[
d_{TV}(q_{t},q)\le\varepsilon+\sqrt{\frac{1}{\varepsilon}\E_{x\sim q_{0}}\frac{dq_{0}(x)}{dq(x)}}\left(1-\frac{(\delta\psi)^{2}}{2}\right)^{t}
\]
where $\psi$ is the conductance of the manifold defined in Definition
\ref{def:manifold_conductance}. In short, the mixing time of Hamiltonian
Monte Carlo is
\[
\widetilde{O}\left(\psi^{-2}\left(\sqrt{n}R_{1}+\frac{R_{0}^{2}}{n}\right)\right).
\]
\end{cor}

\begin{proof}
The statement is basically restating the definition of $R_{1}$ and
$R_{0}$ used in Lemma \ref{lem:simple_ell} and Theorem \ref{thm:gen-convergence_simplified}.
The only difference is that if $\delta\leq\frac{\sqrt{n}}{R_{0}}$,
then we know that 
\[
\norm{\gamma'(t)}_{\gamma(t)}\leq\norm{\gamma'(0)}_{\gamma(0)}+R_{0}t=O(\sqrt{n})
\]
for all $0\leq t\leq\delta$. Therefore, we can relax the constraints
$\ell(\gamma)\leq\ell_{0}$ in the definition of $R_{1}$ to simply
$\norm v_{x}\leq\sqrt{n}$. It allows us to use $R_{1}$ without mentioning
the auxiliary function $\ell$.
\end{proof}

\section{Improved analysis of the convergence\label{sec:Convergence_imporved}}

Corollary \ref{cor:gen-convergence_simplified} gives a polynomial
mixing time for the log barrier function. There are two bottlenecks
to improving the bound. First, the auxiliary function $\ell$ it used
does not capture the fact each curve $\gamma$ in Hamiltonian Monte
Carlo follows a random initial direction. Second, Lemma \ref{lem:smoothness_dx}
also does not take full advantage of the random initial direction.
In this section, we focus on improving Lemma \ref{lem:smoothness_dx}.

Our main theorem for convergence can be stated as follows in terms
of $\psi$ and additional parameters $\ell_{0},\ell_{1},R_{1},R_{2},R_{3}$
(see Definitions \ref{def:manifold_conductance}, \ref{def:ell_func},
\ref{def:R_M_Phi}, \ref{def:R_1}, \ref{def:R2} and \ref{def:R3}).
It uses the following key lemma.

\begin{restatable}{thm}{TVdiff}

\label{thm:TV_diff_improved}For $\step^{2}\leq\frac{1}{R_{1}}$ and
$\delta^{5}\leq\frac{\ell_{0}}{R_{1}^{2}\ell_{1}}$, the one-step
Hamiltonian walk distributions $p_{x},p_{z}$ from $x,z$ satisfy
\[
d_{\mathrm{TV}}(p_{x},p_{z})=O\left(\step^{2}R_{2}+\frac{1}{\step}+\step R_{3}\right)d(x,z)+\frac{1}{25}.
\]

\end{restatable}
\begin{rem*}
The constant term at the end is an artifact that comes from bounding
the probability of some bad events of the Hamiltonian walk, and can
be made to arbitrary small. 
\end{rem*}
We now prove this theorem. In a later section, we specialize to sampling
distributions over polytopes using the logarithmic barrier, by defining
a suitable auxiliary function and bounding all the parameters. The
key ingredient is Theorem \ref{thm:TV_diff_improved} about the overlap
of one-step distributions, which we prove in the next section. 

Here is the consequence of Theorem \ref{thm:TV_diff_improved} and
Theorem \ref{lem:gen-convergence_metric}.
\begin{thm}
\label{thm:gen-convergence}Given a manifold $\mathcal{M}$. Let $\ell_{0},\ell_{1},R_{1},R_{2},R_{3}$
be the parameters of the Hamiltonian Monte Carlo defined in Definition
\ref{def:ell_func}, \ref{def:R_1}, \ref{def:R2} and \ref{def:R3}.
Let $q_{t}$ be the distribution of the current point after $t$ steps
Hamiltonian Monte Carlo starting from initial distribution $q_{0}.$
Let $q$ be the distribution proportional to $e^{-f}$. Suppose that
the step size $\step$ satisfies 
\[
\step^{2}\leq\frac{1}{R_{1}},\ \delta^{5}\leq\frac{\ell_{0}}{R_{1}^{2}\ell_{1}}\text{ and }\step^{3}R_{2}+\delta^{2}R_{3}\leq1.
\]
For any $\varepsilon>0$, we have that
\[
d_{TV}(q_{t},q)\le\varepsilon+\sqrt{\frac{1}{\varepsilon}\E_{x\sim q_{0}}\frac{dq_{0}(x)}{dq(x)}}\left(1-\frac{(\delta\psi)^{2}}{2}\right)^{t}
\]
where $\psi$ is the conductance of the manifold.
\end{thm}

\subsection{Improved one-to-one correspondence for Hamiltonian curve}

In the previous section, we only used $R_{1}$ to analyze how much
a Hamiltonian curves change as one end point varies. Here we derive
a more refined analysis of Lemma \ref{lem:smoothness_dx} using an
additional parameter $R_{3}$. 
\begin{defn}
\label{def:R3}For a manifold $M$ and auxiliary function $\ell$,
$R_{3}$ is a constant such that for any Hamiltonian curve $\gamma(t)$
of step size $\step$ with $\ell(\gamma_{0})\leq\ell_{0}$, if $\zeta(t)$
is the parallel transport of the vector $\gamma'(0)$ along $\gamma(t)$,
then we have
\[
\sup_{0\leq t\leq\step}\norm{\Phi(t)\zeta(t)}_{\gamma(t)}\leq R_{3}.
\]
\end{defn}

\begin{lem}
\label{lem:smoothness_dx2}Under the same assumptions as Lemma \ref{lem:smoothness_dx},
we have that 
\[
\frac{\step}{2}\left|\nabla_{\eta}\norm{v(x)}_{x}^{2}\right|\leq\left|\left\langle v_{x},\eta\right\rangle _{x}\right|+3\step^{2}R_{3}\norm{\eta}_{x}.
\]
\end{lem}

\begin{proof}
Let $\chi=\nabla_{\eta}v(x)$. Define $\overline{\psi}(t)$ as in
the proof of Lemma \ref{lem:smoothness_dx}. Using Lemma \ref{lem:matrix_ODE_est_1}
we get that
\[
0=\overline{\psi}(\step)=\eta+\step\chi+\int_{0}^{\step}(\step-s)\Gamma_{s}\Phi(s)\Gamma_{s}^{-1}(\eta+s\chi+e(s))ds
\]
i.e., 
\[
-\step\chi=\eta+\int_{0}^{\step}(\step-s)\Gamma_{s}\Phi(s)\Gamma_{s}^{-1}(\eta+s\chi+E(s))ds
\]
with 
\[
\norm{E(s)}_{x}\leq R_{1}\left(\step^{2}\norm{\eta}_{x}+\frac{\step^{3}}{5}\norm{\chi}_{x}\right)\leq2R_{1}\step^{2}\norm{\eta}_{x}
\]
for all $0\leq s\leq\step$ where we used $\norm{\nabla_{\eta}v(x)}_{x}\leq\frac{5}{2\step}\norm{\eta}_{x}$
at the end (by Lemma \ref{lem:smoothness_dx}).

Therefore, 
\[
\step\left|\langle v(x),\chi\rangle_{x}\right|\leq\left|\langle v(x),\eta\rangle_{x}\right|+\frac{\step^{2}}{2}\sup_{0\leq s\leq\step}\left|\langle v(x),\Gamma_{s}\Phi(s)\Gamma_{s}^{-1}\left(\eta+s\chi+E(s)\right)\rangle_{x}\right|.
\]
Noting that
\[
\norm{\eta+s\chi+E(s)}_{x}\leq\norm{\eta}_{x}+\step\norm{\chi}_{x}+\norm{E(s)}_{x}\leq(1+\frac{5}{2}+2R_{1}\step^{2})\norm{\eta}_{x}\leq6\norm{\eta}_{x},
\]
we have
\begin{align*}
\step\left|\langle v(x),\chi\rangle_{x}\right| & \leq\left|\langle v(x),\eta\rangle_{x}\right|+3\step^{2}\sup_{0\leq s\leq\step}\norm{\gamma'(0)^{T}\Gamma_{s}\Phi(s)\Gamma_{s}^{-1}}_{x}\norm{\eta}_{x}\\
 & =\left|\langle v(x),\eta\rangle_{x}\right|+3\step^{2}\sup_{0\leq s\leq\step}\norm{\Phi(s)\Gamma_{s}\gamma'(0)}_{\gamma(s)}\norm{\eta}_{x}\\
 & \leq\left|\langle v(x),\eta\rangle_{x}\right|+3\step^{2}R_{3}\norm{\eta}_{x}.
\end{align*}
Finally, we note that
\[
\frac{\step}{2}\left|\nabla_{\eta}\norm{v(x)}_{x}^{2}\right|=\step\left|\left\langle v(x),\nabla_{\eta}v(x)\right\rangle _{x}\right|.
\]
\end{proof}

\subsection{Improved smoothness of one-step distributions}

The proof of Theorem \ref{thm:TV_diff_improved} is pretty similar
to Lemma \ref{lem:TV_diff_simplified}. First, we show that $p_{x}(y)$
is in fact close to 
\begin{equation}
\widetilde{p}_{x}(y)=\sum_{v_{x}:\ham_{x,\step}(v_{x})=y}\sqrt{\frac{\det\left(g(y)\right)}{\left(2\pi\step^{2}\right)^{n}}}\exp\left(-\int_{0}^{\step}\frac{t(\step-t)}{\step}\tr\Phi(t)dt-\frac{1}{2}\norm{v_{x}}_{x}^{2}\right).\label{eq:p_tilde_x}
\end{equation}
Note that this is a more refined estimate than (\ref{eq:p_tilde_x_simplified}).
We use Lemma \ref{lem:smoothness_dx2} to bound the change of $\norm{v_{x}}_{x}^{2}$.
For the change of $\tr\Phi(t)$, we defer the calculation until the
end of this section.

\TVdiff*
\begin{proof}
We first consider the case $d(x,y)<\frac{\ell_{0}}{7\ell_{1}}$. By
a similar argument as Lemma \ref{lem:TV_diff_simplified}, using that
$\step^{2}\leq\frac{1}{R_{1}}$ and $\delta^{5}\leq\frac{\ell_{0}}{R_{1}^{2}\ell_{1}}$,
we have that
\begin{align}
d_{TV}(p_{x},p_{z}) & \leq\frac{\widetilde{\ell}}{50}+2\int_{V_{x}}\int_{s}\left|\frac{d}{ds}\widetilde{p}(v_{c(s)})\right|dsdv_{x}\label{eq:d_TV_integrate-2}
\end{align}
where $\widetilde{\ell}\defeq\min\left(1,\frac{\ell_{0}}{\ell_{1}\delta}\right)$.
By direct calculation, we have
\[
\frac{d}{ds}\widetilde{p}(v_{c(s)})=\left(-\int_{0}^{\step}\frac{t(\step-t)}{\step}\frac{d}{ds}\tr\Phi(\gamma'_{s}(t))dt-\frac{1}{2}\frac{d}{ds}\norm{v(s)}_{c(s)}^{2}\right)\widetilde{p}(v_{c(s)}).
\]
By similar argument as (\ref{eq:jac_approx_2}), we have that $\widetilde{p}(v_{c(s)})\leq2\cdot p(v_{c(s)})$
and hence
\[
\left|\frac{d}{ds}\widetilde{p}(v_{c(s)})\right|\leq2\left(\left|\int_{0}^{\step}\frac{t(\step-t)}{\step}\frac{d}{ds}\tr\Phi(\gamma'_{s}(t))dt\right|+\frac{1}{2}\left|\frac{d}{ds}\norm{v(s)}_{c(s)}^{2}\right|\right)p(v_{c(s)}).
\]

Since $\ell(\gamma_{s})\leq\ell_{0}$, we can use Lemma \ref{lem:smoothness_logdetJ}
to get
\[
\left|\int_{0}^{\step}\frac{t(\step-t)}{\step}\frac{d}{ds}\tr\Phi(\gamma'_{s}(t))dt\right|\leq O\left(\step^{2}R_{2}\right).
\]
Hence, 
\begin{align}
\int_{V_{x}}\left|\frac{d}{ds}\widetilde{p}(v_{c(s)})\right|dv_{x}\leq & O(\step^{2}R_{2})\int_{V_{x}}p(v_{c(s)})dv_{x}+\int_{V_{x}}\left|\frac{d}{ds}\norm{v(s)}_{c(s)}^{2}\right|p(v_{c(s)})dv_{x}.\label{eq:TV_term-1}
\end{align}

For the first term, we note that $\int_{V_{x}}p(v_{c(s)})dv_{x}\leq1$.

For the second term, we have that
\begin{align}
\int_{V_{x}}\left|\frac{d}{ds}\norm{v(s)}_{c(s)}^{2}\right|p(v_{c(s)})dv_{x}\leq & \E_{\ell(\gamma_{s})\leq\ell_{0}}\left|\frac{d}{ds}\norm{v(s)}_{c(s)}^{2}\right|.\label{eq:second_term_TV-1}
\end{align}
By Lemma \ref{lem:smoothness_dx2}, we have that
\[
\frac{\step}{2}\left|\frac{d}{ds}\norm{v(s)}_{c(s)}^{2}\right|\leq\left|\left\langle v(s),\partial_{s}c\right\rangle _{x}\right|+3\step^{2}R_{3}\norm{\partial_{s}c}_{x}.
\]
Since $\norm{\partial_{s}c}_{x}=1$ and $\left|\left\langle v(s),\partial_{s}c\right\rangle _{x}\right|=O(1)$
with high probability since $v(s)$ is a Gaussian vector from the
local metric. We have
\[
\frac{\step}{2}\left|\frac{d}{ds}\norm{v(s)}_{c(s)}^{2}\right|\leq O(1)+3\step^{2}R_{3}.
\]
Putting it into (\ref{eq:second_term_TV-1}) and (\ref{eq:TV_term-1}),
we have that
\[
\int_{V_{x}}\left|\frac{d}{ds}\widetilde{p}(v_{c(s)})\right|dv_{x}=O\left(\step^{2}R_{2}+\frac{1}{\step}+\step R_{3}\right).
\]

Putting this into \ref{eq:d_TV_integrate-2}, we get
\[
d_{\mathrm{TV}}(p_{x},p_{z})=O\left(\step^{2}R_{2}+\frac{1}{\step}+\step R_{3}\right)L+\frac{\widetilde{\ell}}{50}
\]
for any $L<\frac{\ell_{0}}{7\ell_{1}}$. By taking a minimal length
geodesic, and summing over segment of length $\frac{\ell_{0}}{8\ell_{1}}$,
for any $x$ and $z$, we have 
\[
d_{\mathrm{TV}}(p_{x},p_{z})=O\left(\step^{2}R_{2}+\frac{1}{\step}+\step R_{3}\right)d(x,z)+\frac{1}{25}.
\]
\end{proof}
\begin{defn}
\label{def:R2}Given a Hamiltonian curve $\gamma(t)$ with $\ell(\gamma_{0})\leq\ell_{0}$.
Let $R_{2}$ be a constant depending on the manifold $M$ and the
step size $\step$ such that for any $0\leq t\leq\step$, any curve
$c(s)$ starting from $\gamma(t)$ and any vector field $v(s)$ on
$c(s)$ with $v(0)=\gamma'(t)$, we have that

\[
\left|\frac{d}{ds}\left.\tr\Phi(v(s))\right|_{s=0}\right|\leq\left(\norm{\left.\frac{dc}{ds}\right|_{s=0}}_{\gamma(t)}+\step\norm{\left.D_{s}v\right|_{s=0}}_{\gamma(t)}\right)R_{2}.
\]
\end{defn}

\begin{lem}
\label{lem:smoothness_logdetJ}For $\step^{2}\leq\frac{1}{R_{1}}$
and $\ell(\gamma_{s})\leq\ell_{0}$, we have 
\[
\left|\int_{0}^{\step}\frac{t(\step-t)}{\step}\frac{d}{ds}\tr\Phi(\gamma'_{s}(t))dt\right|\leq O\left(\step^{2}R_{2}\right)
\]
where $\gamma_{s}$ is a family of Hamiltonian curve that connect
$c(s)$ to $y$ defined in Lemma \ref{lem:one_one_cor}.
\end{lem}

\begin{proof}
By Definition \ref{def:R2}, we have that
\begin{equation}
\left|\frac{d}{ds}\tr\Phi(\gamma'_{s}(t))\right|\leq\left(\norm{\left.\frac{\partial}{\partial s}\right|_{s=0}\gamma_{s}(t)}_{\gamma(t)}+\step\norm{\left.D_{s}\gamma'_{s}(t)\right|_{s=0}}_{\gamma(t)}\right)\cdot R_{2}.\label{eq:dRicus}
\end{equation}
By definition of $\gamma_{s}$, we have that $\frac{d}{ds}\gamma_{s}(0)=\frac{d}{ds}\ham_{c(s)}(0)=\frac{d}{ds}c(s)$
is a unit vector and hence $\norm{\frac{d}{ds}\gamma_{s}(0)}_{\gamma(0)}=1$.
Since $\ham_{c(s),\delta}(v(s))=y$, Lemma \ref{lem:smoothness_dx}
shows that
\[
\norm{\left.\frac{\partial}{\partial s}\right|_{s=0}\gamma_{s}(t)}_{\gamma(t)}\leq5\quad\text{and}\quad\norm{\left.D_{s}\gamma'_{s}(t)\right|_{s=0}}_{\gamma(t)}\leq\frac{10}{\step}.
\]
Therefore, we have that 
\[
\left|\frac{d}{ds}\tr\Phi(\gamma'_{s}(t))\right|\leq15R_{2}.
\]
\end{proof}

\pagebreak{}
\section{Gibbs sampling on manifolds\label{sec:KLS}}

\subsection{Isoperimetry for Hessian manifolds}

Here we derive a general isoperimetry bound, assuming that the manifold
is defined by the Hessian of a convex function, and that the directional
fourth directive is non-negative, a property satisfied, e.g., by the
standard logarithmic barrier.
\begin{lem}
\label{lem:1-d-gaussian-manifold-isoperimetry}Let $\phi:[a,b]\rightarrow\R$
be a convex function such that $\phi''$ is also convex. For any $x\in[a,b]$,
we have that
\[
e^{-\phi(x)}\geq0.372\sqrt{\phi''(x)}\min\left(\int_{x}^{b}e^{-\phi(t)}dt,\int_{a}^{x}e^{-\phi(t)}dt\right).
\]
Let $f(x)$ be the logconcave density proportional to $e^{-\phi(x)}$.
Then, we have that
\[
\Var_{X\sim f(x)}X\leq\frac{O(1)}{\min_{x\in[a,b]}\phi''(x)}.
\]
\end{lem}

\begin{proof}
Case 1) $\left|\phi'(x)\right|\geq a\sqrt{\phi''(x)}$ (we will pick
$a$ at the end). Without loss of generality, we have that $\phi'(x)\geq a\sqrt{\phi''(x)}$.
Since $\phi$ is convex, we have that 
\[
\phi(t)\geq\phi(x)+\phi'(x)(t-x).
\]
Therefore, we have that
\begin{align*}
\int_{x}^{b}e^{-\phi(t)}dt & \leq\int_{x}^{\infty}e^{-\phi(x)-\phi'(x)(t-x)}dt\\
 & =\frac{e^{-\phi(x)}}{\phi'(x)}\leq\frac{e^{-\phi(x)}}{a\sqrt{\phi''(x)}}.
\end{align*}
Case 2) $\left|\phi'(x)\right|\leq a\sqrt{\phi''(x)}$. Without loss
of generality, we have that $\phi'''(x)\geq0$. Since $\phi''$ is
convex, we have that
\[
\phi''(t)\geq\phi''(x)+\phi'''(x)(t-x)\geq\phi''(x)
\]
for $t\geq x$. Hence, we have that
\begin{align*}
\phi(t) & =\phi(x)+\phi'(x)(t-x)+\int_{x}^{t}(t-s)\phi''(t)ds\\
 & \geq\phi(x)+\phi'(x)(t-x)+\frac{1}{2}\phi''(x)(t-x)^{2}
\end{align*}
for all $t\geq x$. Therefore, we have that 
\begin{align*}
\int_{x}^{b}e^{-\phi(t)}dt & \leq\int_{-\infty}^{\infty}e^{-\phi(x)-\phi'(x)(t-x)-\frac{1}{2}\phi''(x)(t-x)^{2}}ds\\
 & =e^{-\phi(x)}\int_{-\infty}^{\infty}e^{-\phi'(x)s-\frac{1}{2}\phi''(x)s^{2}}ds\\
 & =\sqrt{\frac{2\pi}{\phi''(x)}}e^{-\phi(x)+\frac{1}{2}\frac{\phi'(x)^{2}}{\phi''(x)}}\\
 & \leq\sqrt{\frac{2\pi}{\phi''(x)}}e^{-\phi(x)+\frac{a^{2}}{2}}.
\end{align*}
Hence, we have 
\[
\int_{x}^{b}e^{-\phi(t)}dt\leq\sqrt{2\pi}e^{\frac{a^{2}}{2}}\frac{e^{-\phi(x)}}{\sqrt{\phi''(x)}}.
\]
Combining both cases, the isoperimetric ratio is $\min\left\{ a,\frac{e^{-\frac{a^{2}}{2}}}{\sqrt{2\pi}}\right\} $.
Setting $a$ to be the solution of $ae^{\frac{a^{2}}{2}}=\frac{1}{\sqrt{2\pi}}$,
this minimum is achieved at $\sqrt{W(1/2\pi)}>0.372$ where $W$ is
the inverse Lambert function, i.e., $W(x)e^{W(x)}=\frac{1}{2\pi}$.
This proves the first result.

The variance of $f$ follows from the fact that $f$ is logconcave.
\end{proof}
This generalizes to higher dimension with no dependence on the dimension
using localization, which we review next. Define an exponential needle
$E=(a,b,\gamma)$ as a segment $[a,b]\subseteq\R^{n}$ and $\gamma\in\R$
corresponding to the weight function $e^{\gamma t}$ applied the segment
$[a,b]$. The integral of an $n$-dimensional function $h:\R^{n}\rightarrow\R$
over this one dimensional needle is 
\[
\int_{E}h=\int_{0}^{|b-a|}h(a+tu)e^{\gamma t}\,dt\qquad\qquad\mbox{ where }\quad u=\frac{b-a}{|b-a|}.
\]
\begin{thm}[Theorem~2.7 in \cite{KLS95}]
 \label{thm:exp-needles} Let $f_{1},f_{2},f_{3},f_{3}$ be four
nonnegative continuous functions defined on $\R^{n}$, and $c_{1},c_{2}>0$.
Then, the following are equivalent: 
\begin{enumerate}
\item For every logconcave function $F$ defined on $\R^{n}$ with compact
support, 
\[
\left(\int_{\R^{n}}F(x)f_{1}(x)\,dx\right)^{c_{1}}\left(\int_{\R^{n}}F(x)f_{2}(x)\,dx\right)^{c_{2}}\le\left(\int_{\R^{n}}F(x)f_{3}(x)\,dt\right)^{c_{1}}\left(\int_{\R^{n}}F(x)f_{4}(x)\,dx\right)^{c_{2}}
\]
\item For every exponential needle $E$,
\[
\left(\int_{E}f_{1}\right)^{c_{1}}\left(\int_{E}f_{2}\right)^{c_{2}}\le\left(\int_{E}f_{3}\right)^{c_{1}}\left(\int_{E}f_{4}\right)^{c_{2}}
\]
\end{enumerate}
\end{thm}

\begin{lem}
\label{lem:n-d-gaussian-manifold-isoperimetry}Let $\phi:K\subset\Rn\rightarrow\R$
be a convex function defined over a convex body $K$ such that $D^{4}\phi(x)[h,h,h,h]\geq0$
for all $x\in K$ and $h\in\Rn$. Given any partition $S_{1},S_{2},S_{3}$
of $K$ with $d=\min_{x\in S_{1},y\in S_{2}}d(x,y)$, i.e., the minimum
distance between $S_{1}$ and $S_{2}$ in the Riemannian metric induced
by $\phi$. For any $\alpha>0$, we have that

\[
\frac{\int_{S_{3}}e^{-\alpha\phi(x)}dx}{\min\left\{ \int_{S_{1}}e^{-\alpha\phi(x)}dx,\int_{S_{2}}e^{-\alpha\phi(x)}\,dx\right\} }=\Omega(\sqrt{\alpha}\cdot d)
\]
\end{lem}

\begin{proof}
By rescaling $\phi$, we can assume $\alpha=1$. We write the desired
inequality as follows, for a constant $C$, with $\chi_{S}$ being
the indicator of set $S$: 
\[
Cd\int_{\R^{n}}e^{-\phi(x)}\chi_{S_{1}}(x)\,dx\cdot\int_{\R^{n}}e^{-\phi(x)}\chi_{S_{2}}(x)\,dx\leq\int_{\R^{n}}e^{-\phi(x)}\,dx\cdot\int_{\R^{n}}e^{-\phi(x)}\chi_{S_{3}}(x)\,dx.
\]
Using the localization lemma for exponential needles (Theorem \ref{thm:exp-needles}),
with $f_{i}(x)$ being $Cde^{-\phi(x)}\chi_{S_{1}}(x)$, $e^{-\phi(x)}\chi_{S_{2}}(x)$,
$e^{-\phi(x)}$ and $e^{-\phi(x)}\chi_{S_{3}}(x)$ with respectively,
it suffices to prove the following one-dimensional inequality for
functions $\phi$ defined on an interval and shifted by a linear term:
\begin{align*}
 & Cd\int_{0}^{1}e^{-\phi((1-t)a+tb)}e^{-ct}\chi_{S_{1}}((1-t)a+tb)\,dt\cdot\int_{0}^{1}e^{-\phi((1-t)a+tb)}e^{-ct}\chi_{S_{2}}((1-t)a+tb)\,dt\\
\leq & \int_{0}^{1}e^{-\phi((1-t)a+tb)}e^{-ct}\,dt\cdot\int_{0}^{1}e^{-\phi((1-t)a+tb)}e^{-ct}\chi_{S_{3}}((1-t)a+tb)\,dt.
\end{align*}

Each $T_{i}=\left\{ t:(1-t)a+tb\in S_{i}\right\} $ is a union of
intervals. By a standard argument (see \cite{LS93}), it suffices
to consider the case when each $S_{i}$ is a single interval and add
up over all intervals in $S_{3}$. Thus it suffices to prove the statement
in one dimension for all convex $\phi$ with convex $\phi''$. In
one-dimension, we have  
\[
d(x,y)=\int_{x}^{y}\sqrt{\phi''(t)}\,dt.
\]
Taking $T_{3}=[a',b']\subset[a,b]$, the inequality we need to prove
is that for any convex $\phi$ with convex $\phi''$ is

\[
\frac{\int_{a'}^{b'}e^{-\phi(t)}\,dt}{\int_{a'}^{b'}\sqrt{\phi''(t)}\,dt}\ge\Omega(1)\frac{\int_{a}^{a'}e^{-\phi(t)}\:dt\int_{b'}^{b}e^{-\phi(t)}\:dt}{\int_{a}^{b}e^{-\phi(t)}\,dt}
\]
which is implied by noting that $\frac{\int_{a'}^{b'}e^{-\phi(t)}\,dt}{\int_{a'}^{b'}\sqrt{\phi''(t)}\,dt}\ge\min_{x\in[a',b']}\frac{e^{-\phi(x)}}{\sqrt{\phi''(x)}}$
and applying Lemma \ref{lem:1-d-gaussian-manifold-isoperimetry}. 
\end{proof}

\subsection{Sampling with the log barrier}

For any polytope $M=\{Ax>b\}$, the logarithmic barrier function $\phi(x)$
is defined as
\[
\phi(x)=-\sum_{i=1}^{m}\log(a_{i}^{T}x-b_{i}).
\]

\logbarrier*
\begin{proof}
Lemma \ref{lem:n-d-gaussian-manifold-isoperimetry} shows that the
isoperimetric coefficient $\psi$ is $\Omega(\sqrt{\alpha})$. Also,
we know that isoperimetric coefficient $\psi$ is at worst $\Omega(m^{-\frac{1}{2}})$
\cite{LeeV17geodesic}. Lemma \ref{lem:one-step} shows that the condition
of Theorem \ref{thm:gen-convergence} is satisfied, thus implying
the bound
\begin{align*}
 & \widetilde{O}\left(\max(\sqrt{\alpha},m^{-\frac{1}{2}})^{-2}\delta^{-2}\right)\\
= & \widetilde{O}\left(\max(\sqrt{\alpha},m^{-\frac{1}{2}})^{-2}\max\left(n^{\frac{2}{3}},\alpha^{\frac{2}{3}}m^{\frac{1}{3}}n^{\frac{1}{3}},\alpha m^{\frac{1}{2}}n^{\frac{1}{6}}\right)\right)\\
= & \widetilde{O}\left(\frac{n^{\frac{2}{3}}+\alpha^{\frac{2}{3}}m^{\frac{1}{3}}n^{\frac{1}{3}}+\alpha m^{\frac{1}{2}}n^{\frac{1}{6}}}{\alpha+m^{-1}}\right)\\
= & \widetilde{O}\left(\frac{n^{\frac{2}{3}}}{\alpha+m^{-1}}+\frac{m^{\frac{1}{3}}n^{\frac{1}{3}}}{\alpha^{\frac{1}{3}}+m^{-\frac{1}{3}}}+m^{\frac{1}{2}}n^{\frac{1}{6}}\right)
\end{align*}

To implement the walk, we solve this ODE using the collocation method
as described in \cite{LeeV16}. The key lemma used in that paper is
that the slack of the geodesic does not change by more than a constant
multiplicative factor up to the step size $\delta=O(\frac{1}{n^{1/4}})$.
Similarly in Lemma \ref{lem:geodesic_4}, we proved the the Hamiltonian
flow does not change by more than a constant multiplicative factor
up to the step size $\delta=O(\frac{1}{M_{1}^{1/4}})=O(\frac{n^{-1/4}}{1+\sqrt{\alpha}})$.
Since the step size we use is $\Theta(\frac{n^{-1/3}}{1+\sqrt{\alpha}})$,
we can apply the collocation method as described in \cite{LeeV16}
and obtain an algorithm to compute the Hamiltonian flow in $\widetilde{O}(mn^{\omega-1}\log^{O(1)}\frac{1}{\eta})$
time with additive error $\eta$ distance in the local norm. Due to
the exponential convergence, it suffices for the sampling purpose
with only a polylogarithmic overhead in the total running time.
\end{proof}

\section{Polytope volume computation: Gaussian cooling on manifolds \label{sec:cooling}}

The volume algorithm is essentially the Gaussian cooling algorithm
introduced in \cite{CV2015}. Here we apply it to a sequence of Gibbs
distributions rather than a sequence of Gaussians. More precisely,
for a convex body $K$ and a convex barrier function $\phi:K\rightarrow\R$,
we define

\[
f(\sigma^{2},x)=\begin{cases}
\e{-\sigma^{-2}\phi(x)} & \mbox{ if }x\in K\\
0 & \mbox{ otherwise}
\end{cases}
\]
and 
\[
F(\sigma^{2})=\int_{\R^{n}}f(\sigma^{2},x)\,dx
\]
where $x^{*}$ is the minimizer of $\phi$ (the center of $K$). Let
$\mu_{i}$ be the probability distribution proportional to $f(\sigma_{i}^{2},x)$
where $\sigma_{i}$ is the tempature of the Gibbs distribution to
be fixed. The algorithm estimates each ratio in the following telescoping
product: 

\[
e^{-\sigma^{-2}\phi(x^{*})}\vol(K)\approx F(\sigma_{k})=F\left(\sigma_{0}\right)\prod_{i=1}^{k}\frac{F(\sigma_{i+1}^{2})}{F(\sigma_{i}^{2})}
\]
for some large enough $\sigma_{k}$.

Let $x$ be a random sample point from $\mu_{i}$ and let $Y_{x}=f(\sigma_{i+1}^{2},x)/f(\sigma_{i}^{2},x)$.
Then, 
\begin{align*}
\E_{x\sim\mu_{i}}(Y_{x}) & =\frac{F(\sigma_{i+1}^{2})}{F(\sigma_{i}^{2})}.
\end{align*}

\subsection{Algorithm: cooling schedule}

\begin{algorithm2e}

\caption{Volume($M$, $\varepsilon$)}

\SetAlgoLined

Let $\sigma_{0}^{2}=\Theta(\varepsilon^{2}n^{-3}\log^{-3}(n/\varepsilon))$,
\[
\sigma_{i+1}^{2}=\begin{cases}
\sigma_{i}^{2}\left(1+\frac{1}{\sqrt{n}}\right) & \text{if }\vartheta\le n\sigma_{i}^{2}\\
\sigma_{i}^{2}\left(1+\min(\frac{\sigma_{i}}{\sqrt{\vartheta}},\frac{1}{2})\right) & \text{otherwise}
\end{cases}.
\]
and
\[
k_{i}=\begin{cases}
\Theta(\frac{\sqrt{n}}{\varepsilon^{2}}\log\frac{n}{\varepsilon}) & \text{if }\vartheta\le n\sigma_{i}^{2}\\
\Theta((\frac{\sqrt{\vartheta}}{\sigma_{i}}+1)\epsilon^{-2}\log\frac{n}{\varepsilon}) & \text{otherwise}
\end{cases}.
\]

Set $i=0$. Compute $x^{*}=\arg\min\phi(x)$.

Assume that $\phi(x^{*})=0$ by shifting the barrier function.

Sample $k_{0}$ points $\{X_{1},\ldots,X_{k_{0}}\}$ from the Gaussian
distribution centered at the minimizer $x^{*}$ of $\phi$ with covariance
$\sigma_{0}^{2}\left(\nabla^{2}\phi(x^{*})\right)^{-1}.$

\While{$\sigma_{i}^{2}\leq\Theta(1)\frac{\vartheta}{\varepsilon}\log\frac{n\vartheta}{\varepsilon}$}{

Sample $k_{i}$ points $\{X_{1},\ldots,X_{k}\}$ using \textbf{Hamiltonian
Monte Carlo} with target density $f(\sigma_{i}^{2},X)$ and the previous
$\{X_{1},\ldots,X_{k}\}$ as warm start.

Compute the ratio 
\[
W_{i+1}=\frac{1}{k_{i}}\cdot\sum_{j=1}^{k_{i}}\frac{f(\sigma_{i+1}^{2},X_{j})}{f(\sigma_{i}^{2},X_{j})}.
\]

Increment $i$.

}

\textbf{Output:} $(2\pi\sigma_{0}^{2})^{\frac{n}{2}}\det(\nabla^{2}\phi(x^{*}))^{-\frac{1}{2}}W_{1}\ldots W_{i}$
as the volume estimate. 

\end{algorithm2e}

\subsection{Correctness of the algorithm}

In this subsection, we prove the correctness of the algorithm. We
separate the proof into two parts. In the first part, we estimate
how small $\sigma_{0}$ we start with should be and how large $\sigma_{k}$
we end with should be. In the second part, we estimate the variance
of the estimator $Y_{x}$.

\subsubsection{Initial and terminal conditions }

First, we need the following lemmas about self-concordance functions
and logconcave functions:
\begin{lem}[{\cite[Thm 2.1.1]{nesterov1994interior}}]
\label{lem:self_concordant_quadratic}Let $\phi$ be a self-concordant
function and $x^{*}$ be its minimizer. For any $x$ such that $\phi(x)\leq\phi(x^{*})+r^{2}$
with $r\leq\frac{1}{2}$, we have that 
\[
\phi(x)=\phi(x^{*})+\frac{1\pm\Theta(r)}{2}(x-x^{*})^{T}\nabla^{2}\phi(x^{*})(x-x^{*}).
\]
\end{lem}

\begin{lem}[{\cite[Prop 2.3.2]{nesterov1994interior}}]
\label{lem:phi_difference}For any $\vartheta$-self-concordance
barrier function $\phi$ on convex $K$, for any interior point $x$
and $y$ in $K$, we have that
\[
\phi(x)\leq\phi(y)+\vartheta\ln\frac{1}{1-\pi_{y}(x)}
\]
where $\pi_{y}(x)=\inf\{t\geq0|y+t^{-1}(x-y)\in K\}$. 
\end{lem}

\begin{lem}[{\cite[Lem 5.16]{LV07}}]
\label{lem:max_f}For any logconcave distribution $f$ on $\Rn$
and any $\beta\geq2$, we have
\[
\P_{x\sim f}(f(x)\leq e^{-\beta n}\max_{y}f(y))\leq e^{-O(\beta n)}.
\]
\end{lem}

\begin{lem}[Large and small $\sigma^{2}$]
\label{lem:boundary_case}Let $\phi$ be a $\vartheta$-self concordant
barrier function for $K$. If $\sigma n^{\frac{3}{2}}\log^{\frac{3}{2}}\frac{1}{\sigma}\leq1$
and $\sigma n^{\frac{1}{2}}\log^{\frac{1}{2}}\frac{1}{\sigma}\leq1$,
then we have that
\[
\left|\ln F\left(\sigma\right)-\ln\left(e^{-\sigma^{-2}\phi(x^{*})}(2\pi\sigma^{2})^{\frac{n}{2}}\det(\nabla^{2}\phi(x^{*}))^{-\frac{1}{2}}\right)\right|\leq O(\sigma n^{\frac{3}{2}}\log^{\frac{3}{2}}\frac{1}{\sigma}).
\]
If $\sigma^{2}\geq\vartheta$, then 
\[
\left|\ln F\left(\sigma\right)-\ln\left(e^{-\sigma^{-2}\phi(x^{*})}\vol(K)\right)\right|\leq O(\sigma^{-2}\vartheta\ln(\sigma^{2}n/\vartheta)).
\]
\end{lem}

\begin{proof}
We begin with the first inequality. Let $S$ be the set of $x$ such
that $\phi(x)\leq\phi(x^{*})+\beta\sigma^{2}n$ for some $\beta\geq2$
to be determined. Therefore, 
\[
f(\sigma^{2},x)\geq e^{-\beta n}f(\sigma^{2},x^{*})
\]
for all $x\in S$. Since $f(\sigma^{2},x)$ is a logconcave function,
Lemma \ref{lem:max_f} shows that $\P_{x\sim f}(S)\leq e^{-O(\beta n)}$.
Therefore, 
\[
\int_{S}f(\sigma^{2},x)\,dx\leq\int f(\sigma^{2},x)\,dx\leq(1+e^{-O(\beta n)})\int_{S}f(\sigma^{2},x)\,dx.
\]
In short, we have
\[
F(\sigma)=(1\pm e^{-O(\beta n)})\int_{S}f(\sigma^{2},x)\,dx.
\]
By the construction of $S$, Lemma \ref{lem:self_concordant_quadratic}
and the fact that $\phi$ is self-concordant, if $\beta\sigma^{2}n\leq\frac{1}{4}$,
we have that
\[
\phi(x)=\phi(x^{*})+\frac{1\pm\Theta(\sigma\sqrt{\beta n})}{2}(x-x^{*})^{T}\nabla^{2}\phi(x^{*})(x-x^{*})
\]
for all $x\in S$. Hence, 
\begin{align*}
F(\sigma) & =(1\pm e^{-O(\beta n)})e^{-\sigma^{-2}\phi(x^{*})}\int_{S}e^{-(1\pm O(\sigma\sqrt{\beta n}))\frac{\sigma^{-2}}{2}(x-x^{*})^{T}\nabla^{2}\phi(x^{*})(x-x^{*})}dx.
\end{align*}
Now we note that 
\[
\int_{S^{c}}e^{-\Theta(\frac{\sigma^{-2}}{2}(x-x^{*})^{T}\nabla^{2}\phi(x^{*})(x-x^{*}))}dx=e^{-O(\beta n)}\int_{S}e^{-\frac{\sigma^{-2}}{2}(x-x^{*})^{T}\nabla^{2}\phi(x^{*})(x-x^{*})}dx
\]
because 
\[
\frac{\sigma^{-2}}{2}(x-x^{*})^{T}\nabla^{2}\phi(x^{*})(x-x^{*})=\Omega(\beta n)
\]
 outside $S$. Therefore, 
\begin{align*}
F(\sigma) & =(1\pm e^{-O(\beta n)})e^{-\sigma^{-2}\phi(x^{*})}\int_{\Rn}e^{-(1\pm O(\sigma\sqrt{\beta n}))\frac{\sigma^{-2}}{2}(x-x^{*})^{T}\nabla^{2}\phi(x^{*})(x-x^{*})}dx\\
 & =(1\pm e^{-O(\beta n)}\pm O(\sigma(\beta n)^{\frac{3}{2}}))e^{-\sigma^{-2}\phi(x^{*})}\int_{\Rn}e^{-\frac{\sigma^{-2}}{2}(x-x^{*})^{T}\nabla^{2}\phi(x^{*})(x-x^{*})}dx\\
 & =(1\pm e^{-O(\beta n)}\pm O(\sigma(\beta n)^{\frac{3}{2}}))e^{-\sigma^{-2}\phi(x^{*})}(2\pi\sigma^{2})^{\frac{n}{2}}\det(\nabla^{2}\phi(x^{*}))^{-\frac{1}{2}}
\end{align*}
where we used that $\sigma(\beta n)^{\frac{1}{2}}=O(1)$ in first
sentence and $\sigma(\beta n)^{\frac{3}{2}}=O(1)$ in the second sentence.
Setting $\beta=\Theta(\log\frac{1}{\sigma})$, we get the first result.

For the second inequality, for any $0\leq t<1$ and any $x\in x^{*}+t(K-x^{*})$,
we have that $\pi_{x^{*}}(x)\leq t$ ($\pi_{x^{*}}$ is defined in
Lemma \ref{lem:phi_difference}). Therefore, Lemma \ref{lem:phi_difference}
shows that
\[
\phi(x)\leq\phi(x^{*})+\vartheta\ln\frac{1}{1-t}.
\]
Note that $\P_{x\sim\mu}(x\in x^{*}+t(K-x^{*}))=t^{n}$ where $\mu$
is the uniform distribution in $K$. Therefore, for any $0<\beta<1$,
we have that
\[
\P_{x\sim\mu}(\phi(x)\leq\phi(x^{*})+\vartheta\ln\frac{1}{\beta}\})\geq(1-\beta)^{n}.
\]
Hence, 
\begin{align*}
\vol(K)\cdot e^{-\sigma^{-2}\phi(x^{*})} & \geq F(\sigma^{2})\\
 & \geq\vol(K)\cdot(1-\beta)^{n}\e{-\sigma^{-2}(\phi(x^{*})+\vartheta\ln\frac{1}{\beta})}.
\end{align*}
Setting $\beta=\sigma^{-2}n^{-1}\vartheta$, we get the second result.
\end{proof}

\subsubsection{Variance of $Y_{x}$}

Our goal is to estimate $\E_{x\sim\mu_{i}}(Y_{x})$ within a target
relative error. The algorithm estimates the quantity $\E_{x\sim\mu_{i}}(Y_{x})$
by taking random sample points $x_{1},\ldots,x_{k}$ and computing
the empirical estimate for $\E_{x\sim\mu_{i}}(Y_{x})$ from the corresponding
$Y_{x_{1}},\ldots,Y_{x_{k}}$. The variance of $Y_{x_{i}}$ divided
by its expectation squared will give a bound on how many independent
samples $x_{i}$ are needed to estimate $\E_{x\sim\mu_{i}}(Y_{x})$
within the target accuracy. We have 
\[
\E_{x\sim\mu_{i}}(Y_{x}^{2})=\frac{\int_{K}\e{\frac{\phi(x)}{\sigma_{i}^{2}}-\frac{2\phi(x)}{\sigma_{i+1}^{2}}}\,dx}{\int_{K}\e{-\frac{\phi(x)}{\sigma_{i}^{2}}}\,dx}=\frac{F(\frac{\sigma_{i+1}^{2}\sigma_{i}^{2}}{2\sigma_{i}^{2}-\sigma_{i+1}^{2}})}{F(\sigma_{i}^{2})}
\]
and 
\[
\frac{\E_{x\sim\mu_{i}}(Y_{x}^{2})}{\E_{x\sim\mu_{i}}(Y_{x})^{2}}=\frac{F(\sigma_{i}^{2})F(\frac{\sigma_{i+1}^{2}\sigma_{i}^{2}}{2\sigma_{i}^{2}-\sigma_{i+1}^{2}})}{F(\sigma_{i+1}^{2})^{2}}%.
\]

If we let $\sigma^{2}=\sigma_{i+1}^{2}$ and $\sigma_{i}^{2}=\sigma^{2}/(1+r)$,
then we can further simplify as 
\begin{equation}
\frac{\E_{x\sim\mu_{i}}(Y_{x}^{2})}{\E_{x\sim\mu_{i}}(Y_{x})^{2}}=\frac{F\left(\frac{\sigma^{2}}{1+r}\right)F\left(\frac{\sigma^{2}}{1-r}\right)}{F\left(\sigma^{2}\right)^{2}}.\label{eq:EY2_EY2}
\end{equation}
\begin{lem}
\label{lem:ln_F_formula}For any $1>r\geq0$, we have that
\[
\ln\left(\frac{F\left(\frac{\sigma^{2}}{1+r}\right)F\left(\frac{\sigma^{2}}{1-r}\right)}{F\left(\sigma^{2}\right)^{2}}\right)=\frac{1}{\sigma^{4}}\int_{0}^{r}\int_{1-t}^{1+t}\Var_{x\sim\mu_{s}}\phi(x)dsdt
\]
where $\mu_{s}$ be the probability distribution proportional to $f(\frac{\sigma^{2}}{s},x)$.
\end{lem}

\begin{proof}
Fix $\sigma^{2}$. Let $g(t)=\ln F(\frac{\sigma^{2}}{t})$. Then,
we have that
\begin{align}
\ln\left(\frac{F\left(\frac{\sigma^{2}}{1+r}\right)F\left(\frac{\sigma^{2}}{1-r}\right)}{F\left(\sigma^{2}\right)^{2}}\right) & =\int_{0}^{r}\frac{d}{dt}\ln\left(\frac{F\left(\frac{\sigma^{2}}{1+t}\right)F\left(\frac{\sigma^{2}}{1-t}\right)}{F\left(\sigma^{2}\right)^{2}}\right)dt\nonumber \\
 & =\int_{0}^{r}\frac{d}{dt}g(1+t)-\frac{d}{dt}g(1-t)dt\nonumber \\
 & =\int_{0}^{r}\int_{1-t}^{1+t}\frac{d^{2}}{ds^{2}}g(s)dsdt.\label{eq:ln_F_formula}
\end{align}
For $\frac{d^{2}}{ds^{2}}g(s)$, we have that
\begin{align*}
\frac{d^{2}}{ds^{2}}g(s) & =\frac{d^{2}}{ds^{2}}\ln\int_{K}\e{-\frac{s}{\sigma^{2}}\phi(x)}dx\\
 & =-\frac{1}{\sigma^{2}}\cdot\frac{d}{ds}\frac{\int_{K}\phi(x)\cdot\e{-\frac{s}{\sigma^{2}}\phi(x)}dx}{\int_{K}\e{-\frac{s}{\sigma^{2}}\phi(x)}dx}\\
 & =\left(\frac{1}{\sigma^{2}}\right)^{2}\left(\frac{\int_{K}\phi^{2}(x)\cdot\e{-\frac{s}{\sigma^{2}}\phi(x)}dx}{\int_{K}\e{-\frac{s}{\sigma^{2}}\phi(x)}dx}-\frac{\left(\int_{K}\phi(x)\cdot\e{-\frac{s}{\sigma^{2}}\phi(x)}dx\right)^{2}}{\left(\int_{K}\e{-\frac{s}{\sigma^{2}}\phi(x)}dx\right)^{2}}\right)\\
 & =\frac{1}{\sigma^{4}}\left(\E_{x\sim\mu_{s}}\phi^{2}(x)-\left(\E_{x\sim\mu_{s}}\phi(x)\right)^{2}\right)=\frac{1}{\sigma^{4}}\Var_{x\sim\mu_{s}}\phi(x).
\end{align*}
Putting it into (\ref{eq:ln_F_formula}), we have the result. 
\end{proof}
Now, we bound $\Var_{x\sim\mu_{s}}\phi(x)$. This can be viewed as
a manifold version of the thin shell or variance hypothesis estimate.
\begin{lem}[Thin shell estimates]
\label{lem:thin_shell}Let $\phi$ be a $\vartheta$-self concordant
barrier function for $K$. Then, we have that
\[
\Var_{x\sim\mu_{s}}\phi(x)=O\left(\frac{\sigma^{2}}{s}\vartheta\right).
\]
\end{lem}

\begin{proof}
Let $K_{t}\defeq\{x\in K\text{ such that }\phi(x)\leq t\}$ and $m$
be the number such that $\mu_{s}(K_{m})=\frac{1}{2}$. Let $K_{m,r}=\{x\text{ such that }d(x,y)\leq r\text{ and }y\in K_{m}\}$.
By repeatedly applying Lemma \ref{lem:n-d-gaussian-manifold-isoperimetry},
we have that
\[
\mu_{s}(K_{m,r})=1-e^{-\Omega(\frac{\sqrt{s}}{\sigma}r)}.
\]
By our assumption on $\phi$, for any $x$ and $y$, we have that
$\left|\phi(x)-\phi(y)\right|\leq\sqrt{\vartheta}d(x,y)$. Therefore,
for any $x\in K_{m,r}$, we have that $\phi(x)\leq m+\sqrt{\vartheta}r$.
Therefore, with probability at least $1-e^{-\Omega(\frac{\sqrt{s}}{\sigma}r)}$
in $\mu_{s}$, it follows that $\phi(x)\leq m+\sqrt{\vartheta}r$.
Similarly, $\phi(x)\ge m-\sqrt{\vartheta}r$. Hence, with $1-e^{-\Omega(\frac{\sqrt{s}}{\sigma}r)}$
probability in $\mu_{s}$, we have that $\left|\phi(x)-m\right|\leq\sqrt{\vartheta}r$.
The bound on the variance follows.
\end{proof}
Now we are ready to prove the key lemma.
\begin{lem}
\label{lem:cooling_step}Let $\phi$ be a $\vartheta$-self concordant
barrier function for $K$. For any $\frac{1}{2}>r\geq0$, we have
that
\[
\frac{\E_{x\sim\mu_{i}}(Y_{x}^{2})}{\E_{x\sim\mu_{i}}(Y_{x})^{2}}=O\left(r^{2}\right)\min\left(\frac{\vartheta}{\sigma_{i}^{2}},n\right).
\]
\end{lem}

\begin{proof}
Using Lemma \ref{lem:thin_shell} and Lemma \ref{lem:ln_F_formula},
we have that
\begin{equation}
\ln\left(\frac{F\left(\frac{\sigma^{2}}{1+r}\right)F\left(\frac{\sigma^{2}}{1-r}\right)}{F\left(\sigma^{2}\right)^{2}}\right)=O\left(\frac{r^{2}\vartheta}{\sigma^{2}}\right).\label{eq:ln_F_1}
\end{equation}
This bounds is useful when $\sigma^{2}$ is large. 

For the case $\sigma^{2}$ is small, we recall that for any logconcave
function $f$, the function $a\rightarrow a^{n}\int f(x)^{a}\,dx$
is logconcave (Lemma 3.2 in \cite{KV06}). In particular, this shows
that $a^{n}F\left(\frac{1}{a}\right)$ is logconcave in $a$. Therefore,
with $a=\frac{1+r}{\sigma^{2}},\frac{1}{\sigma^{2}},\frac{1-r}{\sigma^{2}},$we
have 
\[
\frac{1}{\sigma^{4n}}F(\sigma^{2})^{2}\geq\left(\frac{1+r}{\sigma^{2}}\right)^{n}F\left(\frac{\sigma^{2}}{1+r}\right)\left(\frac{1-r}{\sigma^{2}}\right)^{n}F\left(\frac{\sigma^{2}}{1-r}\right).
\]
Rearranging the term, we have that 
\[
\frac{F\left(\frac{\sigma^{2}}{1+r}\right)F\left(\frac{\sigma^{2}}{1-r}\right)}{F\left(\sigma^{2}\right)^{2}}\leq\left(\frac{1}{(1+r)(1-r)}\right)^{n}.
\]
Therefore, we have that
\begin{equation}
\ln\left(\frac{F\left(\frac{\sigma^{2}}{1+r}\right)F\left(\frac{\sigma^{2}}{1-r}\right)}{F\left(\sigma^{2}\right)^{2}}\right)=O\left(nr^{2}\right).\label{eq:ln_F_2}
\end{equation}
Combining (\ref{eq:EY2_EY2}), (\ref{eq:ln_F_1}) and (\ref{eq:ln_F_2}),
we have the result.
\end{proof}

\subsubsection{Main lemma}
\begin{lem}
Given any $\vartheta$-self-concordance barrier $\phi$ on a convex
set $K$ and $0<\varepsilon<\frac{1}{2}$, the algorithm Volume($M$,
$\varepsilon$) outputs the volume of $K$ to within a $1\pm\varepsilon$
multiplicative factor.
\end{lem}

\begin{proof}
By our choice of $\varepsilon$, Lemma \ref{lem:boundary_case} shows
that $e^{-\sigma_{0}^{-2}\phi(x^{*})}(2\pi\sigma_{0}^{2})^{\frac{n}{2}}\det(\nabla^{2}\phi(x^{*}))^{-\frac{1}{2}}$
is an $1\pm\frac{\varepsilon}{4}$ multiplicative approximation of
$F(\sigma_{0})$ and that $e^{-\sigma_{k}^{-2}\phi(x^{*})}\vol(K)$
is a $1\pm\frac{\varepsilon}{4}$ multiplicative approximation of
$F(\sigma_{k})$. Note that we shifted the function $\phi$ such that
$\phi(x^{*})=0$. Therefore,
\[
\vol(K)=(1\pm\frac{\varepsilon}{2})(2\pi\sigma_{0}^{2})^{\frac{n}{2}}\det(\nabla^{2}\phi(x^{*}))^{-\frac{1}{2}}\prod_{i=1}^{k}\frac{F(\sigma_{i+1}^{2})}{F(\sigma_{i}^{2})}.
\]
In Lemma \ref{lem:cooling_step}, we showed that the variance of the
estimator $Y=f(\sigma_{i+1}^{2},X)/f(\sigma_{i}^{2},X)$ is upper
bounded by $O(1)(\E Y)^{2}$. Note that the algorithm takes $O(\sqrt{n})$
iterations to double $\sigma_{i}$ if $\sqrt{\frac{\vartheta}{n}}\le\sigma_{i}$
and $O(\sqrt{\vartheta}\sigma_{i}^{-1})$ iterations otherwise. By
a simple analysis of variance, to have relative error $\varepsilon$,
it suffices to have $\widetilde{O}(k_{i})$ samples in each phase.
\end{proof}

\subsection{Volume computation with the log barrier}

In this section, we prove the Theorem \ref{thm:volume}, restated
below for convenience.

\volume*
\begin{proof}
In the first part, when $\sigma^{2}\leq\frac{m}{n}$, the mixing time
of HMC is $\widetilde{O}(m\cdot n^{-\frac{1}{3}})$. Since the number
of sampling phases to double such $\sigma^{2}$ is $O(\sqrt{n})$
and since we samples $\widetilde{O}(\frac{\sqrt{n}}{\epsilon^{2}})$,
the total number of steps of HMC is 
\[
\widetilde{O}\left(mn^{-\frac{1}{3}}\right)\times O\left(\sqrt{n}\right)\times\widetilde{O}\left(\frac{\sqrt{n}}{\epsilon^{2}}\right)=\widetilde{O}\left(\frac{m\cdot n^{\frac{2}{3}}}{\epsilon^{2}}\right).
\]

In the second part, when $\sigma^{2}\geq\frac{m}{n}$, the mixing
time of HMC is $\widetilde{O}(\frac{n^{\frac{2}{3}}}{\sigma^{-2}+m^{-1}}+\frac{m^{\frac{1}{3}}n^{\frac{1}{3}}}{\sigma^{-\frac{2}{3}}+m^{-\frac{1}{3}}}+m^{\frac{1}{2}}n^{\frac{1}{6}})$.
Since the number of sampling phases to double $\sigma^{2}$ is $O(1+\frac{\sqrt{m}}{\sigma})$
and since we sample $\widetilde{O}((\frac{\sqrt{m}}{\sigma}+1)\varepsilon^{-2})$
in each phase, the total number of steps of HMC is 
\[
\widetilde{O}\left(\frac{n^{\frac{2}{3}}}{\sigma^{-2}+m^{-1}}+\frac{m^{\frac{1}{3}}n^{\frac{1}{3}}}{\sigma^{-\frac{2}{3}}+m^{-\frac{1}{3}}}+m^{\frac{1}{2}}n^{\frac{1}{6}}\right)\times O\left(1+\frac{\sqrt{m}}{\sigma}\right)\times\widetilde{O}\left((\frac{\sqrt{m}}{\sigma}+1)\varepsilon^{-2}\right)=\widetilde{O}\left(\frac{m\cdot n^{\frac{2}{3}}}{\epsilon^{2}}\right).
\]

Combining both parts, the total number of steps of HMC is 
\[
\widetilde{O}\left(\frac{m\cdot n^{\frac{2}{3}}}{\epsilon^{2}}\right).
\]
\end{proof}

\pagebreak{}
\section{Logarithmic barrier\label{sec:Logarithmic}}

\label{sec:Logarithmic-Barrier}For any polytope $\mathcal{M}=\{Ax>b\}$,
the logarithmic barrier function $\phi(x)$ is defined as
\[
\phi(x)=-\sum_{i=1}^{m}\log(a_{i}^{T}x-b_{i}).
\]
We denote the manifold induced by the logarithmic barrier on $\mathcal{M}$
by $\mathcal{M}_{L}$. The goal of this section is to analyze Hamiltonian
Monte Carlo on $\mathcal{M}_{L}$. In Section \ref{subsec:RG_L},
we give explicit formulas for various Riemannian geometry concepts
on $\mathcal{M}_{L}$. In Section \ref{subsec:geodesic_walk_log},
we describe the HMC specialized to $\mathcal{M}_{L}$. In Sections
\ref{sec:walk_is_random_log} to \ref{subsec:Stability-of-norm},
we bound the parameters required by Theorem \ref{thm:gen-convergence},
resulting in Theorem \ref{thm:logbarrier}.

The following parameters that are associated with barrier functions
will be convenient. 
\begin{defn}
\label{def:self_concordance}For a convex function $f$, let $M_{1}$,
$M_{2}$ and $M_{3}$ be the smallest numbers such that
\begin{enumerate}
\item $M_{1}\geq\max_{x\in\mathcal{M}}(\nabla f(x))^{T}\left(A_{x}^{T}A_{x}\right)^{-1}\nabla f(x)$
and $M_{1}\geq n$.
\item $\nabla^{2}f\preceq M_{2}\cdot A_{x}^{T}A_{x}$ .
\item $\left|\tr((A_{x}^{T}A_{x})^{-1}\nabla^{3}f(x)[v])\right|\leq M_{3}\norm v_{x}$
for all $v$.
\end{enumerate}
\end{defn}

For the case $f=\phi$ are the standard logarithmic barrier, these
parameters are $n,1,\sqrt{n}$ respectively.

\subsection{Riemannian geometry on $\mathcal{M}_{L}$ ($G_{2}$)}

\label{subsec:RG_L}We use the following definitions throughout this
section.
\begin{defn}
\label{def:notation}For any matrix $A\in\R^{m\times n}$ and vectors
$b\in\R^{m}$ and $x\in\R^{n}$, define
\begin{enumerate}
\item $s_{x}=Ax-b,\,S_{x}=\Diag(s_{x}),\,A_{x}=S_{x}^{-1}A$.
\item $s_{x,v}=A_{x}v$, $S_{x,v}=\Diag(A_{x}v)$.
\item $P_{x}=A_{x}(A_{x}^{T}A_{x})^{-1}A_{x}^{T}$, $\sigma_{x}=\diag(P_{x})$,
$\Sigma_{x}=\Diag(P)$, $\left(P_{x}^{(2)}\right)_{ij}=\left(P_{x}\right)_{ij}^{2}$.
\item Gradient of $\phi$: $\phi_{i}=-\sum_{\ell}\left(e_{\ell}^{T}A_{x}e_{i}\right)$.
\item Hessian of $\phi$ and its inverse: $g_{ij}=\phi_{ij}=\left(A_{x}^{T}A_{x}\right)_{ij}=\sum\left(e_{\ell}^{T}A_{x}e_{i}\right)\left(e_{\ell}^{T}A_{x}e_{j}\right)$,
$g^{ij}=e_{i}^{T}\left(A_{x}^{T}A_{x}\right)^{-1}e_{j}$.
\item Third derivatives of $\phi$: $\phi_{ijk}=-2\sum_{\ell}\left(e_{\ell}^{T}A_{x}e_{i}\right)\left(e_{\ell}^{T}A_{x}e_{j}\right)\left(e_{\ell}^{T}A_{x}e_{k}\right)$.
\item For brevity (overloading notation), we define $s_{\gamma'}=s_{\gamma,\gamma'}$,
$s_{\gamma''}=s_{\gamma,\gamma''}$ , $S_{\gamma'}=S_{\gamma,\gamma'}$
and $S_{\gamma''}=S_{\gamma,\gamma''}$ for a curve $\gamma(t)$.
\end{enumerate}
\end{defn}

In this section, we will frequently use the following identities derived
from elementary calculus (using only the chain/product rules and the
formula for derivative of the inverse of a matrix). For reference,
we include proofs in Appendix \ref{sec:Calculus}.
\begin{fact}
\label{fact:calculus}For any matrix $A$ and any curve $\gamma(t)$,
we have
\begin{align*}
\frac{dA_{\gamma}}{dt} & =-S_{\gamma'}A_{\gamma},\\
\frac{dP_{\gamma}}{dt} & =-S_{\gamma'}P_{\gamma}-P_{\gamma}S_{\gamma'}+2P_{\gamma}S_{\gamma'}P_{\gamma},\\
\frac{dS_{\gamma'}}{dt} & =\Diag(-S_{\gamma'}A_{\gamma}\gamma'+A_{\gamma}\gamma'')=-S_{\gamma'}^{2}+S_{\gamma''},
\end{align*}
\end{fact}

We also use these matrix inequalities: $\tr(AB)=\tr(BA)$, $\tr(PAP)\le\tr(A)$
for any psd matrix $A$; $\tr(ABA^{T})\le\tr(AZA^{T})$ for any $B\preceq Z$;
the Cauchy-Schwartz, namely, $\tr(AB)\le\tr(AA^{T})^{\frac{1}{2}}\tr(BB^{T})^{\frac{1}{2}}.$
We note $P_{x}^{2}=P_{x}$ because $P_{x}$ is a projection matrix.

Since the manifold $\mathcal{M}_{L}$ is naturally embedded in $\Rn$,
we can identify $T_{x}\mathcal{M}_{L}$ with Euclidean coordinates.
We have that
\[
\left\langle u,v\right\rangle _{x}=u^{T}\nabla^{2}\phi(x)v=u^{T}A_{x}^{T}A_{x}v.
\]
We will use the following two lemmas proved in \cite{LeeV16}.
\begin{lem}
\label{lem:geo_equ_log}Let $w(t)$ be a vector field defined on a
curve $z(t)$ in $\mathcal{M}_{L}$. Then, 
\[
\nabla_{z'}w=\frac{dw}{dt}-\left(A_{z}^{T}A_{z}\right)^{-1}A_{z}^{T}S_{z'}s_{z,w}=\frac{dw}{dt}-\left(A_{z}^{T}A_{z}\right)^{-1}A_{z}^{T}S_{z,w}s_{z'}.
\]
In particular, the equation for parallel transport on a curve $\gamma(t)$
is given by
\begin{equation}
\frac{d}{dt}v(t)=\left(A_{\gamma}^{T}A_{\gamma}\right)^{-1}A_{\gamma}^{T}S_{\gamma'}A_{\gamma}v.\label{eq:parallel_log}
\end{equation}
\end{lem}

\begin{lem}
\label{lem:Riemann_tensor_log}Given $u,v,w,x\in T_{x}\mathcal{M}_{L}$,
the Riemann Curvature Tensor at $x$ is given by
\begin{eqnarray*}
R(u,v)w & = & \left(A_{x}^{T}A_{x}\right)^{-1}A_{x}^{T}\left(S_{x,v}P_{x}S_{x,w}-\Diag(P_{x}s_{x,v}s_{x,w})\right)A_{x}u
\end{eqnarray*}
and the Ricci curvature $\Ric(u)\defeq\tr R(u,u)$ is given by
\begin{eqnarray*}
\Ric(u) & = & s_{x,u}^{T}P_{x}^{(2)}s_{x,u}-\sigma_{x}^{T}P_{x}s_{x,u}^{2}
\end{eqnarray*}
where $R(u,u)$ is the operator defined above.
\end{lem}

\subsection{Hamiltonian walk on $\mathcal{M}_{L}$}

\label{subsec:geodesic_walk_log}We often work in Euclidean coordinates.
In this case, the Hamiltonian walk is given by the formula in the
next lemma. To implement the walk, we solve this ODE using the collocation
method as described in \cite{LeeV16}, after first reducing it to
a first-order ODE. The resulting complexity is $\tilde{O}(mn^{\omega-1})$
per step.
\begin{lem}
\label{lem:geodesic_walk_log_def}The Hamiltonian curve at a point
$x$ in Euclidean coordinates is given by the following equations
\begin{eqnarray*}
\gamma''(t) & = & \left(A_{\gamma}^{T}A_{\gamma}\right)^{-1}A_{\gamma}^{T}s_{\gamma'}^{2}+\mu(\gamma(t))\quad\forall t\geq0\\
\gamma'(0) & = & w,\\
\gamma(0) & = & x.
\end{eqnarray*}
where $\mu(x)=\left(A_{x}^{T}A_{x}\right)^{-1}A_{x}^{T}\sigma_{x}-(A_{x}^{T}A_{x})^{-1}\nabla f(x)$
and $w\sim N(0,(A_{\gamma}^{T}A_{\gamma})^{-1})$.
\end{lem}

\begin{proof}
Recall from Lemma \ref{lem:RMCMC} that the Hamiltonian walk is given
by
\begin{align*}
D_{t}\frac{d\gamma}{dt}= & \mu(\gamma(t)),\\
\frac{d\gamma}{dt}(0)\sim & N(0,g(x)^{-1})
\end{align*}
where $\mu(x)=-g(x)^{-1}\nabla f(x)-\frac{1}{2}g(x)^{-1}\tr\left[g(x)^{-1}Dg(x)\right]$.
By Lemma \ref{lem:geo_equ_log}, applied with $w(t)=\gamma'(t)$,
$z(t)=\gamma(t)$, we have 
\[
D_{t}\frac{d\gamma}{dt}=\gamma''(t)-\left(A_{\gamma}^{T}A_{\gamma}\right)^{-1}A_{\gamma}^{T}s_{\gamma'}^{2}.
\]
For the formula of $\mu$, we note that
\begin{align*}
\frac{1}{2}\tr\left[g(x)^{-1}Dg(x)\right]_{k} & =\frac{1}{2}\sum_{ij}\left((A_{x}^{T}A_{x})^{-1}\right)_{ij}\frac{\partial}{\partial x_{k}}\left(A_{x}^{T}A_{x}\right)_{ji}\\
\mbox{by Defn.\ref{def:notation}(6) } & =-\sum_{ij}\left((A_{x}^{T}A_{x})^{-1}\right)_{ij}\sum_{\ell}\left(e_{\ell}^{T}A_{x}e_{i}\right)\left(e_{\ell}^{T}A_{x}e_{j}\right)\left(e_{\ell}^{T}A_{x}e_{k}\right)\\
 & =-\sum_{\ell}\left(A_{x}^{T}(A_{x}^{T}A_{x})^{-1}A_{x}\right)_{\ell\ell}\left(e_{\ell}^{T}A_{x}e_{k}\right)\\
 & =-A_{x}^{T}\sigma_{x}.
\end{align*}
Therefore, 
\[
\mu(x)=-(A_{x}^{T}A_{x})^{-1}\nabla f(x)+\left(A_{x}^{T}A_{x}\right)^{-1}A_{x}^{T}\sigma_{x}.
\]
\end{proof}
Many parameters for Hamiltonian walk depends on the operator $\Phi(t)$.
Here, we give a formula of $\Phi(t)$ in Euclidean coordinates.
\begin{lem}
\label{lem:formulate_M_R}Given a curve $\gamma(t)$, in Euclidean
coordinates, we have that
\[
\Phi(t)=M(t)-R(t)
\]
where
\begin{align*}
R(t) & =\left(A_{\gamma}^{T}A_{\gamma}\right)^{-1}\left(A_{\gamma}^{T}S_{\gamma'}P_{\gamma}S_{\gamma'}A_{\gamma}-A_{\gamma}^{T}\Diag(P_{\gamma}s_{\gamma'}^{2})A_{\gamma}\right)\\
M(t) & =\left(A_{\gamma}^{T}A_{\gamma}\right)^{-1}\left(A_{x}^{T}\left(S_{x,\mu}-3\Sigma_{x}+2P_{x}^{(2)}\right)A_{x}-\nabla^{2}f(x)\right).
\end{align*}
\end{lem}

\begin{proof}
Lemma \ref{lem:Riemann_tensor_log} with $v=w=\gamma'$, gives the
formula for $R(t)$. 

For $M(t)$, Lemma \ref{lem:geo_equ_log} with $w=\mu(x)$, $z'=u$
shows that
\begin{align*}
D_{u}\mu(x) & =\nabla_{u}\mu(x)-(A_{x}^{T}A_{x})^{-1}A_{x}^{T}S_{x,\mu}A_{x}u.
\end{align*}
For the first term $\nabla_{u}\mu(x)$, we note that
\begin{align*}
\nabla_{u}\mu(x)= & 2\left(A_{x}^{T}A_{x}\right)^{-1}A_{x}^{T}S_{x,u}A_{x}\left(A_{x}^{T}A_{x}\right)^{-1}A_{x}^{T}\sigma_{x}\\
 & -3\left(A_{x}^{T}A_{x}\right)^{-1}A_{x}^{T}S_{x,u}\sigma_{x}\\
 & +2\left(A_{x}^{T}A_{x}\right)^{-1}A_{x}^{T}\diag(P_{x}S_{x,u}P_{x})\\
 & -2\left(A_{x}^{T}A_{x}\right)^{-1}A_{x}^{T}S_{x,u}A_{x}\left(A_{x}^{T}A_{x}\right)^{-1}\nabla f(x)\\
 & -(A_{x}^{T}A_{x})^{-1}\nabla^{2}f(x)u.
\end{align*}
Therefore, we have that
\begin{align*}
 & D_{u}\mu(x)\\
= & 2\left(A_{x}^{T}A_{x}\right)^{-1}A_{x}^{T}\Diag(P_{x}\sigma_{x})A_{x}u-3\left(A_{x}^{T}A_{x}\right)^{-1}A_{x}^{T}\Sigma_{x}A_{x}u+2\left(A_{x}^{T}A_{x}\right)^{-1}A_{x}^{T}P_{x}^{(2)}A_{x}u\\
 & -2\left(A_{x}^{T}A_{x}\right)^{-1}A_{x}^{T}\Diag\left(A_{x}\left(A_{x}^{T}A_{x}\right)^{-1}\nabla f(x)\right)A_{x}u-(A_{x}^{T}A_{x})^{-1}\nabla^{2}f(x)u-(A_{x}^{T}A_{x})^{-1}A_{x}^{T}S_{x,\mu}A_{x}u\\
= & \left(A_{x}^{T}A_{x}\right)^{-1}\left(A_{x}^{T}\left(2\Diag(P_{x}\sigma_{x})-2\Diag\left(A_{\gamma}\left(A_{x}^{T}A_{x}\right)^{-1}\nabla f(x)\right)-S_{x,\mu}-3\Sigma_{x}+2P_{x}^{(2)}\right)A_{x}-\nabla^{2}f(x)\right)u\\
= & \left(A_{x}^{T}A_{x}\right)^{-1}\left(A_{x}^{T}\left(2S_{x,\mu}-S_{x,\mu}-3\Sigma_{x}+2P_{x}^{(2)}\right)A_{x}-\nabla^{2}f(x)\right)u.\\
= & \left(A_{x}^{T}A_{x}\right)^{-1}\left(A_{x}^{T}\left(S_{x,\mu}-3\Sigma_{x}+2P_{x}^{(2)}\right)A_{x}-\nabla^{2}f(x)\right)u.
\end{align*}
where we used the facts that 
\[
\mu=\left(A_{x}^{T}A_{x}\right)^{-1}A_{x}^{T}\sigma_{x}-(A_{x}^{T}A_{x})^{-1}\nabla f(x)
\]
\[
S_{x,\mu}=A_{x}\mu=\Diag(P_{x}\sigma_{x}-A_{x}(A_{x}^{T}A_{x})^{-1}\nabla f(x)).
\]
\end{proof}
\begin{rem*}
Note that $R(t)$ and $M(t)$ is symmetric in $\left\langle \cdot,\cdot\right\rangle _{\gamma}$,
but not in $\left\langle \cdot,\cdot\right\rangle _{2}$. That is
why the formula does not look symmetric.
\end{rem*}

\subsection{\label{sec:walk_is_random_log}Randomness of the Hamiltonian flow
($\ell_{0}$)}

Many parameters of a Hessian manifold relate to how fast a Hamiltonian
curve approaches the boundary of the polytope. Since the initial velocity
of the Hamiltonian curve is drawn from a Gaussian distribution, one
can imagine that $\norm{s_{\gamma'(0)}}_{\infty}=O\left(\frac{1}{\sqrt{m}}\right)\norm{s_{\gamma'(0)}}_{2}$
(each coordinate of $s_{\gamma'}$ measures the relative rate at which
the curve is approaching the corresponding facet). So the walk initial
approaches/leaves every facet of the polytope at roughly the same
slow pace. If this holds for the entire walk, it would allow us to
get very tight bounds on various parameters. Although we are not able
to prove that $\norm{s_{\gamma'(t)}}_{\infty}$ is stable throughout
$0\leq t\leq\step$, we will show that $\norm{s_{\gamma'(t)}}_{4}$
is stable and thereby obtain a good bound on $\norm{s_{\gamma'(t)}}_{\infty}$. 

Throughout this section, we only use the randomness of the walk to
prove that both $\norm{s_{\gamma'(t)}}_{4}$ and $\norm{s_{\gamma'(t)}}_{\infty}$
are small with high probability. Looking ahead, we will show that
$\norm{s_{\gamma'(t)}}_{4}=O(M_{1}^{1/4})$ and $\norm{s_{\gamma'(t)}}_{\infty}=O(\sqrt{\log n}+\sqrt{M_{1}}\step)$
(Lemma \ref{lem:gamma_est}), we define 
\[
\ell(\gamma)\defeq\max_{0\leq t\leq\step}\left(\frac{\norm{s_{\gamma'(t)}}_{2}}{n^{1/2}+M_{1}^{1/4}}+\frac{\norm{s_{\gamma'(t)}}_{4}}{M_{1}^{1/4}}+\frac{\norm{s_{\gamma'(t)}}_{\infty}}{\sqrt{\log n}+\sqrt{M_{1}}\step}+\frac{\norm{s_{\gamma'(0)}}_{2}}{n^{1/2}}+\frac{\norm{s_{\gamma'(0)}}_{4}}{n^{1/4}}+\frac{\norm{s_{\gamma'(0)}}_{\infty}}{\sqrt{\log n}}\right)
\]
to capture this randomness involves in generating the geodesic walk.
This allows us to perturb the geodesic (Lemma \ref{lem:one_one_cor})
without worrying about the dependence on randomness. 

We first prove the the walk is stable in the $L_{4}$ norm and hence
$\ell(\gamma)$ can be simply approximated by $\norm{s_{\gamma'(0)}}_{4}$
and $\norm{s_{\gamma'(0)}}_{\infty}$.
\begin{lem}
\label{lem:geodesic_4}Let $\gamma$ be a Hamiltonian flow in $\mathcal{M}_{L}$
starting at $x$. Let $v_{4}=\norm{s_{\gamma'(0)}}_{4}$. Then, for
$0\leq t\leq\frac{1}{12(v_{4}+M_{1}^{1/4})}$, we have that

\begin{enumerate}
\item $\norm{s_{\gamma'(t)}}_{4}\leq2v_{4}+M_{1}^{1/4}$.
\item $\norm{\gamma''(t)}_{\gamma(t)}^{2}\leq128v_{4}^{4}+30M_{1}$.
\end{enumerate}
\end{lem}

\begin{proof}
Let $u(t)=\norm{s_{\gamma'(t)}}_{4}$. Then, we have (using Holder's
inequality in the first step),
\begin{align}
\frac{du}{dt} & \leq\norm{\frac{d}{dt}\left(A_{\gamma}\gamma'\right)}_{4}=\norm{A_{\gamma}\gamma''-\left(A_{\gamma}\gamma'\right)^{2}}_{4}\nonumber \\
 & \leq\norm{A_{\gamma}\gamma''}_{4}+u^{2}(t).\label{eq:geodesic_error_d}
\end{align}

Under the Euclidean coordinates, by Lemma \ref{lem:geodesic_walk_log_def}
the Hamiltonian flow is given by
\[
\gamma''(t)=\left(A_{\gamma}^{T}A_{\gamma}\right)^{-1}A_{\gamma}^{T}s_{\gamma'}^{2}+\mu(\gamma(t))
\]
with $\mu(x)=\left(A_{x}^{T}A_{x}\right)^{-1}A_{x}^{T}\sigma_{x}-(A_{x}^{T}A_{x})^{-1}\nabla f(x)$.
Hence, we have that
\begin{align}
\norm{\gamma''}_{\gamma}^{2}\leq & 3\left(s_{\gamma'}^{2}\right)^{T}A_{\gamma}\left(A_{\gamma}^{T}A_{\gamma}\right)^{-1}\left(A_{\gamma}^{T}A_{\gamma}\right)\left(A_{\gamma}^{T}A_{\gamma}\right)^{-1}A_{\gamma}^{T}s_{\gamma'}^{2}\nonumber \\
 & +3\sigma_{\gamma}^{T}A_{\gamma}\left(A_{\gamma}^{T}A_{\gamma}\right)^{-1}\left(A_{\gamma}^{T}A_{\gamma}\right)\left(A_{\gamma}^{T}A_{\gamma}\right)^{-1}A_{\gamma}^{T}\sigma_{\gamma}\nonumber \\
 & +3(\nabla f(x))^{T}\left(A_{\gamma}^{T}A_{\gamma}\right)^{-1}\left(A_{\gamma}^{T}A_{\gamma}\right)\left(A_{\gamma}^{T}A_{\gamma}\right)^{-1}\nabla f(x)\nonumber \\
\leq & 3\sum_{i}(s_{\gamma'}^{4})_{i}+3\sum_{i}(\sigma_{\gamma}^{2})_{i}+3(\nabla f(x))^{T}\left(A_{\gamma}^{T}A_{\gamma}\right)^{-1}\nabla f(x)\nonumber \\
\leq & 3u^{4}(t)+3(n+M_{1})\leq3u^{4}(t)+6M_{1}\label{eq:geodesic_changes}
\end{align}
Therefore, we have
\[
\norm{A_{\gamma}\gamma''}_{4}\leq\norm{A_{\gamma}\gamma''}_{2}\leq2u^{2}(t)+3\sqrt{M_{1}}.
\]

Plugging it into (\ref{eq:geodesic_error_d}), we have that
\[
\frac{du}{dt}\leq3u^{2}(t)+3\sqrt{M_{1}}.
\]
Note that when $u\leq2v_{4}+M_{1}^{1/4}$, we have that
\[
\frac{du}{dt}\leq12v_{4}^{2}+9\sqrt{M_{1}}\leq12(v_{4}+M_{1}^{1/4})^{2}.
\]
Since $u(0)=v_{4}$, for $0\leq t\leq\frac{1}{12(v_{4}+M_{1}^{1/4})}$,
we have that $u(t)\leq2v_{4}+M_{1}^{1/4}$ and this gives the first
inequality. 

Using (\ref{eq:geodesic_changes}), we get the second inequality.
\end{proof}
We can now prove that $\ell(\gamma)$ is small with high probability.
\begin{lem}
\label{lem:V0_bound}Assume that $\step\leq\frac{1}{36M_{1}^{1/4}}$,
$\ell_{1}=\Omega(n^{1/4}\delta)$ and $n$ is large enough, we have
that
\[
\P_{\gamma\sim x}\left(\ell(\gamma)\geq128\right)\leq\frac{1}{100}\min\left(1,\frac{\ell_{0}}{\ell_{1}\delta}\right).
\]
Therefore, we have $\ell_{0}\leq256$.
\end{lem}

\begin{proof}
From the definition of the Hamiltonian curve (Lemma \ref{lem:RMCMC}),
we have that 
\[
A_{\gamma}\gamma'(0)=Bz
\]
where $B=A_{\gamma}\left(A_{\gamma}^{T}A_{\gamma}\right)^{-1/2}$
and $z\sim N(0,I)$. 

First, we estimate $\norm{A_{\gamma}\gamma'(t)}_{4}$. Lemma \ref{lem:norm_random_Ax}
shows that
\[
\P_{z\sim N(0,I)}\left(\norm{Bz}_{4}^{4}\leq\left(\left(3\sum_{i}\norm{e_{i}^{T}B}_{2}^{4}\right)^{1/4}+\norm B_{2\rightarrow4}s\right)^{4}\right)\leq1-\exp(-\frac{s^{2}}{2}).
\]
Note that $\sum_{i}\norm{e_{i}^{T}B}_{2}^{4}=\sum_{i}(\sigma_{\gamma})_{i}^{2}\leq n$
and $\norm B_{2\rightarrow4}\leq\norm B_{2\rightarrow2}=1$. Putting
$s=\frac{n^{1/4}}{2}$, we have that
\[
\P_{\gamma'(0)}\left(\norm{A_{\gamma}\gamma'(0)}_{4}^{4}\leq11n\right)\leq1-\exp(-\frac{\sqrt{n}}{8}).
\]
Therefore, we have that $v_{4}\defeq\norm{A_{\gamma}\gamma'(0)}_{4}\leq2n^{1/4}$
with probability at least $1-\exp(-\frac{\sqrt{n}}{8})$. Now, we
apply Lemma \ref{lem:geodesic_4} to get that
\[
\norm{s_{\gamma'(t)}}_{4}\leq2v_{4}+M_{1}^{1/4}\leq5M_{1}^{1/4}
\]
for all $0\leq t\leq\frac{1}{12(v_{4}+M_{1}^{1/4})}$. 

Next, we estimate $\norm{A_{\gamma}\gamma'(t)}_{\infty}$. Since $e_{i}^{T}A_{\gamma}\gamma'(0)=e_{i}^{T}Bx\sim N(0,\sigma_{i})$,
we have
\[
\P_{\gamma'(0)}\left(\left|e_{i}^{T}A_{\gamma}\gamma'(0)\right|\geq\sqrt{\sigma_{i}}t\right)\leq2\exp\left(-\frac{t^{2}}{2}\right).
\]
Hence, we have that
\[
\P_{\gamma'(0)}\left(\norm{A_{\gamma}\gamma'(0)}_{\infty}\geq2\sqrt{\log n}\right)\leq2\sum_{i}\exp\left(-\frac{2\log n}{\sigma_{i}}\right)
\]
Since $\sum_{i}\exp\left(-\frac{2\log n}{\sigma_{i}}\right)$ is concave
in $\sigma$, the maximum of $\sum_{i}\exp\left(-\frac{\log n}{\sigma_{i}}\right)$
on the feasible set $\{0\leq\sigma\leq1,\sum\sigma_{i}=n\}$ occurs
on its vertices. Hence, we have that
\[
\P_{\gamma'(0)}\left(\norm{A_{\gamma}\gamma'(0)}_{\infty}\geq2\sqrt{\log n}\right)\leq2n\exp\left(-2\log n\right)=\frac{2}{n}.
\]
Lemma \ref{lem:geodesic_4} shows that $\norm{A_{\gamma}\gamma''}_{\infty}\leq\norm{\gamma''}_{\gamma(t)}\leq46\sqrt{n}+6\sqrt{M_{1}}$.
Hence, for any $0\leq t\leq\step$, we have that

\begin{align}
\norm{s_{\gamma'(t)}}_{\infty} & \leq\norm{A_{\gamma(t)}\gamma'(0)}_{\infty}+\int_{0}^{t}\norm{A_{\gamma(t)}\gamma''(r)}_{\infty}dr\nonumber \\
 & \leq\left(\max_{i,0\leq s\leq t}\left|\frac{s_{\gamma(t),i}}{s_{\gamma(s),i}}\right|\right)\left(\norm{A_{\gamma(0)}\gamma'(0)}_{\infty}+\int_{0}^{\step}\norm{A_{\gamma(r)}\gamma''(r)}_{\infty}dr\right)\nonumber \\
 & \leq\left(\max_{i,0\leq s\leq t}\left|\frac{s_{\gamma(t),i}}{s_{\gamma(s),i}}\right|\right)\left(2\sqrt{\log n}+(46\sqrt{n}+6\sqrt{M_{1}})\step\right)\nonumber \\
 & \leq\left(\max_{i,0\leq s\leq t}\left|\frac{s_{\gamma(t),i}}{s_{\gamma(s),i}}\right|\right)\left(2\sqrt{\log n}+52\sqrt{M_{1}}\step\right).\label{eq:s_gamma_dot}
\end{align}
Let $z(t)=\max_{i,0\leq s\leq t}\left|\frac{s_{\gamma(t),i}}{s_{\gamma(s),i}}\right|$.
Note that
\[
s_{\gamma(t),i}=s_{\gamma(s),i}\exp\left(\int_{r}^{t}s_{\gamma'(\alpha),i}d\alpha\right)\qquad\mbox{since}\qquad(A\gamma(t)-b)_{i}=(A\gamma(r)-b)_{i}\exp\left(\int_{r}^{t}\frac{a_{i}\cdot\gamma'(\alpha)}{(A\gamma(\alpha)-b)_{i}}d\alpha\right)
\]
Hence, we have that
\[
z'(t)\leq z(t)\norm{s_{\gamma'(t)}}_{\infty}\leq z^{2}(t)\left(2\sqrt{\log n}+52\sqrt{M_{1}}\step\right).
\]
Solving this, since $z(0)=1$, we get
\[
z(t)\leq\frac{1}{1-\left(2\sqrt{\log n}+52\sqrt{M_{1}}\step\right)t}.
\]
Since $t\leq\step\leq\frac{1}{36M_{1}^{1/4}}$, we have that $z(t)\leq1.05$.
Putting this into (\ref{eq:s_gamma_dot}), we have that
\[
\norm{s_{\gamma'(t)}}_{\infty}\leq3\sqrt{\log n}+55\sqrt{M_{1}}\step.
\]

Finally, we estimate $\norm{s_{\gamma'(t)}}_{2}=\norm{A_{\gamma}\gamma'(t)}_{2}$.
Lemma \ref{lem:norm_random_Ax} shows that
\[
\P_{z\sim N(0,I)}\left(\norm{Bz}_{2}^{2}\leq\left(\left(\sum_{i}\norm{e_{i}^{T}B}_{2}^{2}\right)^{1/2}+\norm B_{2\rightarrow2}r\right)^{2}\right)\leq1-\exp(-\frac{r^{2}}{2}).
\]
Note that $\sum_{i}\norm{e_{i}^{T}B}_{2}^{2}=\sum_{i}(\sigma_{\gamma})_{i}\leq n$
and $\norm B_{2\rightarrow2}\leq1$. Putting $s=\frac{n^{1/2}}{3}$,
we have that
\[
\P_{\gamma'(0)}\left(\norm{A_{\gamma}\gamma'(0)}_{2}^{2}\leq2n\right)\leq1-\exp(-\frac{n}{18}).
\]
Therefore, $\norm{A_{\gamma}\gamma'(0)}_{2}^{2}\leq2n$ with high
probability. Next, we note that
\begin{align*}
\frac{d}{dt}\norm{A_{\gamma}\gamma'}_{2}^{2} & =\left\langle A_{\gamma}\gamma',A_{\gamma}\gamma''-s_{\gamma'}^{2}\right\rangle \\
 & =\left\langle A_{\gamma}\gamma',A_{\gamma}\left(A_{\gamma}^{T}A_{\gamma}\right)^{-1}A_{\gamma}^{T}s_{\gamma'}^{2}+A_{\gamma}\mu(\gamma(t))-A_{\gamma}s_{\gamma'}^{2}\right\rangle \\
 & =\left\langle A_{\gamma}\gamma',A_{\gamma}\left(A_{\gamma}^{T}A_{\gamma}\right)^{-1}A_{\gamma}^{T}s_{\gamma'}^{2}-s_{\gamma'}^{2}\right\rangle +\left\langle A_{\gamma}\gamma',A_{\gamma}\mu(\gamma(t))\right\rangle \\
 & =\left\langle A_{\gamma}\gamma',A_{\gamma}\mu(\gamma(t))\right\rangle .
\end{align*}
Using that $\mu(x)=\left(A_{x}^{T}A_{x}\right)^{-1}A_{x}^{T}\sigma_{x}-(A_{x}^{T}A_{x})^{-1}\nabla f(x)$,
we have that
\begin{align*}
\frac{d}{dt}\norm{A_{\gamma}\gamma'}_{2}^{2} & =\sum_{i}(s_{\gamma'})_{i}(\sigma_{\gamma})_{i}-\sum_{i}(\gamma')_{i}(\nabla f)_{i}\\
 & \leq\sqrt{\sum_{i}(s_{\gamma'})_{i}^{2}}\sqrt{\sum(\sigma_{\gamma})_{i}^{2}}+\sqrt{\sum_{i}(s_{\gamma'})_{i}^{2}}\sqrt{(\nabla f)^{T}(A_{\gamma}^{T}A_{\gamma})^{-1}(\nabla f)}\\
 & \leq2\norm{A_{\gamma}\gamma'}_{2}\sqrt{M_{1}}.
\end{align*}
Therefore, we have that $\left|\frac{d}{dt}\norm{A_{\gamma}\gamma'}_{2}\right|\leq\sqrt{M_{1}}.$
Since $\step\leq\frac{1}{36M_{1}^{1/4}}$, we have that $\norm{A_{\gamma}\gamma'(t)}_{2}\leq\norm{A_{\gamma}\gamma'(0)}_{2}+\frac{M_{1}^{1/4}}{36}\leq\sqrt{2n}+\frac{M_{1}^{1/4}}{36}$.
Therefore, we have that $\norm{A_{\gamma}\gamma'(t)}_{2}\leq2\sqrt{n}+M_{1}^{1/4}$
with probability at least $1-\exp(-\frac{n}{18})$.

Combining our estimates on $\norm{s_{\gamma'(t)}}_{2}$, $\norm{s_{\gamma'(t)}}_{4}$
and $\norm{s_{\gamma'(t)}}_{\infty}$ and using the assumption on
$\step$, we have that
\begin{align*}
 & \P\left(\max_{0\leq t\leq\step}\left(\frac{\norm{s_{\gamma'(t)}}_{2}}{n^{1/2}+M_{1}^{1/4}}+\frac{\norm{s_{\gamma'(t)}}_{4}}{M_{1}^{1/4}}+\frac{\norm{s_{\gamma'(t)}}_{\infty}}{\sqrt{\log n}+\sqrt{M_{1}}\step}+\frac{\norm{s_{\gamma'(0)}}_{2}}{n^{1/2}}+\frac{\norm{s_{\gamma'(0)}}_{4}}{n^{1/4}}+\frac{\norm{s_{\gamma'(0)}}_{\infty}}{\sqrt{\log n}}\right)\geq128\right)\\
\leq & \exp(-\frac{n}{18})+\exp(-\frac{\sqrt{n}}{8})+\frac{2}{n}.
\end{align*}
In Lemma \ref{lem:log_V1}, we indeed have that $\ell_{1}=\Omega(n^{1/4}\delta)$
and hence $\frac{\ell_{0}}{\ell_{1}\delta}=\Omega(n^{-1/4}\delta^{-2})=\Omega(1)$.
Therefore, the probability is less than $\frac{1}{100}\min\left(1,\frac{\ell_{0}}{\ell_{1}\delta}\right)$
when $n$ is large enough.
\end{proof}
Here, we collect some simple consequences of small $\ell(\gamma)$
that we will use later.
\begin{lem}
\label{lem:gamma_est}Given a Hamiltonian flow $\gamma$ on $\mathcal{M}_{L}$
with $\ell(\gamma)\leq\ell_{0}\le256$. For any $0\leq t\leq\step$,

\begin{enumerate}
\item $\norm{A_{\gamma}\gamma'(t)}_{2}\leq256(n^{1/2}+M_{1}^{1/4})$, $\norm{A_{\gamma}\gamma'(0)}_{2}\leq256n^{1/2}$.
\item $\norm{A_{\gamma}\gamma'(t)}_{4}\leq256M_{1}^{1/4}$, $\norm{A_{\gamma}\gamma'(0)}_{4}\leq256n^{1/4}$.
\item $\norm{A_{\gamma}\gamma'(t)}_{\infty}\leq256\left(\sqrt{\log n}+\sqrt{M_{1}}\step\right)$,
$\norm{A_{\gamma}\gamma'(0)}_{\infty}\leq256\sqrt{\log n}$.
\item $\norm{\gamma''(t)}_{\gamma}^{2}\leq10^{13}M_{1}$.
\end{enumerate}
\end{lem}

\begin{proof}
The first three inequalities simply follow from the definition of
$\ell(\gamma)$. Since $\norm{A_{\gamma}\gamma'(0)}_{4}\leq256M_{1}^{1/4}$,
Lemma \ref{lem:geodesic_4} shows the last inequality.
\end{proof}

\subsection{Parameters $R_{1}$, $R_{2}$ and $R_{3}$\label{subsec:Stability-of-Jacobian-log-barrier}}
\begin{lem}
\label{lem:total_ricci2}For a Hamiltonian curve $\gamma$ on $\mathcal{M}_{L}$
with $\ell(\gamma)\leq\ell_{0}$, we have that 
\[
\sup_{0\leq t\leq\ell}\norm{\Phi(t)}_{F,\gamma(t)}\leq R_{1}
\]
with $R_{1}=O(\sqrt{M_{1}}+M_{2}\sqrt{n})$.
\end{lem}

\begin{proof}
Note that $\Phi(t)=M(t)-R(t)$ where

\begin{align*}
R(t) & =\left(A_{\gamma}^{T}A_{\gamma}\right)^{-1}\left(A_{\gamma}^{T}S_{\gamma'}P_{\gamma}S_{\gamma'}A_{\gamma}-A_{\gamma}^{T}\Diag(P_{\gamma}s_{\gamma'}^{2})A_{\gamma}\right),\\
M(t) & =\left(A_{\gamma}^{T}A_{\gamma}\right)^{-1}\left(A_{\gamma}^{T}(S_{\gamma,\mu}-3\Sigma_{\gamma}+2P_{\gamma}^{(2)})A_{\gamma}-\nabla^{2}f\right).
\end{align*}
We bound the Frobenius norm of $\Phi(t)$ separately.

For $\norm{R(t)}_{F,\gamma}$, we note that
\begin{align*}
 & \norm{R(t)}_{F,\gamma}^{2}\\
\leq & 2\norm{(A_{\gamma}^{T}A_{\gamma})^{-1/2}A_{\gamma}^{T}S_{\gamma'}P_{\gamma}S_{\gamma'}A_{\gamma}(A_{\gamma}^{T}A_{\gamma})^{-1/2}}_{F}^{2}+2\norm{(A_{\gamma}^{T}A_{\gamma})^{-1/2}A_{\gamma}^{T}\Diag(P_{\gamma}s_{\gamma'}^{2})A_{\gamma}(A_{\gamma}^{T}A_{\gamma})^{-1/2}}_{F}^{2}\\
= & 2\tr P_{\gamma}S_{\gamma'}P_{\gamma}S_{\gamma'}P_{\gamma}S_{\gamma'}P_{\gamma}S_{\gamma'}+2\tr P_{\gamma}\Diag(P_{\gamma}s_{\gamma'}^{2})P_{\gamma}\Diag(P_{\gamma}s_{\gamma'}^{2})\\
\leq & 4\norm{s_{\gamma'}}_{4}^{4}.
\end{align*}

For $\norm{M(t)}_{F,\gamma}$, we note that
\begin{align*}
 & \norm{M(t)}_{F,\gamma}^{2}\\
\leq & 2\norm{(A_{\gamma}^{T}A_{\gamma})^{-1/2}A_{\gamma}^{T}(S_{\gamma,\mu}-3\Sigma_{\gamma}+2P_{\gamma}^{(2)})A_{\gamma}(A_{\gamma}^{T}A_{\gamma})^{-1/2}}_{F}^{2}+2\norm{(A_{\gamma}^{T}A_{\gamma})^{-1/2}\nabla^{2}f(A_{\gamma}^{T}A_{\gamma})^{-1/2}}_{F}^{2}\\
\leq & 6\tr P_{\gamma}S_{\gamma,\mu}P_{\gamma}S_{\gamma,\mu}+54\tr P_{\gamma}\Sigma_{\gamma}P_{\gamma}\Sigma_{\gamma}+24\tr P_{\gamma}P_{\gamma}^{(2)}P_{\gamma}P_{\gamma}^{(2)}+2M_{2}^{2}n\\
\leq & 6\tr S_{\gamma,\mu}^{2}+54\tr\Sigma_{\gamma}^{2}+24\tr(P_{\gamma}^{(2)})^{2}+2M_{2}^{2}n.
\end{align*}
where we used that $\norm{\diag(P_{\gamma}^{(2)})}^{2}\leq\norm{\sigma_{\gamma}}^{2}\leq n$
and
\begin{align*}
\tr P_{\gamma}MP_{\gamma}M & =\tr P_{\gamma}MP_{\gamma}MP_{\gamma}=\tr P_{\gamma}MMP_{\gamma}\\
 & =\tr MP_{\gamma}P_{\gamma}M\leq\tr M^{2}.
\end{align*}
Note that 
\[
\norm{s_{\gamma,\mu}}_{2}^{2}=\norm{\mu(\gamma)}_{\gamma}^{2}=\norm{P_{\gamma}\sigma_{\gamma}-A_{\gamma}(A_{\gamma}^{T}A_{\gamma})^{-1}\nabla f(\gamma)}_{2}^{2}\leq2n+2M_{1}\leq4M_{1}
\]
and
\[
\norm{A_{\gamma}\left(A_{\gamma}^{T}A_{\gamma}\right)^{-1}\nabla f(\gamma(t))}_{2}^{2}=\nabla f(\gamma(t))^{T}\left(A_{\gamma}^{T}A_{\gamma}\right)^{-1}\nabla f(\gamma(t))=M_{1}.
\]
Therefore, we have that
\[
\norm{M(t)}_{F,\gamma}^{2}\leq102M_{1}+2M_{2}^{2}n.
\]
The claim follows from Lemma \ref{lem:gamma_est}.
\end{proof}
\begin{lem}
\label{lem:log_R2}Let $\gamma$ be a Hamiltonian curve on $\mathcal{M}_{L}$
with $\ell(\gamma)\leq\ell_{0}$. Assume that $\step^{2}\leq\frac{1}{\sqrt{n}}$.
For any $0\leq t\leq\step$, any curve $c(r)$ starting from $\gamma(t)$
and any vector field $v(r)$ on $c(r)$ with $v(0)=\gamma'(t)$, we
have that

\[
\left|\frac{d}{dr}\left.\tr\Phi(v(r))\right|_{r=0}\right|\leq R_{2}\left(\norm{\left.\frac{dc}{dr}\right|_{r=0}}_{\gamma(t)}+\step\norm{\left.D_{r}v\right|_{r=0}}_{\gamma(t)}\right).
\]
where 
\[
R_{2}=O(\sqrt{nM_{1}}+\sqrt{n}M_{1}\step^{2}+\frac{M_{1}^{1/4}}{\step}+\frac{\sqrt{n\log n}}{\step}+\sqrt{n}M_{2}+M_{3}).
\]
\end{lem}

\begin{proof}
We first bound $\tr R(t)$. By Lemma \ref{lem:Riemann_tensor_log},
we know that
\begin{eqnarray*}
\Ric(v(r)) & = & s_{c(r),v(r)}^{T}P_{c(r)}^{(2)}s_{c(r),v(r)}-\sigma_{c(r)}^{T}P_{c(r)}s_{c(r),v(r)}^{2}\\
 & = & \tr(S_{c(r),v(r)}P_{c(r)}S_{c(r),v(r)}P_{c(r)})-\tr(\Diag(P_{c(r)}s_{c(r),v(r)}^{2})P_{c(r)}).
\end{eqnarray*}
For simplicity, we suppress the parameter $r$ and hence, we have
\[
\Ric(v)=\tr(S_{c,v}P_{c}S_{c,v}P_{c})-\tr(\Diag(P_{c}s_{c,v}^{2})P_{c}).
\]
We write $\frac{d}{dr}c=c'$ and $\frac{d}{dr}v=v'$ (in Euclidean
coordinates). Since $\frac{d}{dr}P_{c}=-S_{c'}P_{c}-P_{c}S_{c'}+2P_{c}S_{c'}P_{c}$
and $\frac{d}{dr}S_{c,v}=-S_{c'}S_{c,v}+S_{c,v'}$, we have that
\begin{eqnarray*}
 &  & \frac{d}{dr}\Ric(v)\\
 & = & -2\tr(S_{c,v}S_{c'}P_{c}S_{c,v}P_{c})-2\tr(S_{c,v}P_{c}S_{c'}S_{c,v}P_{c})+4\tr(S_{c,v}P_{c}S_{c'}P_{c}S_{c,v}P_{c})\\
 &  & -2\tr(S_{c'}S_{c,v}P_{c}S_{c,v}P_{c})+2\tr(S_{c,v'}P_{c}S_{c,v}P_{c})\\
 &  & +\tr(\Diag(P_{c}s_{c,v}^{2})S_{c'}P_{c})+\tr(\Diag(P_{c}s_{c,v}^{2})P_{c}S_{c'})-2\tr(\Diag(P_{c}s_{c,v}^{2})P_{c}S_{c'}P_{c})\\
 &  & +\tr(\Diag(P_{c}S_{c'}s_{c,v}^{2})P_{c})+\tr(\Diag(S_{c'}P_{c}s_{c,v}^{2})P_{c})-2\tr(\Diag(P_{c}S_{c'}P_{c}s_{c,v}^{2})P_{c})\\
 &  & +2\tr(\Diag(P_{c}S_{c,v}S_{c'}s_{c,v})P_{c})-2\tr(\Diag(P_{c}S_{c,v}s_{c,v'})P_{c})\\
 & = & -6\tr(S_{c,v}S_{c'}P_{c}S_{c,v}P_{c})+4\tr(S_{c,v}P_{c}S_{c'}P_{c}S_{c,v}P_{c})+2\tr(S_{c,v'}P_{c}S_{c,v}P_{c})\\
 &  & +3\tr(\Diag(P_{c}s_{c,v}^{2})S_{c'}P_{c})-2\tr(\Diag(P_{c}s_{c,v}^{2})P_{c}S_{c'}P_{c})\\
 &  & +3\tr(\Diag(P_{c}S_{c'}s_{c,v}^{2})P_{c})-2\tr(\Diag(P_{c}S_{c'}P_{c}s_{c,v}^{2})P_{c})\\
 &  & -2\tr(\Diag(P_{c}S_{c,v}s_{c,v'})P_{c}).
\end{eqnarray*}
Let $\frac{d}{dr}\Ric(v)=(1)+(2)$ where $(1)$ is the sum of all
terms not involving $v'$ and $(2)$ is the sum of other terms. 

For the first term $(1)$, we have that
\begin{eqnarray*}
\left|(1)\right| & \leq & 6\left|\tr(S_{c,v}S_{c'}P_{c}S_{c,v}P_{c})\right|+4\left|\tr(S_{c,v}P_{c}S_{c'}P_{c}S_{c,v}P_{c})\right|\\
 &  & +3\left|\tr(\Diag(P_{c}s_{c,v}^{2})S_{c'}P_{c})\right|+2\left|\tr(\Diag(P_{c}s_{c,v}^{2})P_{c}S_{c'}P_{c})\right|\\
 &  & +3\left|\tr(\Diag(P_{c}S_{c'}s_{c,v}^{2})P_{c})\right|+2\left|\tr(\Diag(P_{c}S_{c'}P_{c}s_{c,v}^{2})P_{c})\right|\\
 & \leq & 6\norm{s_{c'}}_{\infty}\sqrt{\sum_{i}(s_{c,v})_{i}^{2}}\sqrt{\sum_{i}(s_{c,v})_{i}^{2}}+4\norm{s_{c'}}_{\infty}\left|\tr(P_{c}S_{c,v}P_{c}S_{c,v}P_{c})\right|\\
 &  & +3\sqrt{\sum_{i}(s_{c,v})_{i}^{4}}\norm{S_{c'}}_{2}+2\sqrt{\sum_{i}(s_{c,v})_{i}^{4}}\sqrt{\sum_{i}(P_{c}S_{c'}P_{c})_{ii}^{2}}\\
 &  & +3\norm{s_{c'}}_{\infty}\sqrt{\sum_{i}(s_{c,v})_{i}^{4}}\sqrt{\sum_{i}(P_{c})_{ii}^{2}}+2\norm{s_{c'}}_{\infty}\sqrt{\sum_{i}(s_{c,v})_{i}^{4}}\sqrt{\sum_{i}(P_{c})_{ii}^{2}}\\
 & \leq & 10\norm{s_{c'}}_{\infty}\norm{s_{c,v}}_{2}^{2}+3\norm{s_{c,v}}_{4}^{2}\norm{s_{c'}}_{2}+7\norm{s_{c'}}_{\infty}\norm{s_{c,v}}_{4}^{2}\sqrt{n}\\
 & \leq & 20\norm{s_{c'}}_{2}\norm{s_{c,v}}_{4}^{2}\sqrt{n}.
\end{eqnarray*}
Since $s_{c,v}=s_{\gamma'}$ at $r=0$, we have that $\norm{s_{c,v}}_{4}^{2}=O(M_{1}^{1/2})$
and hence
\[
\left|(1)\right|=O\left(\sqrt{nM_{1}}\right)\norm{s_{c'}}_{2}.
\]

For the second term $(2)$, we have that
\begin{eqnarray*}
\left|(2)\right| & \leq & 2\left|\tr(S_{c,v'}P_{c}S_{c,v}P_{c})\right|+2\left|\tr(\Diag(P_{c}S_{c,v}s_{c,v'})P_{c})\right|\\
 & \leq & 2\norm{s_{c,v'}}_{2}\norm{s_{c,v}}_{2}+2\sqrt{n}\sqrt{\sum_{i}(s_{c,v'}s_{c,v})_{i}^{2}}\\
 & \leq & O\left(n^{1/2}+M_{1}^{1/4}\right)\norm{s_{c,v'}}_{2}+O\left(\sqrt{n\log n}+\sqrt{nM_{1}}\step\right)\norm{s_{c,v'}}_{2}\\
 & = & O\left(M_{1}^{1/4}+\sqrt{n\log n}+\sqrt{nM_{1}}\step\right)\norm{s_{c,v'}}_{2}
\end{eqnarray*}
where we used $\norm{s_{c,v}}_{\infty}=\norm{s_{\gamma'}}_{\infty}=O\left(\sqrt{\log n}+\sqrt{M_{1}}\step\right)$
and at $r=0,$ we have $\norm{s_{c,v}}_{2}=\norm{A_{\gamma(0)}\gamma'(0)}_{2}=\norm{s_{\gamma'}}_{2}=O(n^{1/2}+M_{1}^{1/4})$
in the second-to-last line.

Note that at $r=0$, by Lemma \ref{lem:geo_equ_log}, we have
\begin{align*}
D_{r}v & =\frac{dv}{dr}-\left(A_{c}^{T}A_{c}\right)^{-1}A_{c}^{T}S_{c'}s_{c,v}.
\end{align*}
Therefore, 
\[
s_{c,v'}=A_{c}v'=A_{c}\left(D_{r}v\right)-A_{c}\left(A_{c}^{T}A_{c}\right)^{-1}A_{c}^{T}S_{c'}s_{c,v}
\]
and hence
\begin{eqnarray*}
\norm{s_{c,v'}}_{2} & \leq & \norm{D_{r}v}+\norm{A_{c}\left(A_{c}^{T}A_{c}\right)^{-1}A_{c}^{T}S_{c'}s_{c,v}}_{2}\\
 & \leq & \norm{D_{r}v}+\norm{s_{\gamma'}}_{\infty}\norm{s_{c'}}_{2}
\end{eqnarray*}
Therefore, 
\begin{align*}
\left|(2)\right| & =O\left(M_{1}^{1/4}+\sqrt{n\log n}+\sqrt{nM_{1}}\step\right)\left(\norm{D_{r}v}+\left(\sqrt{\log n}+\sqrt{M_{1}}\step\right)\norm{s_{c'}}_{2}\right).
\end{align*}

Therefore, we have
\begin{align*}
 & \left|\frac{d}{dr}\left.\Ric(v(r))\right|_{s=0}\right|\\
= & O\left(\sqrt{nM_{1}}\right)\norm{s_{c'}}_{2}+O\left(M_{1}^{1/4}+\sqrt{n\log n}+\sqrt{nM_{1}}\step\right)\norm{D_{r}v}\\
 & +O\left(M_{1}^{1/4}+\sqrt{n\log n}+\sqrt{nM_{1}}\step\right)\left(\sqrt{\log n}+\sqrt{M_{1}}\step\right)\norm{s_{c'}}_{2}\\
= & O\left(\sqrt{nM_{1}}+M_{1}^{1/4}\sqrt{\log n}+M_{1}^{1/4}\sqrt{M_{1}}\step+\sqrt{n}\log n+\sqrt{n}M_{1}\step^{2}\right)\norm{s_{c'}}_{2}\\
 & +O\left(\frac{M_{1}^{1/4}}{\step}+\frac{\sqrt{n\log n}}{\step}+\sqrt{nM_{1}}\right)\step\norm{D_{r}v}\\
= & O\left(\sqrt{nM_{1}}+\sqrt{n}M_{1}\step^{2}+\frac{M_{1}^{1/4}}{\step}+\frac{\sqrt{n\log n}}{\step}\right)\left(\norm{s_{c'}}_{2}+\step\norm{D_{r}v}\right).
\end{align*}
where we used $M_{1}^{3/4}\step=O(\sqrt{n}M_{1}\step^{2}+\sqrt{nM_{1}})$
at the last line.

Next, we bound $\tr M(t)$. Lemma \ref{lem:formulate_M_R} shows that
\begin{align*}
\tr M(r)= & \tr((A_{c}^{T}A_{c})^{-1}A_{c}^{T}(S_{c,\mu}-3\Sigma_{c}+2P_{c}^{(2)})A_{c}-(A_{c}^{T}A_{c})^{-1}\nabla^{2}f)\\
= & \tr(P_{c}(S_{c,\mu}-3\Sigma_{c}+2P_{c}^{(2)}))-\tr((A_{c}^{T}A_{c})^{-1}\nabla^{2}f)\\
= & \sigma_{c}^{T}P_{c}\sigma_{c}-\sigma_{c}^{T}A_{c}(A_{c}^{T}A_{c})^{-1}\nabla f-3\tr(\Sigma_{c}^{2})+2\tr(P_{c}^{(3)})-\tr((A_{c}^{T}A_{c})^{-1}\nabla^{2}f).
\end{align*}
where in the last step we used 
\[
S_{c,\mu}=A_{c}\mu=\Diag(P_{c}\sigma_{c}-A_{c}(A_{c}^{T}A_{c})^{-1}\nabla f).
\]
Since $\frac{d}{dr}P_{c}=-S_{c'}P_{c}-P_{c}S_{c'}+2P_{c}S_{c'}P_{c}$
and $\frac{d}{dr}A_{c}=-S_{c'}A_{c}$, we have that
\begin{align*}
 & \frac{d}{dr}\tr M(r)\\
= & -\sigma_{c}^{T}S_{c'}P_{c}\sigma_{c}-\sigma_{c}^{T}P_{c}S_{c'}\sigma_{c}+2\sigma_{c}^{T}P_{c}S_{c'}P_{c}\sigma_{c}\\
 & -4\sigma_{c}^{T}P_{c}S_{c'}\sigma_{c}+4\sigma_{c}^{T}P_{c}\diag(P_{c}S_{c'}P_{c})\\
 & +3\sigma_{c}^{T}S_{c'}A_{c}(A_{c}^{T}A_{c})^{-1}\nabla f-2\diag(P_{c}S_{c'}P_{c})^{T}A_{c}(A_{c}^{T}A_{c})^{-1}\nabla f\\
 & -2\sigma_{c}^{T}A_{c}(A_{c}^{T}A_{c})^{-1}A_{c}^{T}S_{c'}A_{c}(A_{c}^{T}A_{c})^{-1}\nabla f-\sigma_{c}^{T}A_{c}(A_{c}^{T}A_{c})^{-1}\nabla^{2}f\cdot c'\\
 & +12\tr(S_{c'}\Sigma_{c}^{2})-12\tr(\Sigma_{c}\diag(P_{c}S_{c'}P_{c}))\\
 & -6\tr(P_{c}^{(2)}S_{c'}P_{c}+P_{c}^{(2)}P_{c}S_{c'})+12\tr(P_{c}^{(2)}P_{c}S_{c'}P_{c})\\
 & -2\tr((A_{c}^{T}A_{c})^{-1}A_{c}^{T}S_{c'}A_{c}(A_{c}^{T}A_{c})^{-1}\nabla^{2}f)-\tr((A_{c}^{T}A_{c})^{-1}\nabla^{3}f[c']).
\end{align*}
Simplifying it, we have
\begin{align*}
 & \frac{d}{dr}\tr M(r)\\
= & -6\sigma_{c}^{T}S_{c'}P_{c}\sigma_{c}+2\sigma_{c}^{T}P_{c}S_{c'}P_{c}\sigma_{c}+4\sigma_{c}^{T}P_{c}P_{c}^{(2)}s_{c'}\\
 & +3\sigma_{c}^{T}S_{c'}A_{c}(A_{c}^{T}A_{c})^{-1}\nabla f-2s_{c'}^{T}P_{c}^{(2)}A_{c}(A_{c}^{T}A_{c})^{-1}\nabla f\\
 & -2\sigma_{c}^{T}A_{c}(A_{c}^{T}A_{c})^{-1}A_{c}^{T}S_{c'}A_{c}(A_{c}^{T}A_{c})^{-1}\nabla f-\sigma_{c}^{T}A_{c}(A_{c}^{T}A_{c})^{-1}\nabla^{2}f\cdot c'\\
 & +12\tr(S_{c'}\Sigma_{c}^{2})-12\sigma_{c}^{T}P_{c}^{(2)}s_{c'}\\
 & -6\tr(P_{c}^{(2)}S_{c'}P_{c}+P_{c}^{(2)}P_{c}S_{c'})+12\tr(P_{c}^{(2)}P_{c}S_{c'}P_{c})\\
 & -2\tr((A_{c}^{T}A_{c})^{-1}A_{c}^{T}S_{c'}A_{c}(A_{c}^{T}A_{c})^{-1}\nabla^{2}f)-\tr((A_{c}^{T}A_{c})^{-1}\nabla^{3}f[c']).
\end{align*}
Let $\frac{d}{dr}\tr M(r)=(3)+(4)$ where $(3)$ is the sum of all
terms not involving $f$ and $(4)$ is the sum of other terms with
$f$. 

For the first term $(3)$, we have that
\begin{align*}
\left|(3)\right|\leq & 6\left|\sigma_{c}^{T}S_{c'}P_{c}\sigma_{c}\right|+2\left|\sigma_{c}^{T}P_{c}S_{c'}P_{c}\sigma_{c}\right|+4\left|\sigma_{c}^{T}P_{c}P_{c}^{(2)}s_{c'}\right|+12\left|\tr(S_{c'}\Sigma_{c}^{2})\right|+12\left|\sigma_{c}^{T}P_{c}^{(2)}s_{c'}\right|\\
 & +6\left|\tr(P_{c}^{(2)}S_{c'}P_{c}+P_{c}^{(2)}P_{c}S_{c'})\right|+12\left|\tr(P_{c}^{(2)}P_{c}S_{c'}P_{c})\right|\\
\leq & 6\sqrt{\sigma_{c}^{T}S_{c'}^{2}\sigma_{c}}\sqrt{n}+2\norm{s_{c'}}_{\infty}n+4\sqrt{n}\norm{s_{c'}}_{2}+12\norm{s_{c'}}_{2}\sqrt{n}+12\norm{S_{c'}}_{2}\sqrt{n}\\
 & +6\norm{\diag(P_{c}P_{c}^{(2)})}_{2}\norm{s_{c'}}_{2}+6\norm{\diag(P_{c}^{(2)}P_{c})}_{2}\norm{s_{c'}}_{2}+12\norm{\diag(P_{c}P_{c}^{(2)}P_{c})}_{2}\norm{s_{c'}}_{2}\\
\leq & 36n\norm{s_{c'}}_{2}.
\end{align*}

For the second term $(4)$, we have that
\begin{align*}
(4)\leq & 3\left|\sigma_{c}^{T}S_{c'}A_{c}(A_{c}^{T}A_{c})^{-1}\nabla f\right|+2\left|s_{c'}^{T}P_{c}^{(2)}A_{c}(A_{c}^{T}A_{c})^{-1}\nabla f\right|\\
 & +2\left|\sigma_{c}^{T}A_{c}(A_{c}^{T}A_{c})^{-1}A_{c}^{T}S_{c'}A_{c}(A_{c}^{T}A_{c})^{-1}\nabla f\right|+\left|\sigma_{c}^{T}A_{c}(A_{c}^{T}A_{c})^{-1}\nabla^{2}f\cdot c'\right|\\
 & +2\left|\tr((A_{c}^{T}A_{c})^{-1}A_{c}^{T}S_{c'}A_{c}(A_{c}^{T}A_{c})^{-1}\nabla^{2}f)\right|+\left|\tr((A_{c}^{T}A_{c})^{-1}\nabla^{3}f[c'])\right|\\
\leq & 3\sqrt{s_{c'}^{T}\Sigma_{c}s_{c'}}\sqrt{\nabla f^{T}(A_{c}^{T}A_{c})^{-1}\nabla f}+\sqrt{s_{c'}^{T}P_{c}^{(2)}P_{c}P_{c}^{(2)}s_{c'}}\sqrt{\nabla f^{T}(A_{c}^{T}A_{c})^{-1}\nabla f}\\
 & +2\sqrt{\sigma_{c}^{T}P_{c}S_{c'}P_{c}S_{c'}P_{c}\sigma_{c}}\sqrt{\nabla f^{T}(A_{c}^{T}A_{c})^{-1}\nabla f}+\sqrt{\sigma_{c}^{T}P_{c}\sigma_{c}}\sqrt{c'\nabla^{2}f(A_{c}^{T}A_{c})^{-1}\nabla^{2}f\cdot c'}\\
 & +2\norm{\diag(A_{c}(A_{c}^{T}A_{c})^{-1}\nabla^{2}f(A_{c}^{T}A_{c})^{-1}A_{c}^{T})}_{2}\norm{s_{c'}}_{2}+\left|\tr((A_{c}^{T}A_{c})^{-1}\nabla^{3}f[c'])\right|\\
\leq & 4\norm{s_{c'}}_{2}\sqrt{M_{1}}+2\norm{s_{c'}}_{\infty}\sqrt{nM_{1}}+3\sqrt{n}M_{2}\norm{s_{c'}}_{2}+M_{3}\norm{s_{c'}}_{2}\\
\leq & \left(6\sqrt{M_{1}}+3\sqrt{n}M_{2}+M_{3}\right)\norm{s_{c'}}_{2}.
\end{align*}

Therefore, 
\[
\left|\frac{d}{dr}\tr M(s)\right|\leq O\left(n+\sqrt{M_{1}}+\sqrt{n}M_{2}+M_{3}\right)\norm{s_{c'}}_{2}.
\]
\end{proof}
\begin{lem}
\label{lem:log_D2}Let $\gamma$ be a Hamiltonian curve on $\mathcal{M}_{L}$
with $\ell(\gamma)\leq\ell_{0}$. Assume that $\step\leq\frac{1}{36M_{1}^{1/4}}$.
Let $\zeta(t)$ be the parallel transport of the vector $\gamma'(0)$
on $\gamma(t)$. Then,
\[
\sup_{0\leq t\leq\step}\norm{\Phi(t)\zeta(t)}_{\gamma(t)}\leq R_{3}
\]
where 
\[
R_{3}=O\left(M_{1}^{\frac{1}{2}}\sqrt{\log n}+M_{1}^{\frac{3}{4}}n^{\frac{1}{4}}\step+M_{2}n^{\frac{1}{2}}\right).
\]
\end{lem}

\begin{proof}
By Lemma \ref{lem:formulate_M_R}, we have that
\begin{align*}
\Phi(t) & =\left(A_{\gamma}^{T}A_{\gamma}\right)^{-1}\left(A_{\gamma}^{T}(S_{\gamma,\mu}-3\Sigma_{\gamma}+2P_{\gamma}^{(2)}-S_{\gamma'}P_{\gamma}S_{\gamma'}+\Diag(P_{\gamma}s_{\gamma'}^{2}))A_{\gamma}-\nabla^{2}f\right)\\
 & =(1)+(2)
\end{align*}
where $(2)$ is the last term $-\left(A_{\gamma}^{T}A_{\gamma}\right)^{-1}\nabla^{2}f$.

For the first term, we have that
\begin{align*}
\norm{(1)\zeta}_{\gamma}= & \norm{P_{\gamma}(S_{\gamma,\mu}-3\Sigma_{\gamma}+2P_{\gamma}^{(2)}-S_{\gamma'}P_{\gamma}S_{\gamma'}+\Diag(P_{\gamma}s_{\gamma'}^{2}))s_{\gamma,\zeta}}_{2}\\
\leq & \norm{(S_{\gamma,\mu}-3\Sigma_{\gamma}+2P_{\gamma}^{(2)}-S_{\gamma'}P_{\gamma}S_{\gamma'}+\Diag(P_{\gamma}s_{\gamma'}^{2}))s_{\gamma,\zeta}}_{2}\\
\leq & \norm{S_{\gamma,\mu}s_{\gamma,\zeta}}_{2}+3\norm{\Sigma_{\gamma}s_{\gamma,\zeta}}_{2}+2\norm{P_{\gamma}^{(2)}s_{\gamma,\zeta}}_{2}+\norm{S_{\gamma'}P_{\gamma}S_{\gamma'}s_{\gamma,\zeta}}_{2}+\norm{\Diag(P_{\gamma}s_{\gamma'}^{2})s_{\gamma,\zeta}}_{2}\\
\leq & \norm{s_{\gamma,\zeta}}_{\infty}\left(\norm{s_{\gamma,\mu}}_{2}+\norm{\sigma_{\gamma}}_{2}+\norm{P_{\gamma}s_{\gamma'}^{2}}_{2}\right)+2\norm{P_{\gamma}^{(2)}s_{\gamma,\zeta}}_{2}+\norm{s_{\gamma'}}_{\infty}\norm{S_{\gamma'}s_{\gamma,\zeta}}_{2}\\
\leq & \norm{s_{\gamma,\zeta}}_{\infty}\left(2\sqrt{M_{1}}+\sqrt{n}+\norm{s_{\gamma'}}_{4}^{2}\right)+2\norm{P_{\gamma}^{(2)}s_{\gamma,\zeta}}_{2}+\norm{s_{\gamma'}}_{\infty}\norm{s_{\gamma'}}_{4}\norm{s_{\gamma,\zeta}}_{4}
\end{align*}
where we used that $\norm{s_{\gamma,\mu}}_{2}^{2}=\norm{\mu(x)}_{x}^{2}\leq4M_{1}$,
$\norm{\sigma_{\gamma}}^{2}\leq n$. Now, we note that 
\begin{align*}
\norm{P_{\gamma}^{(2)}s_{\gamma,\zeta}}_{2}^{2} & =\sum_{i}\left(\sum_{j}(P_{\gamma})_{ij}^{2}(s_{\gamma,\zeta})_{j}\right)^{2}\\
 & \leq\norm{s_{\gamma,\zeta}}_{\infty}^{2}\sum_{i}\left(\sum_{j}(P_{\gamma})_{ij}^{2}\right)\\
 & =\norm{s_{\gamma,\zeta}}_{\infty}^{2}\sum_{i}\left(\sigma_{\gamma}^{2}\right)_{i}=n\norm{s_{\gamma,\zeta}}_{\infty}^{2}.
\end{align*}
Using also that $\norm{s_{\gamma'}}_{4}=O(M_{1}^{1/4})$ and $\norm{s_{\gamma'}}_{\infty}=O(\sqrt{\log n}+\sqrt{M_{1}}\step)$,
we have that
\begin{align*}
\norm{(1)\zeta}_{\gamma} & \leq\norm{s_{\gamma,\zeta}}_{\infty}\left(5\sqrt{M_{1}}+\norm{s_{\gamma'}}_{4}^{2}\right)+\norm{s_{\gamma'}}_{\infty}\norm{s_{\gamma'}}_{4}\norm{s_{\gamma,\zeta}}_{4}\\
 & \leq O\left(M_{1}^{1/2}\right)\norm{s_{\gamma,\zeta}}_{\infty}+O(\sqrt{\log n}M_{1}^{1/4}+M_{1}^{3/4}\step)\norm{s_{\gamma,\zeta}}_{4}.
\end{align*}

For the second term, we have that
\[
\norm{(1)\zeta}_{\gamma}=\norm{P_{\gamma}\nabla^{2}f\zeta}_{2}^{2}\leq\zeta^{T}\nabla^{2}f\cdot\zeta\leq M_{2}\norm{s_{\gamma,\zeta}}_{2}.
\]
Combining both terms, we have that
\begin{equation}
\norm{\Phi(t)\zeta}_{\gamma}=O\left(M_{1}^{1/2}\right)\norm{s_{\gamma,\zeta}}_{\infty}+O(\sqrt{\log n}M_{1}^{1/4}+M_{1}^{3/4}\step)\norm{s_{\gamma,\zeta}}_{4}+M_{2}\norm{s_{\gamma,\zeta}}_{2}.\label{eq:phi_gamma_bound}
\end{equation}

Now, we bound $\norm{s_{\gamma,\zeta}}_{2}$, $\norm{s_{\gamma,\zeta}}_{4}$
and $\norm{s_{\gamma,\zeta}}_{\infty}$. (\ref{eq:parallel_log})
shows that
\[
\frac{d}{dt}\zeta(t)=\left(A_{\gamma}^{T}A_{\gamma}\right)^{-1}A_{\gamma}^{T}S_{\gamma'}A_{\gamma}\zeta.
\]

Let $w_{p}(t)=\norm{A_{\gamma}\zeta(t)}_{p}$. Then, we have that
\begin{align*}
\frac{d}{dt}w_{p}(t) & \leq\norm{\frac{d}{dt}A_{\gamma}\zeta(t)}_{p}\\
 & \leq\norm{S_{\gamma'}A_{\gamma}\zeta(t)}_{p}+\norm{A_{\gamma}\frac{d}{dt}\zeta(t)}_{p}\\
 & \leq\norm{S_{\gamma'}A_{\gamma}\zeta(t)}_{p}+\norm{P_{\gamma}S_{\gamma'}A_{\gamma}\zeta}_{p}\\
 & \leq\norm{s_{\gamma'}}_{\infty}w_{p}(t)+\norm{S_{\gamma'}A_{\gamma}\zeta}_{2}.
\end{align*}
For $p=2$, we have that $\frac{d}{dt}w_{2}(t)\leq2\norm{s_{\gamma'}}_{\infty}w_{2}(t)$.
Using that $\norm{s_{\gamma'}}_{\infty}\leq256\left(\sqrt{\log n}+\sqrt{M_{1}}\step\right)$
and that $t\leq\step\leq\frac{1}{36M_{1}^{1/4}}$, we have
\begin{align*}
w_{2}(t) & \leq e^{512\left(\sqrt{\log n}+\sqrt{M_{1}}\step\right)t}w_{2}(0)\\
 & =O(n^{1/2})
\end{align*}
where we used that $\zeta(0)=\gamma'(0)$ and $t\leq\frac{1}{12(v_{4}+M_{1}^{1/4})}$
at the end. Therefore, we have that $\norm{s_{\gamma,\zeta}}_{2}=O(n^{1/2})$.

For $p=4$, we note that 
\begin{align*}
\frac{d}{dt}w_{4}(t) & \leq\norm{s_{\gamma'}}_{\infty}w_{4}(t)+\norm{s_{\gamma'}}_{4}\norm{s_{\gamma,\zeta}}_{4}\\
 & \leq2\norm{s_{\gamma'}}_{4}w_{4}(t)=O(M_{1}^{1/4}w_{4}(t)).
\end{align*}
Since $t\leq\step\leq\frac{1}{36M_{1}^{1/4}}$, we have again that
$w_{4}(t)=O(w_{4}(0))$. Since $w_{4}(0)=\norm{A_{\gamma}\zeta(0)}_{4}=\norm{A_{\gamma}\gamma'(0)}_{4}=O(n^{1/4})$,
we have that
\[
w_{4}(t)=O(n^{1/4}).
\]

For $p=\infty$, we note that
\begin{align*}
\frac{d}{dt}w_{\infty}(t) & \leq\norm{s_{\gamma'}}_{\infty}w_{\infty}(t)+\norm{s_{\gamma'}}_{4}\norm{s_{\gamma,\zeta}}_{4}\\
 & \leq O(M_{1}^{1/4}w_{\infty}(t))+O(M_{1}^{1/4}n^{1/4}).
\end{align*}
Again using $t\leq\step\leq\frac{1}{36M_{1}^{1/4}}$, we have that
$w_{\infty}(t)\leq O(\sqrt{\log n}+M_{1}^{1/4}n^{1/4}\step)$.

Combining our bounds on $w_{2}$, $w_{4}$, $w_{\infty}$ to (\ref{eq:phi_gamma_bound}),
we get
\begin{align*}
\norm{\Phi(t)\zeta}_{\gamma} & =O\left(M_{1}^{\frac{1}{2}}\sqrt{\log n}+M_{1}^{\frac{3}{4}}n^{\frac{1}{4}}\step\right)+O(\sqrt{\log n}M_{1}^{\frac{1}{4}}n^{\frac{1}{4}}+M_{1}^{\frac{3}{4}}n^{\frac{1}{4}}\step)+O\left(M_{2}n^{\frac{1}{2}}\right)\\
 & =O\left(M_{1}^{\frac{1}{2}}\sqrt{\log n}+M_{1}^{\frac{3}{4}}n^{\frac{1}{4}}\step+M_{2}n^{\frac{1}{2}}\right).
\end{align*}
\end{proof}

\subsection{\label{subsec:Stability-of-norm}Stability of $L_{2}+L_{4}+L_{\infty}$
norm ($\ell_{1}$) }
\begin{lem}
\label{lem:log_V1}Given a family of Hamiltonian curves $\gamma_{r}(t)$
on $\mathcal{M}_{L}$ with $\ell(\gamma_{0})\leq\ell_{0}$, for $\step^{2}\leq\frac{1}{\sqrt{M_{1}}+M_{2}\sqrt{n}}$,
we have that
\begin{equation}
\left|\frac{d}{dr}\ell(\gamma_{r})\right|\leq O\left(M_{1}^{1/4}\step+\frac{1}{\step\sqrt{\log n}}\right)\left(\norm{\frac{d}{dr}\gamma_{r}(0)}_{\gamma_{r}(0)}+\step\norm{D_{r}\gamma_{r}'(0)}_{\gamma_{r}(0)}\right).\label{eq:ell_1_bound}
\end{equation}
Hence, $\ell_{1}=O\left(M_{1}^{1/4}\step+\frac{1}{\step\sqrt{\log n}}\right)$.
\end{lem}

\begin{proof}
For brevity, all $\frac{d}{dr}$ are evaluated at $r=0$. Since $\frac{d}{dr}s_{\gamma_{r}'}=-S_{\gamma_{r},\frac{d}{dr}\gamma_{r}}s_{\gamma_{r}'}+A_{\gamma_{r}}\frac{d}{dr}\gamma_{r}'$,
we have 
\begin{align*}
\norm{\frac{d}{dr}s_{\gamma_{r}'}}_{2} & \leq\norm{s_{\gamma_{r}'}}_{\infty}\norm{A_{\gamma_{r}}\frac{d}{dr}\gamma_{r}}_{2}+\norm{A_{\gamma_{r}}\frac{d}{dr}\gamma_{r}'}_{2}.
\end{align*}
For the last term, we note that
\begin{align*}
D_{r}\gamma_{r}' & =\frac{d}{dr}\gamma_{r}'-\left(A_{\gamma_{r}}^{T}A_{\gamma_{r}}\right)^{-1}A_{\gamma_{r}}^{T}S_{\gamma_{r},\frac{d}{dr}\gamma_{r}}s_{\gamma_{r}'}.
\end{align*}
Hence, we have
\[
\norm{A_{\gamma_{r}}\frac{d}{dr}\gamma_{r}'}_{2}\leq\norm{D_{r}\gamma_{r}'}_{\gamma_{r}}+\norm{S_{\gamma_{r},\frac{d}{dr}\gamma_{r}}s_{\gamma'}}_{2}\leq\norm{D_{r}\gamma_{r}'}_{\gamma_{r}}+\norm{s_{\gamma_{r}'}}_{\infty}\norm{s_{\gamma_{r},\frac{d}{dr}\gamma_{r}}}_{2}.
\]
Therefore, we have that
\begin{align}
\norm{\frac{d}{dr}s_{\gamma_{r}'}}_{2} & \leq2\norm{s_{\gamma_{r}'}}_{\infty}\norm{\frac{d}{dr}\gamma_{r}}_{\gamma_{r}}+\norm{D_{r}\gamma'}_{\gamma}\nonumber \\
 & =O\left(\sqrt{\log n}+\sqrt{M_{1}}\step\right)\norm{\frac{d}{dr}\gamma_{r}}_{\gamma_{r}}+\norm{D_{r}\gamma_{r}'}.\label{eq:log_d_s_s_gamma}
\end{align}

Since $\gamma_{r}$ is a family of Hamiltonian curves, Lemma \ref{lem:Jacobi_field}
shows that
\[
\overline{\psi}''(t)=\Gamma_{t}\Phi(t)\Gamma_{t}^{-1}\overline{\psi}(t)
\]
where $\overline{\psi}(t)$ is the parallel transport of $\frac{d}{dr}\gamma_{r}(t)$
from $\gamma_{r}(t)$ to $\gamma_{r}(0)$. By Lemma \ref{lem:geo_equ_log},
we have that
\[
\overline{\psi}(t)=\frac{d}{dr}\gamma_{r}(0)+D_{r}\gamma_{r}'(0)t+\int_{0}^{t}(t-r)\Gamma_{r}\Phi(r)\Gamma_{r}^{-1}(\frac{d}{dr}\gamma_{r}(0)+D_{r}\gamma_{r}'(0)r+E(r))dr
\]
with $\norm{E(r)}_{F}\leq O(1)\Delta$ and $\Delta\defeq\norm{\frac{d}{dr}\gamma_{r}(0)}_{\gamma_{r}(0)}+\step\norm{D_{r}\gamma_{r}'(0)}_{\gamma_{r}(0)}$
where we used that $\norm{\Phi(t)}_{F,\gamma}=O(\sqrt{M_{1}}+M_{2}\sqrt{n})$
(Lemma \ref{lem:total_ricci2}) and that $s^{2}\leq\step^{2}\leq\frac{1}{\sqrt{M_{1}}+M_{2}\sqrt{n}}$. 

Therefore, we have that
\begin{align*}
\norm{\frac{d}{dr}\gamma_{r}(t)}_{\gamma_{r}(t)}=\norm{\overline{\psi}(t)}_{\gamma_{r}(0)} & \leq\Delta+O(\Delta)\int_{0}^{t}(t-s)\norm{\Gamma_{r}\Phi(s)\Gamma_{r}^{-1}}_{\gamma_{r}(0)}dr\\
 & \leq O(\Delta)
\end{align*}
where we used again $\norm{\Phi(s)}_{F,\gamma}=O(\sqrt{M_{1}}+M_{2}\sqrt{n})$
and $s^{2}\leq\step^{2}\leq\frac{1}{\sqrt{M_{1}}+M_{2}\sqrt{n}}$.

Similarly, we have that $\norm{D_{r}\gamma_{r}'(t)}_{\gamma_{r}(t)}=\norm{\overline{\psi}'(t)}_{\gamma_{r}(0)}\leq O(\frac{\Delta}{\step}).$

Putting these into (\ref{eq:log_d_s_s_gamma}) and using $h\leq\frac{1}{\sqrt{n}}$,
we have
\begin{align}
\norm{\frac{d}{dr}s_{\gamma_{r}'}}_{2} & =O\left(\sqrt{\log n}+\sqrt{M_{1}}\step\right)\Delta+\frac{\Delta}{\step}\nonumber \\
 & =O\left(\sqrt{M_{1}}\step+\frac{1}{\step}\right)\Delta.\label{eq:ell_1_1}
\end{align}

We write 
\[
\ell(\gamma_{r})=\max_{0\leq t\leq\step}\left(\frac{\norm{s_{\gamma_{r}'(t)}}_{2}}{n^{1/2}+M_{1}^{1/4}}+\frac{\norm{s_{\gamma_{r}'(t)}}_{4}}{M_{1}^{1/4}}+\frac{\norm{s_{\gamma_{r}'(t)}}_{\infty}}{\sqrt{\log n}+\sqrt{M_{1}}\step}+\frac{\norm{s_{\gamma_{r}'(0)}}_{2}}{n^{1/2}}+\frac{\norm{s_{\gamma_{r}'(0)}}_{4}}{n^{1/4}}+\frac{\norm{s_{\gamma_{r}'(0)}}_{\infty}}{\sqrt{\log n}}\right).
\]
According to same calculation as (\ref{eq:log_d_s_s_gamma}), we can
improve the estimate on $\norm{\frac{d}{dr}s_{\gamma_{r}'}}_{2}$
for $t=0$ and get
\begin{align}
\norm{\frac{d}{dr}s_{\gamma_{r}'}(0)}_{2} & \leq\sqrt{\log n}\norm{\frac{d}{dr}\gamma_{r}(0)}_{\gamma_{r}}+\norm{D_{r}\gamma_{r}'(0)}\nonumber \\
 & \leq\frac{\Delta}{\step}.\label{eq:ell_1_2}
\end{align}
Using (\ref{eq:ell_1_1}) and (\ref{eq:ell_1_2}), we have that
\begin{align*}
\left|\frac{d}{dr}\ell(\gamma_{r})\right| & =O\left(\max_{0\leq t\leq\step}\frac{\norm{\frac{d}{dr}s_{\gamma_{r}'}(t)}_{2}}{M_{1}^{1/4}}+\frac{\norm{\frac{d}{dr}s_{\gamma_{r}'}(0)}_{2}}{\sqrt{\log n}}\right)\\
 & =O\left(\frac{\sqrt{M_{1}}\step+\frac{1}{\step}}{M_{1}^{1/4}}+\frac{1}{\step\sqrt{\log n}}\right)\Delta\\
 & =O\left(M_{1}^{1/4}\step+\frac{1}{\step\sqrt{\log n}}\right)\Delta.
\end{align*}
\end{proof}

\subsection{Mixing Time}
\begin{lem}
\label{lem:log-params}If $f(x)=\alpha\cdot\phi(x)$ (logarithmic
barrier), then, we have that $M_{1}=n+\alpha^{2}m$, $M_{2}=\alpha$
and $M_{3}=2\alpha\cdot\sqrt{n}$.
\end{lem}

\begin{proof}
For $M_{1}$, we note that $(\nabla f(x))^{T}\left(A_{x}^{T}A_{x}\right)^{-1}\nabla f(x)=\alpha^{2}1^{T}A_{x}(A_{x}^{T}A_{x})^{-1}A_{x}^{T}1\le\alpha^{2}m$.
Hence, $M_{1}=n+\alpha^{2}m$. 

For $M_{2}$, it directly follows from the definition.

For $M_{3}$, we note that
\begin{align*}
\tr((A_{x}^{T}A_{x})^{-1}\nabla^{3}f(x)[v]) & =-2\alpha\tr((A_{x}^{T}A_{x})^{-1}A_{x}^{T}S_{x,v}A_{x})\\
 & =-2\alpha\sum_{i}\sigma_{x,i}(s_{x,v})_{i}.
\end{align*}
Hence, we have
\[
\left|\tr((A_{x}^{T}A_{x})^{-1}\nabla^{3}f(x)[v])\right|\leq2\alpha\sqrt{\sum_{i}\sigma_{x,i}^{2}}\sqrt{\sum_{i}(s_{x,v})_{i}^{2}}\leq2\alpha\sqrt{n}\norm v_{x}.
\]
\end{proof}
Using Theorem \ref{thm:TV_diff_improved}, we have the following
\begin{lem}
\label{lem:one-step}There is a universal constant $c>0$ such that
if the step size 
\[
\step\le c\cdot\min\left(n^{-\frac{1}{3}},\alpha^{-\frac{1}{3}}m^{-\frac{1}{6}}n^{-\frac{1}{6}},\alpha^{-\frac{1}{2}}m^{-\frac{1}{4}}n^{-\frac{1}{12}}\right),
\]
then, all the $\delta$ conditions for Theorem \ref{thm:gen-convergence}
are satisfied.
\end{lem}

\begin{proof}
In the previous section, we proved that if $\step\leq\frac{1}{36M_{1}^{1/4}}$
and $n$ is large enough,

\begin{enumerate}
\item $\ell_{0}=256$ (Lemma \ref{lem:V0_bound})
\item $\ell_{1}=O\left(M_{1}^{1/4}\step+\frac{1}{\step\sqrt{\log n}}\right)$
(Lemma \ref{lem:log_V1})
\item $R_{1}=O(\sqrt{M_{1}}+M_{2}\sqrt{n})$ (Lemma \ref{lem:total_ricci2})
\item $R_{2}=O(\sqrt{nM_{1}}+\sqrt{n}M_{1}\step^{2}+\frac{M_{1}^{1/4}}{\step}+\frac{\sqrt{n\log n}}{\step}+\sqrt{n}M_{2}+M_{3})$
(Lemma \ref{lem:log_R2})
\item $R_{3}=O(M_{1}^{\frac{1}{2}}\sqrt{\log n}+M_{1}^{\frac{3}{4}}n^{\frac{1}{4}}\step+M_{2}n^{\frac{1}{2}})$
(Lemma \ref{lem:log_D2})
\end{enumerate}
Substituting the value of $M_{1}$, $M_{2}$ and $M_{3}$ and using
that $\delta\lesssim n^{-\frac{1}{3}}$, $\delta\lesssim\alpha^{-\frac{1}{2}}n^{-\frac{1}{3}}$
and $\delta\lesssim\alpha^{-\frac{1}{2}}m^{-\frac{1}{4}}$, we have
that
\begin{enumerate}
\item $\ell_{1}=O\left(\alpha^{\frac{1}{2}}m^{\frac{1}{4}}\delta+\frac{1}{\step\sqrt{\log n}}\right)$
\item $R_{1}=O(\sqrt{n}+\alpha\sqrt{m})$
\item $R_{2}=O(n+\alpha\sqrt{nm}+\frac{\sqrt{\alpha}m^{\frac{1}{4}}+\sqrt{n\log n}}{\step})$
\item $R_{3}=O(\sqrt{n\log n}+\alpha\sqrt{m\log n}+n\delta+\alpha^{\frac{3}{2}}m^{\frac{3}{4}}n^{\frac{1}{4}}\delta)$
\end{enumerate}
Now, we verify all the $\delta$ conditions for Theorem \ref{thm:gen-convergence}.

Using $\delta\lesssim n^{-\frac{1}{3}}$ and $\delta\lesssim\alpha^{-\frac{1}{2}}m^{-\frac{1}{4}}$,
we have that $\step^{2}\lesssim\frac{1}{R_{1}}$.

Using $\delta\lesssim n^{-\frac{1}{3}}$ and $\delta\lesssim\alpha^{-\frac{1}{2}}m^{-\frac{1}{4}}$,
we have that
\begin{align*}
\delta^{5}R_{1}^{2}\ell_{1} & \lesssim\delta^{5}(n+\alpha^{2}m)\cdot(\alpha^{\frac{1}{2}}m^{\frac{1}{4}}\delta+\frac{1}{\step\sqrt{\log n}})\leq\ell_{0}.
\end{align*}

For the last condition, we note that 
\begin{align*}
\step^{3}R_{2}+\step^{2}R_{3}\lesssim & n\delta^{3}+\alpha\sqrt{nm}\delta^{3}+\sqrt{\alpha}m^{\frac{1}{4}}\delta^{2}+\sqrt{n\log n}\delta^{2}\\
 & +\sqrt{n\log n}\delta^{2}+\alpha\sqrt{m\log n}\delta^{2}+n\delta^{3}+\alpha^{\frac{3}{2}}m^{\frac{3}{4}}n^{\frac{1}{4}}\delta^{3}\\
\lesssim & \sqrt{\alpha}m^{\frac{1}{4}}\delta^{2}+\sqrt{n\log n}\delta^{2}+\alpha\sqrt{m\log n}\delta^{2}+n\delta^{3}+\alpha\sqrt{nm}\delta^{3}+\alpha^{\frac{3}{2}}m^{\frac{3}{4}}n^{\frac{1}{4}}\delta^{3}\\
\lesssim & \sqrt{n\log n}\delta^{2}+\alpha\sqrt{m\log n}\delta^{2}+n\delta^{3}+\alpha\sqrt{nm}\delta^{3}+\alpha^{\frac{3}{2}}m^{\frac{3}{4}}n^{\frac{1}{4}}\delta^{3}
\end{align*}
where we used that $\sqrt{\alpha}m^{\frac{1}{4}}\lesssim(1+\alpha\sqrt{m})$
at the end

Therefore, if 
\[
\delta\leq c\cdot\min\left(n^{-\frac{1}{3}},\alpha^{-\frac{1}{3}}m^{-\frac{1}{6}}n^{-\frac{1}{6}},\alpha^{-\frac{1}{2}}m^{-\frac{1}{4}}n^{-\frac{1}{12}}\right)
\]
 for small enough constant, then all the $\delta$ conditions for
Theorem \ref{thm:gen-convergence} are satisfied.
\end{proof}

\begin{acknowledgement*}
We thank Ben Cousins for helpful discussions. This work was supported
in part by NSF awards CCF-1563838, CCF-1717349 and CCF-1740551.
\end{acknowledgement*}
\bibliographystyle{plain}
\bibliography{acg}

\appendix

\section{\label{sec:Matrix-ODE}Matrix ODE}

In this section, we prove Lemmas (\ref{lem:ODE_upper}) and (\ref{lem:matrix_ODE_est_1})
for the solution of the ODE (\ref{eq:matrix_ODE}), restated below
for convenience.
\begin{align*}
\frac{d^{2}}{dt^{2}}\Psi(t) & =\Phi(t)\Psi(t),\\
\frac{d}{dt}\Psi(0) & =B,\\
\Psi(0) & =A.
\end{align*}
\begin{lem}
Consider the matrix ODE (\ref{eq:matrix_ODE}). Let $\lambda=\max_{0\leq t\leq\ell}\norm{\Phi(t)}_{2}$
. For any $t\geq0$, we have that
\[
\norm{\Psi(t)}_{2}\leq\norm A_{2}\cosh(\sqrt{\lambda}t)+\frac{\norm B_{2}}{\sqrt{\lambda}}\sinh(\sqrt{\lambda}t).
\]
\end{lem}

\begin{proof}
Note that
\begin{eqnarray}
\Psi(t) & = & \Psi(0)+t\Psi'(0)+\int_{0}^{t}(t-s)\Psi''(s)ds\nonumber \\
 & = & A+tB+\int_{0}^{t}(t-s)\Phi(s)\Psi(s)ds.\label{eq:second_taylor_psi}
\end{eqnarray}
Let $a(t)=\norm{\Psi(t)}_{2}$, then we have that
\[
a(t)\leq\norm A_{2}+t\norm B_{2}+\lambda\int_{0}^{t}(t-s)a(s)ds.
\]
Let $\overline{a}(t)$ be the solution of the integral equation
\[
\overline{a}(t)=\norm A_{2}+t\norm B_{2}+\lambda\int_{0}^{t}(t-s)\overline{a}(s)ds.
\]
By induction, we have that $a(t)\leq\overline{a}(t)$ for all $t\geq0$.
By taking derivatives on both sides, we have that
\[
\overline{a}''(t)=\lambda\overline{a}(t),\ \overline{a}(0)=\norm A_{2},\ \overline{a}'(0)=\norm B_{2}.
\]
Solving these equations, we have
\[
\norm{\Psi(t)}_{2}=a(t)\leq\overline{a}(t)=\norm A_{2}\cosh(\sqrt{\lambda}t)+\frac{\norm B_{2}}{\sqrt{\lambda}}\sinh(\sqrt{\lambda}t)
\]
for all $t\geq0$.
\end{proof}
\begin{lem}
Consider the matrix ODE (\ref{eq:matrix_ODE}). Let $\lambda=\max_{0\leq t\leq\ell}\norm{\Phi(t)}_{F}$.
For any $0\leq t\leq\frac{1}{\sqrt{\lambda}}$, we have that
\[
\norm{\Psi(t)-A-Bt}_{F}\leq\lambda\left(t^{2}\norm A_{2}+\frac{t^{3}}{5}\norm B_{2}\right).
\]
In particular, this shows that
\[
\Psi(t)=A+Bt+\int_{0}^{t}(t-s)\Phi(s)(A+Bs+E(s))ds
\]
with $\norm{E(s)}_{F}\leq\lambda\left(s^{2}\norm A_{2}+\frac{s^{3}}{5}\norm B_{2}\right)$. 
\end{lem}

\begin{proof}
Recall from (\ref{eq:second_taylor_psi}) that
\begin{align}
\Psi(t) & =A+tB+\int_{0}^{t}(t-s)\Phi(s)\Psi(s)ds.\label{eq:psi_2nd}
\end{align}
Let $E(t)=\Psi(t)-(A+tB)$. Using Lemma \ref{lem:ODE_upper}, we have
that
\begin{align*}
\norm{E(t)}_{F} & =\norm{\int_{0}^{t}(t-s)\Phi(s)\Psi(s)ds}_{F}\\
 & \leq\lambda\int_{0}^{t}(t-s)\norm{\Psi(s)}_{2}ds\\
 & \leq\lambda\int_{0}^{t}(t-s)\norm A_{2}\cosh(\sqrt{\lambda}s)+\frac{\norm B_{2}}{\sqrt{\lambda}}\sinh(\sqrt{\lambda}s)ds\\
 & =\lambda\left(\norm A_{2}(\cosh(\sqrt{\lambda}t)-1)+\frac{\norm B_{2}}{\sqrt{\lambda}}(\sinh(\sqrt{\lambda}t)-\sqrt{\lambda}t)\right).
\end{align*}
Since $0\leq t\leq\frac{1}{\sqrt{\lambda}}$, we have that $\left|\cosh(\sqrt{\lambda}t)-1\right|\leq\lambda t^{2}$
and $\left|\sinh(\sqrt{\lambda}t)-\sqrt{\lambda}t\right|\leq\frac{\lambda^{3/2}t^{3}}{5}$.
This gives the result.

The last equality follows again from (\ref{eq:psi_2nd})
\end{proof}
Next, we have an elementary lemma about the determinant.
\begin{lem}
\label{lem:logdet_est}Suppose that $E$ is a matrix (not necessarily
symmetric) with $\norm E_{2}\leq\frac{1}{4}$, we have
\[
\left|\log\det(I+E)-\tr E\right|\leq\norm E_{F}^{2}.
\]
\end{lem}

\begin{proof}
Let $f(t)=\log\det(I+tE)$. Then, by Jacobi's formula, we have
\begin{eqnarray*}
f'(t) & = & \tr\left((I+tE)^{-1}E\right),\\
f''(t) & = & -\tr((I+tE)^{-1}E(I+tE)^{-1}E).
\end{eqnarray*}
Since $\norm E_{2}\leq\frac{1}{4}$, we have that $\norm{(I+tE)^{-1}}_{2}\leq\frac{4}{3}$
and hence
\begin{eqnarray*}
\left|f''(t)\right| & = & \left|\tr((I+tE)^{-1}E(I+tE)^{-1}E)\right|\\
 & \leq & \left|\tr(E^{T}\left((I+tE)^{-1}\right)^{T}(I+tE)^{-1}E)\right|\\
 & \leq & 2\left|\tr(E^{T}E)\right|=2\norm E_{F}^{2}.
\end{eqnarray*}
The result follows from 
\begin{eqnarray*}
f(1) & = & f(0)+f'(0)+\int_{0}^{1}(1-s)f''(s)ds\\
 & = & \tr(E)+\int_{0}^{1}(1-s)f''(s)ds.
\end{eqnarray*}
\end{proof}

\section{Concentration}
\begin{lem}[{\cite[Ver 3, Lemma 90]{LeeV16}}]
\label{lem:norm_random_Ax}For $p\geq1$, we have 
\[
P_{x\sim N(0,I)}\left(\norm{Ax}_{p}^{p}\leq\left(\left(\frac{2^{p/2}\Gamma(\frac{p+1}{2})}{\sqrt{\pi}}\sum_{i}\norm{a_{i}}_{2}^{p}\right)^{1/p}+\norm A_{2\rightarrow p}t\right)^{p}\right)\leq1-\exp\left(-\frac{t^{2}}{2}\right).
\]
In particular, we have
\[
P_{x\sim N(0,I)}\left(\norm{Ax}_{4}^{4}\leq\left(\left(3\sum_{i}\norm{a_{i}}_{2}^{4}\right)^{1/4}+\norm A_{2\rightarrow4}t\right)^{4}\right)\leq1-\exp\left(-\frac{t^{2}}{2}\right)
\]
and
\[
P_{x\sim N(0,I)}\left(\norm{Ax}_{2}^{2}\leq\left(\left(\sum_{i}\norm{a_{i}}_{2}^{2}\right)^{1/2}+\norm A_{2\rightarrow2}t\right)^{2}\right)\leq1-\exp\left(-\frac{t^{2}}{2}\right).
\]
\end{lem}

\section{\label{sec:Calculus}Calculus}
\begin{proof}[Proof of Fact \ref{fact:calculus}]
 Recall Definition \ref{def:notation} and write

\begin{align*}
\frac{dA_{\gamma}}{dt} & =\frac{dS_{\gamma}^{-1}}{dt}A\\
 & =-S_{\gamma}^{-1}\frac{dS_{\gamma}}{dt}S_{\gamma}^{-1}A\\
 & =-S_{\gamma}^{-1}\Diag\left(\frac{d(A\gamma-b)}{dt}\right)A_{\gamma}\\
 & =-\Diag\left(S_{\gamma}^{-1}A\gamma'\right)A_{\gamma}\\
 & =-\Diag(A_{\gamma}\gamma')A_{\gamma}=-S_{\gamma'}A_{\gamma}.
\end{align*}

For the second, using the first,
\begin{align*}
\frac{dP_{\gamma}}{dt} & =\frac{dA_{\gamma}(A_{\gamma}^{T}A_{\gamma})^{-1}A_{\gamma}^{T}}{dt}\\
 & =\frac{dA_{\gamma}}{dt}(A_{\gamma}^{T}A_{\gamma})^{-1}A_{\gamma}^{T}+A_{\gamma}(A_{\gamma}^{T}A_{\gamma})^{-1}\frac{dA_{\gamma}}{dt}+A_{\gamma}\frac{d(A_{\gamma}^{T}A_{\gamma})^{-1}}{dt}A_{\gamma}^{T}\\
 & =-S_{\gamma'}P_{\gamma}-P_{\gamma}S_{\gamma'}-A_{\gamma}(A_{\gamma}^{T}A_{\gamma})^{-1}\frac{d(A_{\gamma}^{T}A_{\gamma})}{dt}(A_{\gamma}^{T}A_{\gamma})^{-1}A_{\gamma}^{T}\\
 & =-S_{\gamma'}P_{\gamma}-P_{\gamma}S_{\gamma'}+2A_{\gamma}(A_{\gamma}^{T}A_{\gamma})^{-1}\left(A_{\gamma}^{T}S_{\gamma'}A_{\gamma}\right)(A_{\gamma}^{T}A_{\gamma})^{-1}A_{\gamma}^{T}\\
 & =-S_{\gamma'}P_{\gamma}-P_{\gamma}S_{\gamma'}+2P_{\gamma}S_{\gamma'}P_{\gamma}.
\end{align*}

And for the last, 
\begin{align*}
\frac{dS_{\gamma'}}{dt} & =\Diag\left(\frac{dA_{\gamma}}{dt}\gamma'+A_{\gamma}\gamma''\right)\\
 & =\Diag(-S_{\gamma'}A_{\gamma}\gamma'+A_{\gamma}\gamma'')=-S_{\gamma'}^{2}+S_{\gamma''}.
\end{align*}
\end{proof}

\section{Basic definitions of Riemannian geometry \label{sec:RG}}

Here we recall basic notions of Riemannian geometry. One can think
of a manifold $M$ as a $n$-dimensional ``surface'' in $\R^{k}$
for some $k\geq n$. 
\begin{enumerate}
\item Tangent space $T_{p}M$: For any point $p$, the tangent space $T_{p}M$
of $M$ at point $p$ is a linear subspace of $\R^{k}$ of dimension
$n$. Intuitively, $T_{p}M$ is the vector space of possible directions
that are tangential to the manifold at $x$. Equivalently, it can
be thought as the first-order linear approximation of the manifold
$M$ at $p$. For any curve $c$ on $M$, the direction $\frac{d}{dt}c(t)$
is tangent to $M$ and hence lies in $T_{c(t)}M$. When it is clear
from context, we define $c'(t)=\frac{dc}{dt}(t)$. For any open subset
$M$ of $\Rn$, we can identify $T_{p}M$ with $\Rn$ because all
directions can be realized by derivatives of some curves in $\Rn$.
\item Riemannian metric: For any $v,u\in T_{p}M$, the inner product (Riemannian
metric) at $p$ is given by $\left\langle v,u\right\rangle _{p}$
and this allows us to define the norm of a vector $\norm v_{p}=\sqrt{\left\langle v,v\right\rangle _{p}}$.
We call a manifold a Riemannian manifold if it is equipped with a
Riemannian metric. When it is clear from context, we define $\left\langle v,u\right\rangle =\left\langle v,u\right\rangle _{p}$.
In $\Rn$ , $\left\langle v,u\right\rangle _{p}$ is the usual $\ell_{2}$
inner product.
\item Differential (Pushforward) $d$: Given a function $f$ from a manifold
$M$ to a manifold $N$, we define $df(x)$ as the linear map from
$T_{x}M$ to $T_{f(x)}N$ such that
\[
df(x)(c'(0))=(f\circ c)'(0)
\]
for any curve $c$ on $M$ starting at $x=c(0)$. When $M$ and $N$
are Euclidean spaces, $df(x)$ is the Jacobian of $f$ at $x$. We
can think of pushforward as a manifold Jacobian, i.e., the first-order
approximation of a map from a manifold to a manifold.
\item Hessian manifold: We call $M$ a Hessian manifold (induced by $\phi$)
if $M$ is an open subset of $\Rn$ with the Riemannian metric at
any point $p\in M$ defined by
\[
\left\langle v,u\right\rangle _{p}=v^{T}\nabla^{2}\phi(p)u
\]
where $v,u\in T_{p}M$ and $\phi$ is a smooth convex function on
$M$.
\item Length: For any curve $c:[0,1]\rightarrow M$, we define its length
by
\[
L(c)=\int_{0}^{1}\norm{\frac{d}{dt}c(t)}_{c(t)}dt.
\]
\item Distance: For any $x,y\in M$, we define $d(x,y)$ be the infimum
of the lengths of all paths connecting $x$ and $y$. In $\Rn$ ,
$d(x,y)=\norm{x-y}_{2}$.
\item Geodesic: We call a curve $\gamma(t):[a,b]\rightarrow M$ a geodesic
if it satisfies both of the following conditions:

\begin{enumerate}
\item The curve $\gamma(t)$ is parameterized with constant speed. Namely,
$\norm{\frac{d}{dt}\gamma(t)}_{\gamma(t)}$ is constant for $t\in[a,b]$.
\item The curve is the locally shortest length curve between $\gamma(a)$
and $\gamma(b)$. Namely, for any family of curve $c(t,s)$ with $c(t,0)=\gamma(t)$
and $c(a,s)=\gamma(a)$ and $c(b,s)=\gamma(b)$, we have that $\left.\frac{d}{ds}\right|_{s=0}\int_{a}^{b}\norm{\frac{d}{dt}c(t,s)}_{c(t,s)}dt=0$.
\end{enumerate}
Note that, if $\gamma(t)$ is a geodesic, then $\gamma(\alpha t)$
is a geodesic for any $\alpha$. Intuitively, geodesics are local
shortest paths. In $\Rn$, geodesics are straight lines. 
\item Exponential map: The map $\exp_{p}:T_{p}M\rightarrow M$ is defined
as
\[
\exp_{p}(v)=\gamma_{v}(1)
\]
where $\gamma_{v}$ is the unique geodesic starting at $p$ with initial
velocity $\gamma_{v}'(0)$ equal to $v$. The exponential map takes
a straight line $tv\in T_{p}M$ to a geodesic $\gamma_{tv}(1)=\gamma_{v}(t)\in M$.
Note that $\exp_{p}$ maps $v$ and $tv$ to points on the same geodesic.
Intuitively, the exponential map can be thought as point-vector addition
in a manifold. In $\R^{n}$, we have $\exp_{p}(v)=p+v$. 
\item Parallel transport: Given any geodesic $c(t)$ and a vector $v$ such
that $\left\langle v,c'(0)\right\rangle _{c(0)}=0$, we define the
parallel transport $\Gamma$ of $v$ along $c(t)$ by the following
process: Take $h$ to be infinitesimally small and $v_{0}=v$. For
$i=1,2,\cdots,1/h$, we let $v_{ih}$ be the vector orthogonal to
$c'(ih)$ that minimizes the distance on the manifold between $\exp_{c(ih)}(hv_{ih})$
and $\exp_{c((i-1)h)}(hv_{(i-1)h})$. Intuitively, the parallel transport
finds the vectors on the curve such that their end points are closest
to the end points of $v$. For general vector $v\in T_{c'(0)}$, we
write $v=\alpha c'(0)+w$ and we define the parallel transport of
$v$ along $c(t)$ is the sum of $\alpha c'(t)$ and the parallel
transport of $w$ along $c(t)$. For a non-geodesic curve, see the
definition in Fact \ref{fact:basic_RG}. 
\item Orthonormal frame: Given vector fields $v_{1},v_{2},\cdots,v_{n}$
on a subset of $M$, we call $\{v_{i}\}_{i=1}^{n}$ is an orthonormal
frame if $\left\langle v_{i},v_{j}\right\rangle _{x}=\delta_{ij}$
for all $x$. Given a curve $c(t)$ and an orthonormal frame at $c(0)$,
we can extend it on the whole curve by parallel transport and it remains
orthonormal on the whole curve.
\item \label{def:Directional-derivatives}Directional derivatives and the
Levi-Civita connection: For a vector $v\in T_{p}M$ and a vector field
$u$ in a neighborhood of $p$, let $\gamma_{v}$ be the unique geodesic
starting at $p$ with initial velocity $\gamma_{v}'(0)=v$. Define
\[
\nabla_{v}u=\lim_{h\rightarrow0}\frac{u(h)-u(0)}{h}
\]
where $u(h)\in T_{p}M$ is the parallel transport of $u(\gamma(h))$
from $\gamma(h)$ to $\gamma(0)$. Intuitively, Levi-Civita connection
is the directional derivative of $u$ along direction $v$, \emph{taking
the metric into account}. In particular, for $\Rn$, we have $\nabla_{v}u(x)=\frac{d}{dt}u(x+tv)$.
When $u$ is defined on a curve $c$, we define $D_{t}u=\nabla_{c'(t)}u$.
In $\Rn$, we have $D_{t}u(\gamma(t))=\frac{d}{dt}u(\gamma(t))$.
We reserve $\frac{d}{dt}$ for the usual derivative with Euclidean
coordinates. 
\end{enumerate}
We list some basic facts about the definitions introduced above that
are useful for computation and intuition.
\begin{fact}
\label{fact:basic_RG}Given a manifold $M$, a curve $c(t)\in M$,
a vector $v$ and vector fields $u,w$ on $M$, we have the following:

\begin{enumerate}
\item (alternative definition of parallel transport) $v(t)$ is the parallel
transport of $v$ along $c(t)$ if and only if $\nabla_{c'(t)}v(t)=0$.
\item (alternative definition of geodesic) $c$ is a geodesic if and only
if $\nabla_{c'(t)}c'(t)=0$. 
\item (linearity) $\nabla_{v}(u+w)=\nabla_{v}u+\nabla_{v}w$.
\item (product rule) For any scalar-valued function f, $\nabla_{v}(f\cdot u)=\frac{\partial f}{\partial v}u+f\cdot\nabla_{v}u$. 
\item (metric preserving) $\frac{d}{dt}\left\langle u,w\right\rangle _{c(t)}=\left\langle D_{t}u,w\right\rangle _{c(t)}+\left\langle u,D_{t}w\right\rangle _{c(t)}$.
\item (torsion free-ness) For any map $c(t,s)$ from a subset of $\R^{2}$
to $M$, we have that $D_{s}\frac{\partial c}{\partial t}=D_{t}\frac{\partial c}{\partial s}$
where $D_{s}=\nabla_{\frac{\partial c}{\partial s}}$ and $D_{t}=\nabla_{\frac{\partial c}{\partial t}}$.
\item (alternative definition of Levi-Civita connection) $\nabla_{v}u$
is the unique linear mapping from the product of vector and vector
field to vector field that satisfies (3), (4), (5) and (6).
\end{enumerate}
\end{fact}

\subsection{Curvature}

Roughly speaking, curvature measures the amount by which a manifold
deviates from Euclidean space. Given vector $u,v\in T_{p}M$, in this
section, we define $uv$ be the point obtained from moving from $p$
along direction $u$ with distance $\norm u_{p}$ (using geodesic),
then moving along direction ``$v$'' with distance $\norm v_{p}$
where ``$v$'' is the parallel transport of $v$ along the path
$u$. In $\Rn$, $uv$ is exactly $p+u+v$ and hence $uv=vu$, namely,
parallelograms close up. For a manifold, parallelograms almost close
up, namely, $d(uv,vu)=o(\norm u\norm v)$. This property is called
being \emph{torsion-free. }
\begin{enumerate}
\item Riemann curvature tensor: Three-dimensional parallelepipeds might
not close up, and the curvature tensor measures how far they are from
closing up. Given vector $u,v,w\in T_{p}M$, we define $uvw$ as the
point obtained by moving from $uv$ along direction ``$w$'' for
distance $\norm w_{p}$ where ``$w$'' is the parallel transport
of $w$ along the path $uv$. In a manifold, parallelepipeds do not
close up and the Riemann curvature tensor how much $uvw$ deviates
from $vuw$. Formally, for vector fields $v$, $w$, we define $\tau_{v}w$
be the parallel transport of $w$ along the vector field $v$ for
one unit of time. Given vector field $v,w,u$, we define the Riemann
curvature tensor by 
\begin{equation}
R(u,v)w=\left.\frac{d}{ds}\frac{d}{dt}\tau_{su}^{-1}\tau_{tv}^{-1}\tau_{su}\tau_{tv}w\right|_{t,s=0}.\label{eq:Ruvw_explain}
\end{equation}
Riemann curvature tensor is a tensor, namely, $R(u,v)w$ at point
$p$ depends only on $u(p)$, $v(p)$ and $w(p)$. 
\item Ricci curvature: Given a vector $v\in T_{p}M$, the Ricci curvature
$\text{Ric}(v)$ measures if the geodesics starting around $p$ in
direction $v$ converge together. Positive Ricci curvature indicates
the geodesics converge while negative curvature indicates they diverge.
Let $S(0)$ be a small shape around $p$ and $S(t)$ be the set of
point obtained by moving $S(0)$ along geodesics in the direction
$v$ for $t$ units of time. Then, 
\begin{equation}
\text{vol}S(t)=\text{vol}S(0)(1-\frac{t^{2}}{2}\text{Ric}(v)+\text{smaller terms}).\label{eq:Ric_explain}
\end{equation}
Formally, we define
\[
\text{Ric}(v)=\sum_{u_{i}}\left\langle R(v,u_{i})u_{i},v\right\rangle 
\]
where $u_{i}$ is an orthonormal basis of $T_{p}M$. Equivalently,
we have $\text{Ric}(v)=\E_{u\sim N(0,I)}\left\langle R(v,u)u,v\right\rangle $.
For $\R^{n}$, $\text{Ric}(v)=0$. For a sphere in $n+1$ dimension
with radius $r$, $\text{Ric}(v)=\frac{n-1}{r^{2}}\norm v^{2}$.
\end{enumerate}
\begin{fact}[Alternative definition of Riemann curvature tensor]
\label{fact:formula_R}Given any $M$-valued function $c(t,s)$,
we have vector fields $\frac{\partial c}{\partial t}$ and $\frac{\partial c}{\partial s}$
on $M$. Then, for any vector field $z$, 
\[
R(\frac{\partial c}{\partial t},\frac{\partial c}{\partial s})z=\nabla_{\frac{\partial c}{\partial t}}\nabla_{\frac{\partial c}{\partial s}}z-\nabla_{\frac{\partial c}{\partial s}}\nabla_{\frac{\partial c}{\partial t}}z.
\]
Equivalently, we write $R(\partial_{t}c,\partial_{s}c)z=D_{t}D_{s}z-D_{s}D_{t}z$.
\end{fact}

\begin{fact}
\label{fact:sym_cuv}Given vector fields $v,u,w,z$ on $M$,
\[
\left\langle R(v,u)w,z\right\rangle =\left\langle R(w,z)v,u\right\rangle =-\left\langle R(u,v)w,z\right\rangle =-\left\langle R(v,u)z,w\right\rangle .
\]
\end{fact}

\subsection{Hessian manifolds}

Recall that a manifold is called Hessian if it is a subset of $\Rn$
and its metric is given by $g_{ij}=\frac{\partial^{2}}{\partial x^{i}\partial x^{j}}\phi$
for some smooth convex function $\phi$. We let $g^{ij}$ be entries
of the inverse matrix of $g_{ij}$. For example, we have $\sum_{j}g^{ij}g_{jk}=\delta_{ik}$.
We use $\phi_{ij}$ to denote $\frac{\partial^{2}}{\partial x^{i}\partial x^{j}}\phi$
and $\phi_{ijk}$ to denote $\frac{\partial^{3}}{\partial x^{i}\partial x^{j}\partial x^{k}}\phi$. 

Since a Hessian manifold is a subset of Euclidean space, we identify
tangent spaces $T_{p}M$ by Euclidean coordinates. The following lemma
gives formulas for the Levi-Civita connection and curvature under
Euclidean coordinates. 
\begin{lem}[\cite{totaro2004curvature}]
\label{lem:Hessian_formula}Given a Hessian manifold $M$, vector
fields $v,u,w,z$ on $M$, we have the following:

\begin{enumerate}
\item (Levi-Civita connection) $\nabla_{v}u=\sum_{ik}v_{i}\frac{\partial u_{k}}{\partial x_{i}}e_{k}+\sum_{ijk}v_{i}u_{j}\Gamma_{ij}^{k}e_{k}$
where $e_{k}$ are coordinate vectors and the Christoffel symbol 
\[
\Gamma_{ij}^{k}=\frac{1}{2}\sum_{l}g^{kl}\phi_{ijl}.
\]
\item (Riemann curvature tensor) $\left\langle R(u,v)w,z\right\rangle =\sum_{ijlk}R_{klij}u_{i}v_{j}w_{l}z_{k}$
where 
\[
R_{klij}=\frac{1}{4}\sum_{pq}g^{pq}\left(\phi_{jkp}\phi_{ilq}-\phi_{ikp}\phi_{jlq}\right).
\]
\item (Ricci curvature) $Ric(v)=\frac{1}{4}\sum_{ijlkpq}g^{pq}g^{jl}\left(\phi_{jkp}\phi_{ilq}-\phi_{ikp}\phi_{jlq}\right)v_{i}v_{k}$.
\end{enumerate}
\end{lem}

\end{document}